\documentclass[10pt, journal]{IEEEtran}

\usepackage{graphicx}
\usepackage{booktabs} 
\usepackage{makecell}
\usepackage{longtable}
\usepackage{wrapfig}
\usepackage{hyperref}
\usepackage{amsmath}
\usepackage{amssymb}
\usepackage{capt-of}
\usepackage{comment}
\usepackage{dsfont}
\usepackage{url}

\usepackage{enumitem}
\usepackage{bm}

\usepackage{amsmath}
\usepackage{amssymb}

\usepackage{mathtools}
\usepackage{amsthm}

\usepackage{xcolor}
\usepackage{bm}
\usepackage{subcaption}
\usepackage{algorithm}
\usepackage{algorithmic}

\newtheorem{theorem}{Theorem}[section]
\newtheorem{proposition}{Proposition}[section]
\newtheorem{lemma}{Lemma}[section]

\theoremstyle{definition}
\newtheorem{definition}{Definition}

\theoremstyle{remark}
\newtheorem{remark}{Remark}[section]

\newcommand{\pre}{{\text{pre}}}
\newcommand{\post}{{\text{post}}}

\newcommand{\bA}{{\mathbf{A}}}
\newcommand{\bS}{{\mathbf{S}}}
\newcommand{\bD}{{\mathbf{D}}}
\newcommand{\ba}{{\mathbf{a}}}
\newcommand{\calX}{{\mathcal{X}}}
\newcommand{\calF}{{\mathcal{F}}}
\newcommand{\bE}{{\mathbb{E}}}
\newcommand{\WADD}{{\mathsf{WADD}}}
\newcommand{\CBM}{{\mathsf{CBM}}}

\newcommand{\KL}{{\mathsf{KL}}}
\newcommand{\TV}{{\mathsf{TV}}}
\DeclareMathOperator*{\esssup}{ess\,sup}

\usepackage{authblk}

\title{Differentially Private Online Community Detection for Censored Block Models: Algorithms and Fundamental Limits}

\author[1]{Mohamed~Seif\thanks{
 This work was supported by the AFOSR award \#002484665, the U.S. National Science Foundation under Grant CNS-2147631, and an Innovation Grant from Princeton NextG.}}

\author[2]{Liyan~Xie}
\author[1]{Andrea J. Goldsmith}
\author[1]{H. Vincent Poor}
\affil[1]{\small
Department of Electrical and Computer Engineering, Princeton University}

\affil[2]{\small
Department of Industrial and Systems Engineering, University of Minnesota}

\date{}

\begin{document}

\markboth{Published in IEEE Transactions on Information Forensics and Security}{Seif \MakeLowercase{\textit{et al.}}: Differentially Private Online Community Detection for Censored Block Models}


\maketitle

\begin{abstract}
We study the private online change detection problem for dynamic communities, using a censored block model (CBM). We consider edge differential privacy (DP) in both local and central settings, and propose joint change detection and community estimation procedures for both scenarios. We seek to understand the fundamental tradeoffs between the privacy budget, detection delay, and exact community recovery of community labels. Further, we provide theoretical guarantees for the effectiveness of our proposed method by showing necessary and sufficient conditions for change detection and exact recovery under edge DP.  Simulation and real data examples are provided to validate the proposed methods. 
\end{abstract}


\section{Introduction}
\label{introduction}

Community detection in networks is a fundamental problem in the area of graph mining and unsupervised learning \cite{fortunato2010community}. The main goal is to find partitions (or communities) in the graph that have dense connections within communities and sparse connections across communities. The Stochastic Block Model (SBM) is commonly used to model the graphical structure for networks \cite{abbe2018community, hajek2016achieving}. In an SBM model, nodes are divided into communities, and the connection probability between nodes depends on whether they belong to the same community or not.
Connections between nodes are more probable within the same community than between communities. This difference in connection probabilities is central to the challenge of community detection. The study and improvement of community detection methods using the SBM framework has been a dynamic area of research, with many advances and findings detailed in comprehensive surveys such as the one by Abbe et al. \cite{abbe2017community}. 


In contemporary research on network community detection, a significant focus has been placed on static network structures. 
However, this focus on static networks does not completely capture the {\it dynamic} nature of real-world networks, where the arrangement of nodes and their links may constantly change \cite{spectralcusum}. Despite some advances (e.g., \cite{barucca2018disentangling,xu2014dynamic}) in applying the dynamic SBM to these scenarios, a comprehensive understanding of the conditions that allow for accurate community detection in such dynamic environments, particularly within the framework of the dynamic Censored Block Model (CBM), remains an open and critical area of research.

The study of CBMs in dynamical networks has become increasingly relevant in network science, offering critical insights into the dynamics of complex systems \cite{hajek2015exact,moghaddam2022exact}. CBM models effectively represent real-world networks characterized by block structures and are particularly valuable in analyzing networks with {\it incomplete} or {\it censored} data. This feature is essential for understanding various network dynamics, such as information spread and epidemic outbreaks. 
CBM models are also instrumental in the development and testing of network science algorithms, especially for community detection and network reconstruction. 

The information-theoretic conditions for recoverability are well understood in terms of the scaling of the CBM parameters $n$, $p$ and  $\zeta \in (0, 1/2)$, where $n$ denotes the number of nodes, $p$ represents the probability of an edge being observed, and $\zeta$ denotes the connection probability (these quantities will be defined more precisely in the sequel). In particular, in the {\it dense} regime (which is the focus of this paper), with $p = a \log(n)/n$ for some constant $a > 0$, it is known that exact recovery under binary  CBMs is possible \emph{if and only if}  $a (\sqrt{1 - \zeta} - \sqrt{\zeta})^{2} > 1$ \cite{hajek2016achieving_extensions}. This condition is considered to be a \textit{sharp} threshold for recovery, which can also be interpreted as a phase transition.
In the literature, 
efficient algorithms for recovering communities have been developed using spectral methods (e.g., \cite{dhara2022power,dhara2022spectral}) and semidefinite programming (SDP)~\cite{hajek2016achieving}. 
Some of these results (e.g., \cite{dhara2022power, hajek2016achieving}) were proven to achieve the optimal threshold for exact recovery.



In this work, we study the dynamic CBM models under privacy constraints, since sensitive information is frequently present in network data. As such, ensuring the privacy of individuals during data analysis is imperative. Differential Privacy (DP) \cite{dwork2014algorithmic} has become the recognized standard for delivering robust privacy assurances. DP guarantees that any single user's participation in a dataset only minutely impacts the statistical outcomes of queries. 
In the literature \cite{karwa2011private}, two main privacy notions have been proposed for data analysis over graphs: $(1)$ edge DP that aims to protect the individual relationships (edges) in a graph, 
and $(2)$ node DP that focuses of the privacy of the nodes and their corresponding interactions (edges). Notably, DP algorithms have been tailored to specific network analysis problems, including counting stars, triangles, cuts, dense subgraphs, and communities, as well as the generation of synthetic graphs \cite{blocki:itcs13,imola2021locally,nguyen2016detecting, qin2017generating}. 

Our privacy threat model incorporates both \textit{local} and \textit{central} settings. We highlight that the central DP is a weaker guarantee compared to the local DP (as it can be readily shown that the local DP guarantee implies central DP). The rationale for considering both local and central DP simultaneously is as follows. In scenarios where a trusted curator/server (e.g., Facebook) is initially considered trustworthy, users could send their raw (unperturbed) local connections to the server. However, if the raw connections received by the server were leaked (e.g., via a data breach), then having an additional local DP guarantee can provide a second layer of protection. Furthermore, the local privacy setting has a worse utility guarantee due to the users' perturbation of distributed data.

Under edge DP constraints, we study the online change detection of community structures for CBMs. We assume an online sequence of graph observations generated from CBMs, where the underlying community structure changes at an unknown change-point. Our goal is to detect the change
as quickly as possible, subject to false alarm constraints {and privacy requirements} \cite{poor2008quickest}. {Specifically, we aim to identify potential changes in the underlying community structure of CBMs based on streaming graph observations, while simultaneously ensuring the privacy of each individual graph.} We propose two types of private online detection algorithms under local DP and central DP settings, respectively. Sufficient conditions for exact recovery of the unknown community structure, under DP constraints, are also established. Furthermore, we provide theoretical guarantees for the detection performance of proposed methods.

\paragraph{Related Works}

The problem of private change detection is an evolving field with significant applications. {It has been studied under $(\epsilon,\delta)$ differential privacy constraints, assuming parametric distributions \cite{canonne2019structure,cummings2018differentially}. Both works assume that the pre- and post-change distributions are fully known in advance. The change detection for unknown post-change distributions was considered in \cite{cummings2020privately}, but the analysis was limited to univariate distributions, and thus is not directly applicable to the graph-based observations considered in our setting.
In \cite{berrett2021locally}, online change detection for multivariate nonparametric regression was studied under local differential privacy. The proposed privacy mechanism relies on discretizing the continuous space into a set of cubes, whose number grows exponentially with the data dimension. This approach is not directly applicable to graph observations and poses computational challenges in high-dimensional settings. In \cite{li2022network}, the authors extended the analysis of privatized networks from static to dynamic, underscoring the complexities and challenges of preserving privacy in the dynamic analysis of network data. 
It is worthwhile mentioning that they considered the offline change-point estimation problem with a fixed number of graph observations. In contrast, we study the {\it online} setting, where graphs arrive in a streaming fashion and the objective is to detect changes as quickly as possible. The online setting differs fundamentally from the offline one in several aspects, including its intrinsic challenges (since online methods can only use past data) and distinct performance metrics.}

{{In the context of privacy‑preserving community detection on random graphs, Hehir et al. \cite{hehir2021consistency} analyzed a simple spectral approach for multi‑community settings, extending the convergence‑rate analysis of Lei \& Rinaldo’s algorithm \cite{lei2015consistency} and quantifying how privacy parameters influence misclassification between true and estimated labels. In our own work \cite{mohamed2022differentially}, we characterized the information‑theoretic trade‑offs of three private community‑recovery mechanisms—(a) stability‑based, (b) sampling‑based, and (c) graph‑perturbation. We showed that stability‑ and sampling‑based approaches yield a more favorable balance between the connection probabilities (i.e., random graph model parameters), accuracy and privacy budget $(\epsilon,\delta)$, albeit at higher computational cost, whereas graph‑perturbation, though cheaper, requires $\epsilon = \Omega(\log n)$ to guarantee exact recovery.}}

In contrast, this paper broadens the scope of analysis by considering possible network changes, in particular, shifts in community memberships. This approach presents a more comprehensive and challenging problem, as it involves not only the detection of statistical changes but also the identification of structural and compositional alterations within the network dynamics. Our methodology, therefore, addresses a broader range of changes, offering deeper insights into the complexities of dynamic network analysis.

Our work is also related to the community recovery problem over censored block models. The semidefinite relaxation based method is a widely used approach to estimate the underlying community structures from adjacency matrices. The exact recovery (necessary and sufficient) conditions have been studied in \cite{hajek2016achieving, hajek2016achieving_extensions}. In our work, we assume that the community after the change is \textit{unknown}, and thus a joint estimation and detection scheme is utilized, in which the semidefinite relaxation is used to estimate the community structure in a sequential manner. Furthermore, we derive new information-theoretic conditions for exact recovery under edge DP constraints.

\paragraph{Notation}
We use boldface uppercase letters to denote matrices (e.g., $\bA$) and
boldface lowercase letters for vectors (e.g., $\ba$). We denote scalars by non-boldface lowercase letters (e.g., $x$), and sets by capital calligraphic letters (e.g., $\calX$). Let $[n] \triangleq \{1, 2, \cdots, n\}$ represent the set of all integers from $1$ to $n$.  We use $\operatorname{Lap}(\beta)$ to denote the Laplace distribution with zero mean and scale $\beta$, and the probability density function is $\operatorname{Lap}(x|\beta)=\frac{1}{2\beta}e^{-{|x|}/{\beta}}$.
For asymptotic analysis, we write the function $f(n) = o(g(n))$ when $\lim_{n \rightarrow \infty} f(n)/g(n) = 0$.  Also, $f(n) = \mathcal{O}(g(n))$ means there exist some constant $C > 0$ such that $|f(n)/g(n)| \leq C$, $\forall n$, and $f(n) = \Omega(g(n))$ means there exists some constant $c > 0$ such that $|f(n)/g(n)| \geq c, \forall n$. A summary of the key symbols is provided in Table \ref{table:symbols_table}.

  \begin{figure*}[t]
\centering
	\centering
	{\includegraphics[width=1.4\columnwidth]{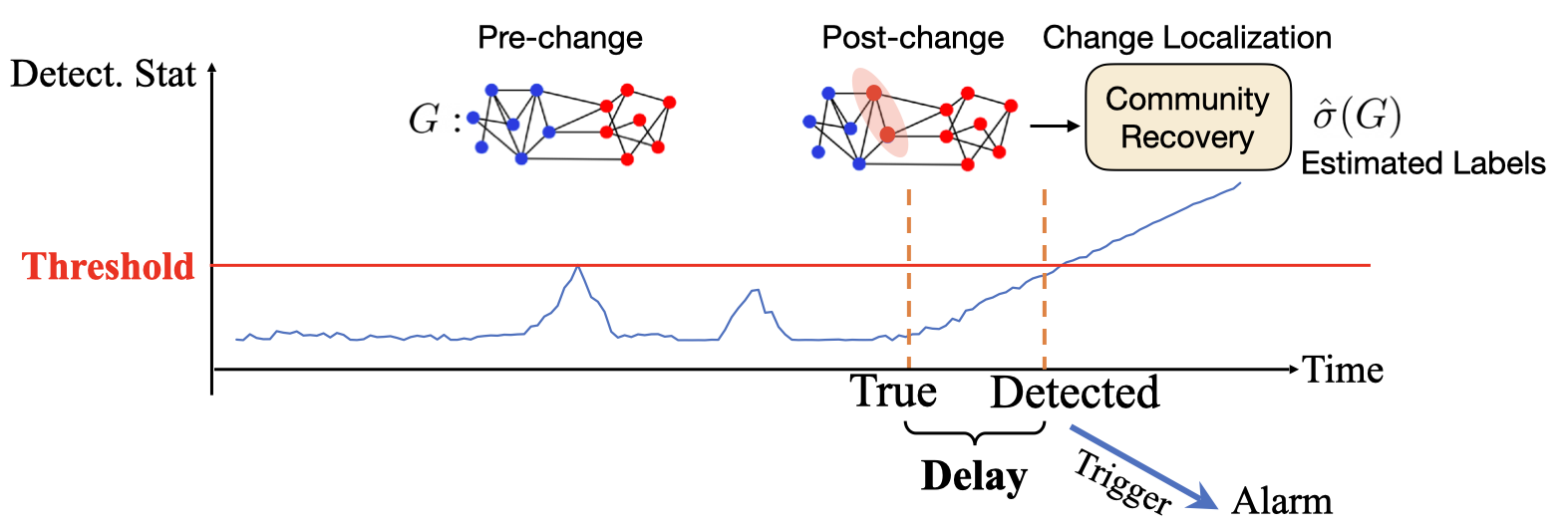}
    \caption{\small{Overview of the proposed change‑detection scheme. We develop online community‑detection frameworks that detect changes in community memberships via CUSUM-based tests. The framework also incorporates an SDP-based community recovery component to estimate the unknown post-change communities. Both the estimation and detection components satisfy edge differential privacy to protect individual node interactions. }}
    \label{fig:problem_statement_clear}}
\end{figure*}

\begin{table}[t]
\centering
\caption{\small{List of symbols.}}
\label{table:symbols_table}
\begin{tabular}{l l}
\hline
$n$ & Number of vertices\\
$\mathbf{A}$ & Adjacency matrix\\ 
$\epsilon$ & Privacy budget\\ 
$\delta$ & Failure probability \\ 
$\bm{\sigma}^{*}$ & Ground truth labels \\
$\hat{\bm{\sigma}}$ & Estimated labels \\
$\mathbf{I}$ & Identity matrix \\ 
$\mathbf{J}$ & All ones matrix\\
$\mathbf{1}$ & All ones vector\\
\hline
\end{tabular}
\end{table}

\section{Problem Setup} 
\label{sec:preliminaries_and_problem_statement}


We begin by summarizing the problem statement in Fig.~\ref{fig:problem_statement_clear}.

\subsection{Data Model}\label{sec:problem}

We observe a sequence of undirected graphs $\{G_t,t\in\mathbb{N}\}$, each generated via a CBM. Each graph $G = ({V}, E)$ consists of $n$ vertices divided into two communities\footnote{Extending our approach to encompass multiple communities is relatively straightforward, since multi-community detection can also be reformulated as an SDP problem, with analogous recovery conditions anticipated, as demonstrated in Section 5 of \cite{hajek2016achieving_extensions}.}, with community label $\bm{\sigma}^{*} = \{{\sigma}^{*}_{1}, {\sigma}^{*}_{2}, \cdots, {\sigma}^{*}_{n}\}$, $\sigma^{*}_{i} \in \{-1, +1\}$, $\forall i \in [n]$. 
The connections between vertices are represented by a weighted adjacency matrix $\mathbf{A}$, where each edge $A_{i,j} \in \{+1, -1, 0\}$, $i<j$, is drawn as follows,
\begin{equation}\label{eq:model}
 \Pr(A_{i,j} = a) =
\begin{cases}
p (1-\zeta) & \text{if } a = \sigma^{*}_{i} \cdot \sigma^{*}_{j}, \\
p \zeta & \text{if } a = -\sigma^{*}_{i} \cdot \sigma^{*}_{j}, \\
1-p & \text{if } a = 0,
\end{cases} 
\end{equation}
where $p$ represents the revealed probability of an edge and \( \zeta \in (0, \frac{1}{2})\). We set \( A_{i,i} = 0 \) for all \( i \) and \( A_{i,j} = A_{j,i} \) for all \( i > j \). This model is denoted as  $\CBM(\bm{\sigma}^{*},p,\zeta)$. 

\begin{remark}[Estimator of $\bm{\sigma}^{*}$]
 Upon observing the adjacency matrix $\bA$, we may derive the maximum likelihood (ML) estimator of $\bm{\sigma}^{*}$ as (see Appendix \ref{sec:llr} for details): 
\begin{align}
\hat{\bm{\sigma}}_{\text{ML}} = \arg \max_{\bm{\sigma}} \{c_{1} \bm{\sigma}^{T}\mathbf{A} \bm{\sigma} + c_{0}: \sigma_{i} = \pm 1\}, \label{eqn:optimal_ML_estimator}
\end{align}
 where $c_{1} = 0.25 \log((1-\zeta)/\zeta)$, and $c_{0} \geq 0$ is a deterministic scalar that does not depend on the labeling. 
It is worth mentioning that \eqref{eqn:optimal_ML_estimator} is equivalent to solving the \textit{max-cut} problem in the graph, which is known to be NP-hard \cite{hajek2016achieving_extensions}.
Subsequently, we will adopt the commonly used semidefinite relaxation of the ML estimator which can be solved in polynomial time and with theoretical guarantees \cite{hajek2016achieving_extensions}. Define $\mathbf{Y} = {\bm{\sigma}} {\bm{\sigma}}^{T}$, where $Y_{i,i} = 1, \forall i \in [n]$, then the relaxed ML estimator is
\begin{align}
    \hat{\mathbf{Y}} = \arg\max_{\mathbf{Y}} &  \hspace{0.1in}\text{tr}({\mathbf{A}} \mathbf{Y})    \quad \text{s.t.}  \hspace{0.1in}  \mathbf{Y} \succcurlyeq \mathbf{0} , \, Y_{i, i} = 1, \forall i \in [n]. \label{eqn:SDP_relaxation_asymmetric} 
\end{align}
\end{remark}

We now consider a sequence of graphs $\{G_t,t\in\mathbb{N}\}$ with $G_t\sim\CBM(\bm{\sigma}_t^{*},p,\zeta)$, and let $\mathbf{A}_t$ be the adjacency matrix of $G_t$. We assume there is an abrupt change to the community label at time $\nu$, which is called a {\it change-point} and is {\it unknown}. The problem can be formulated as follows,
\begin{equation}
	\label{eq:model1}
 \bm{\sigma}_t^{*}=\begin{cases}
  \bm{\sigma}^{\text{pre}}, &  ~~~t=1,\dots,\nu-1,\\
  \bm{\sigma}^{\text{post}}, & ~~~ t = \nu,\tau+1,\dots\\
 \end{cases}
\end{equation}
where $\bm{\sigma}^\pre$ and $\bm{\sigma}^\post$ represent the pre- and post-change communities, respectively. We further assume $\bm{\sigma}^\pre$ is known since it can usually be estimated from historical data under the pre-change regime. The post-change community $\bm{\sigma}^\post$ is assumed to be unknown.
We measure the strength of change using the Hamming distance between $\bm\sigma^\pre$ and $\bm\sigma^\post$,
\begin{equation*}\label{eq:hamming}
 \operatorname{Ham}({\bm{\sigma}^{\pre}},{\bm{\sigma}^{\post}}):=\sum\nolimits_{i=1}^n \mathbb{I}\{\sigma_i^{\pre} \neq \sigma_i^{\post}\}.
\end{equation*}
We note that  $\operatorname{Ham}({\bm{\sigma}^{\pre}},{\bm{\sigma}^{\post}}) \le n/2$, which holds without loss of generality since larger distances can be reduced by flipping community labels. 

Our goal is to design a detection algorithm that is able to protect the user's privacy while detecting the change as quickly as possible from streaming data. The detection is performed through a {\it stopping time} $\tau$ on the observation sequence at which the change is declared \cite{poor2008quickest}. 
To measure the performance of a stopping time, we introduce two metrics commonly used in online detection problems: detection delay and false alarm performances. Let $\Pr_\nu$ denote the probability measure on the observation sequence when the change-point is $\nu$, and let $\mathbb E_\nu$ denote the corresponding expectation.  Also, $\Pr_\infty$ and $\mathbb E_\infty$ denote the probability and expectation operator when there is no change (i.e., $\nu=\infty$). 
The commonly adopted definition of {\it worst-case average detection delay (WADD)} is \cite{lorden1971}: 
\begin{equation}\label{eq:WADDdef}
\WADD(\tau) =  \underset{\nu \geq 1}\sup \esssup \ \mathbb E_\nu\left[(\tau-\nu+1)^+| \mathcal F_{\nu-1}\right],
\end{equation}
where $\{\mathcal F_k,\ k\in \mathbb{N} \}$ is the filtration $\mathcal F_k = \sigma(G_1, \ldots, G_k)$, with $\mathcal F_0$ denoting the trivial sigma algebra. We measure the false alarm performance of a detector (stopping time) $\tau$ in terms of its {\it Average Run Length (ARL)}, defined as $\mathbb{E}_\infty\left[\tau\right]$. The larger the ARL, the smaller the probability of false alarms. We aim to design a detection method that satisfies the desired requirement on ARL while having a small WADD, under the private setting to be detailed next.

\subsection{Differential Privacy for Graphs}\label{sec:dp}

In addition to the dynamic community detection, we require that the algorithm also protect the individual's data privacy. 
In the context of graph data, our threat model is structured around the principle of {\it edge differential privacy} (DP), incorporating both \textit{local} and \textit{central} settings.

Local DP (LDP) offers robust privacy guarantees during the data collection process, ensuring that individual data contributions are privatized at the {\it source}. This is crucial in scenarios where trust in the data curator is limited. Similar to \cite{li2022network}, we consider the privacy mechanism which privatizes each edge $A_{i,j}$ to $\tilde A_{i,j}$ following a conditional distribution $Q$ as defined below.

\begin{definition}[$\epsilon$-edge LDP] \label{def:edgeLDP}
We say the privacy mechanism $\tilde A_{i,j} \sim Q(\cdot|A_{i,j})$, with $\tilde \bA$ being the privatized and $\bA$ the original adjacency matrix, is $\epsilon$-edge LDP, if for all $i<j$ and all $a,a',\tilde{a}\in\{-1,+1,0\}$, it holds that
\[
{Q(\tilde A_{i,j} = \tilde{a} |A_{i,j}=a)} \leq e^\epsilon {Q(\tilde A_{i,j} = \tilde{a} |A_{i,j}=a')}.
\]
\end{definition}

Conversely, in the central setting, privacy protections are applied to the aggregated data, focusing on preserving the presence or absence of single edges. We present the notion of $(\epsilon, \delta)$-edge central DP (CDP) \cite{karwa2011private} below. Here, $\epsilon$ is coined as privacy budget and $\delta$ represents failure probability. It is worth noting that smaller $\epsilon$, $\delta$ correspond to stronger privacy guarantee and vice versa. The setting when $\delta = 0$ is referred to as pure $\epsilon$-edge CDP.
 \begin{definition} [$(\epsilon, \delta)$-edge CDP] \label{def:edgeDP}  A (randomized) community estimator $\hat{\bm{\sigma}}$ as a function of $\bA$ satisfies $(\epsilon, \delta)$-edge CDP for some $\epsilon \in \mathds{R}^{+}$, if for all pairs of adjacency matrices $\mathbf{A}$ and $\mathbf{A}'$ that differ in {\it one} edge, and any measurable subset $\mathcal{S} \subseteq \operatorname{Range}(\hat{\bm{\sigma}})$, we have 
\begin{align*}\label{eq:edgeDP}
    {\operatorname{Pr}(\hat{\bm{\sigma}}(\mathbf{A}) \in \mathcal{S}}) \leq e^{\epsilon} { \operatorname{Pr}(\hat{\bm{\sigma}}(\mathbf{A}') \in \mathcal{S}}) + \delta,
\end{align*}
where the probabilities are computed only over the randomness in the estimation process. 
\end{definition}

\section{Proposed Method: Privatized  
 Adaptive \\ CUSUM Test} \label{sec:detect_algo}

In this section, we introduce the proposed joint estimation and detection scheme, in which we perform sequential estimates for the unknown post-change community structure $\bm{\sigma}^{\text{post}}$ and use it to calculate the detection statistic. We present two types of solution schemes for the LDP and CDP settings, respectively. Under the LDP setting, we first introduce the graph perturbation mechanism that outputs privatized data through a conditional distribution, and then introduce the Adaptive CUSUM type detection statistics utilizing the estimates of unknown post-change community $\bm{\sigma}^{\text{post}}$. Under the CDP setting, we first compute private estimates of $\bm{\sigma}^{\text{post}}$ using the original (non-private) data, and then calculate the detection statistics with added Laplace noise.

To obviate redundancy, we shall describe the first mechanism, graph perturbation mechanism, in full detail while treating the subsequent mechanism with conciseness. Our emphasis will be directed towards elucidating the main algorithmic distinctions, and their respective recovery conditions will be deliberated upon in subsequent sections.  We summarize our theoretical results in Table \ref{table:CBM_summary_results_approaches}, which shows the tradeoffs between $(a, \zeta)$, $(\epsilon, \delta)$ as well as the computational complexity of the mechanisms.

\begin{table*}[t]
\caption{Summary of the recovery threshold(s), detection delay(s), complexity, and privacy.} \label{table:CBM_summary_results_approaches}
    \centering
    \begin{tabular}{| c | c | c |}
    \hline
   Mechanism & Graph Perturbation & Perturbation Stability \\ 
    \hline \hline
    \small{$a(\sqrt{1-\zeta} - \sqrt{\zeta})^{2} \geq$} & Theorem \ref{thm:private_threshold_condition_one_time_instance} & Theorem \ref{thm:perturbation_stability} \\ 
    \hline
     Detection Delay & Theorem \ref{th:upper_wadd} & Theorem \ref{thm:wadd2} \\
    \hline
      Time complexity & \small{$\mathcal{O}(\operatorname{poly}(n))$} & \small{$n^{(\mathcal{O}(\log{(n)}))}$} \\
    \hline
      $\epsilon$ & \small{$\Omega(\log(n))$} & \small{$\mathcal{O}(1)$} \\ 
    \hline
    $\delta$ & 0 & $1/n^{2}$ \\
    \hline
    \end{tabular}
\end{table*}





\subsection{Detection Procedure under the LDP Setting}\label{sec:ldp-method}

\paragraph{Privacy Mechanism (1): Graph Perturbation Mechanism}\label{sec:private}


  \begin{figure}[t]
\centering
	\centering
	{\includegraphics[width=\columnwidth]{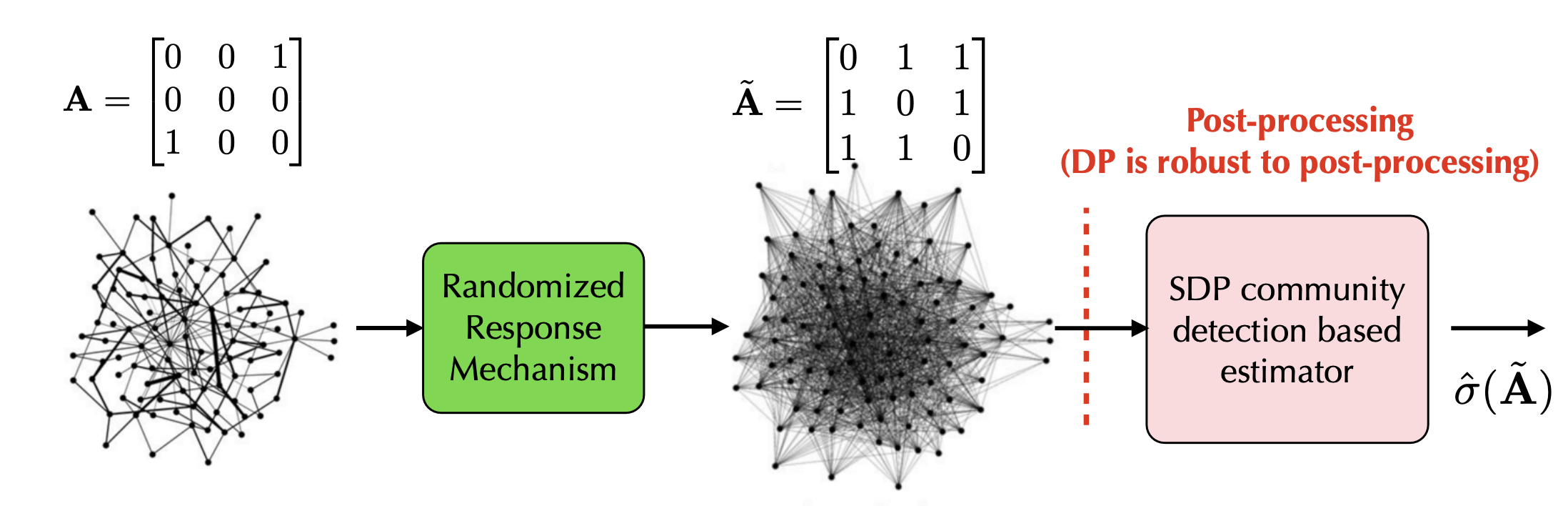}}	
	\caption{\small{Illustration of the graph perturbation-based mechanism: The original edges are first randomized via a randomized response scheme, and the resulting noisy adjacency matrix is then supplied to the community recovery algorithm.}}
    \label{fig:graph_perturbation_illustration}
\end{figure}

Under the LDP requirement, we present the graph perturbation mechanism that we adopted for private inference, as illustrated in Fig.~\ref{fig:graph_perturbation_illustration}. The basic idea is to perturb each edge of the random graph (i.e., the adjacency matrix \( \mathbf{A} \)) independently to obtain a perturbed adjacency matrix \( \tilde{\mathbf{A}} \) to satisfy \( \epsilon \)-edge LDP. 
More specifically, when \( A_{i,j} = x \) where \( x \in \{0, +1, -1\} \), the mechanism outputs \( \tilde{A}_{i,j} \) that follows the distribution \( Q_x(\cdot)\). We let $Q_x(y)=c_2 = \frac{1}{e^{\epsilon} + 2}$ when $x\neq y$, and $Q_x(y)=c_1 = 1-2c_2= \frac{e^{\epsilon}}{e^{\epsilon} + 2}$ when $x=y$. 
Note that for the original data distribution $\CBM(\bm{\sigma}^{*},p,\zeta)$, the perturbed elements in the adjacency matrix $\tilde{\mathbf{A}}$ are distributed as $\CBM(\bm{\sigma}^{*},\tilde p,\tilde\zeta)$ with (see detailed derivations in Appendix \ref{sec:pert_calc}), 
\begin{equation}\label{eqn:perturbed_parameters}
\tilde p = \frac{2 + p(e^\epsilon-1)}{e^\epsilon+2}, \ \tilde\zeta =  \frac{1+p\zeta(e^\epsilon-1)}{2 + p(e^\epsilon-1)}. 
\end{equation}
Also, the randomization mechanism is structured in such a way that \( \tilde{A}_{i,j} = \tilde{A}_{j,i} \) and \( \tilde{A}_{i,i} = 0 \). It is obvious that the above construction {\it satisfies  $\epsilon$-edge LDP} in Definition \ref{def:edgeLDP}.  

\paragraph{Detection Procedure} We then apply the detection procedure on the privatized data $\{\tilde \bA_t, t\in\mathbb N\}$. We tailor the framework of the Adaptive CUSUM test \cite{wlcusum2023} to our community detection problem. 
At each time $t$, we use the latest community structure estimate $\bm{\hat\sigma}_{t-1}$ obtained from the recent private sample $\tilde \bA_{t-1}$\footnote{\label{myfootnote}We note that it is straightforward to replace the estimator $\bm{\hat\sigma}_{t-1}$ as the one obtained from the most recent $w>1$ private samples $\{\tilde \bA_{t-w},\ldots, \tilde \bA_{t-1}\}$, which coincides with the approximate maximum likelihood solution to a multi-view stochastic block model \cite{cohen2024multi,zhang2024community}.}. We adopt the ML estimator obtained from semidefinite relaxation as in \eqref{eqn:SDP_relaxation_asymmetric}. Then, we calculate the detection statistics recursively via
\begin{equation}\label{eq:stat-wlcusum}
  S_t = (S_{t-1})^{+} + \log\frac{\Pr(\tilde\bA_t;\bm{\hat\sigma}_{t-1})}{\Pr(\tilde\bA_t;\bm{\sigma}^{\text{pre}})}, t \geq 1; \quad S_0=0,
\end{equation}
where $(x)^{+}:=\max\{x,0\}$ and  $\Pr(\cdot;\bm\sigma)$ denotes the probability mass function of $\CBM(\bm\sigma,p,\zeta)$. Here we assume the distribution parameters $p$ and $\zeta$ are {\it known}.
It is worthwhile mentioning that, by definition, the estimate $\bm{\hat\sigma}_{t-1}$ is {\it independent} with the current sample $\tilde \bA_t$. Such independence is essential for subsequent theoretical analysis. 

Given the detection statistic $S_t$, the detection procedure is conducted via a stopping rule, 
\begin{equation}\label{eq:stop_time}
   T_{\rm L}:=\inf\{t: S_t \geq b\}, 
\end{equation}
where $b$ is a pre-specified detection threshold chosen to meet the desired false alarm constraint. 
We will discuss the choice of the threshold $b$ in Section \ref{sec:theory}. The complete detection procedure is summarized in Algorithm \ref{alg:alg1}.

\begin{algorithm}[tb!]
   \caption{Private Adaptive CUSUM via Graph Perturbation Mechanism}
   \label{alg:alg1}
\begin{algorithmic}
   \STATE {\bfseries Input:} Data sequence $\{\bA_t,t\in\mathbb N\}$, detection threshold $b$.
   \STATE Initialize $S_1=0$; $t=2$. 
   \REPEAT
   \STATE Calculate $\bm{\hat\sigma}_{t-1}$ via solving \eqref{eqn:SDP_relaxation_asymmetric} on perturbed graph $\tilde\bA_{t-1}$.
   \STATE Update $S_{t}=(S_{t-1})^{+} + \log\frac{\Pr(\tilde\bA_t;\bm{\hat\sigma}_{t-1})}{\Pr(\tilde\bA_t;\bm{\sigma}^{\text{pre}})}$.
   \IF{$S_{t} \geq b$}
   \STATE Raise Alarm
   \ENDIF
   \STATE $t=t+1$;
   \UNTIL{Raise Alarm}
\end{algorithmic}
\end{algorithm}

\begin{remark}[Unknown Parameters] In the detection procedure, for simplicity, we make the assumption that $\bm{\sigma}^{\text{pre}}$, $p$, and $\zeta$ are known. Our detection procedure can be extended to the case with unknown parameters as well {(see Appendix~\ref{app:unknown-para} for more details)}. In short, when these parameters are unknown, usually there is a sufficient amount of historical data sampled from the pre-change nominal scenario which can be used to estimate them to a desired accuracy. Furthermore, in the more general case where the parameters $p$ and $\zeta$ may also change, and their post-change values are unknown, we can use the most recent samples to jointly estimate $p$ and $\zeta$ and then substitute to yield the detection statistics.
\end{remark}


\subsection{Detection Procedure under the CDP Setting}\label{sec:private_CDP}


\paragraph{Privacy Mechanism (2): Perturbation Stability Mechanism}

We present a stability-based community estimator $\hat{\bm\sigma}$ that satisfies the CDP conditions, which will then be used in the calculation of detection statistics. Given a graph represented by the adjacency matrix $\bA$, we first compute privately the stability of the ML estimator with respect to ${\bA}$ denoted by $d({\bA})$, which is defined by the \emph{minimum} number of edge modifications on ${\bA}$, so that the estimator output on the modified graph ${\bA}'$ differs from that on ${\bA}$, i.e., $\hat{\bm{\sigma}}({\bA})\neq \hat{\bm{\sigma}}({\bA}')$. The formal definition of stability is given next.
\begin{definition} [Stability of $\hat{\bm{\sigma}}$]\label{def:stab}
 The stability of an estimator $\hat{\bm{\sigma}}$ with respect to a graph ${\bA}$ is defined as follows, 
 \begin{align}
 \label{eqn:stab-sigma}
    d(\bA) = \{\min k:  \exists \bA', \text{dist}(\bA, \bA')\leq k+1, \hat{\bm{\sigma}}(\bA)\neq \hat{\bm{\sigma}}(\bA')\}, \nonumber 
\end{align}
where $\text{dist}(\bA, {\bA}')$ is defined as the number of differing entries between the two adjacency matrices.
\end{definition}
If the graph ${\bA}$ is stable enough w.r.t. the recovery algorithm (i.e., if $d(\bA)$ is larger than a threshold, which depends on $(\epsilon, \delta)$), then we release the non-private estimate $\hat{\bm{\sigma}}({\bA})$; otherwise we release a {\it random} label vector. The key intuition is that from the output of a stable estimator, one cannot precisely infer the presence or absence of a single edge (thereby providing edge DP guarantee). 
The stability-based recovery mechanism is presented in Algorithm \ref{algo:perturb_stability}. The proposed mechanism can be readily shown to satisfy $(\epsilon, \delta)$-edge CDP in Definition \ref{def:edgeDP}  \cite{dwork2006calibrating}.

\begin{algorithm}[ht]
  \caption{Perturbation Stability-based Community Recovery Algorithm}
  \label{algo:perturb_stability}
  \begin{algorithmic}[1]
     \STATE {\bfseries Input:} The graph $G({V}, {E})$ with adjacency matrix $\bA$.
     \STATE {\bfseries Output:} A private labelling vector $\hat{\bm\sigma}$. 
    \STATE $d(\bA) \leftarrow$ compute stability of $\hat{\bm{\sigma}}$ with respect to graph ${G}$ (Definition \ref{def:stab}).
    \STATE $\tilde{d} \leftarrow d({\bA}) + \operatorname{Lap}(0, 1/\epsilon)$.
    \IF{$\tilde{d}  > \frac{\log{1/\delta}}{\epsilon}$}
    \STATE Output $\hat{\bm{\sigma}}({\bA})$
    \ELSE
    \STATE Output $\perp$ (a random label vector) 
    \ENDIF
    \end{algorithmic}
\end{algorithm}

\begin{remark}[Efficient Computation] The complexity of the stability based mechanism using semidefinite relaxation of ML estimator in \eqref{eqn:SDP_relaxation_asymmetric} is $\Theta(n^{\mathcal{O}(\log(n))})$. Demonstrably, the algorithm maintains its efficacy if we opt for $\min\{{d({\bA}), \mathcal{O}(\log{n})}\}$ in place of $d({\bA})$. To achieve this, it is sufficient to compute $\hat{\bm{\sigma}}({\bA}')$ solely for those graphs ${\bA}'$ where $\text{dist}({\bA}, {\bA}')= \mathcal{O}(\log{n})$. Computing the distance to instability still needs to be computationally efficient.  We remedy this problem by designing a computationally efficient private estimator for the distance to instability that runs in $\mathcal{O}(\log^{3}(n) \cdot \text{poly}(n) / \epsilon^{2})$, with a negligible privacy leakage term $\delta$ on the leakage as we show in the following.
\end{remark}

\paragraph{From Concentration to Stability} The complexity of proving the stability of the semidefinite relaxation for ML surpasses that of the vanilla ML due to certain factors related to the optimization problem of SDP. In SDP, the ground truth label is considered as the optimal solution only under some conditions, making it challenging to establish that a solution is not a minimum bisection.  Consequently, we employ the concept of concentration to establish stability. We demonstrate that all graphs within a distance of $\mathcal{O}(\log{n})$ from $G$ also exhibit concentration. Lastly, we derive a lower bound on $a \big(\sqrt{1-\zeta} - \sqrt{\zeta} \big)^{2}$ that satisfies two conditions: (1) it is sufficient for concentration and (2) it preserves concentration even when we flip up to $\Omega(\log{n})$ connections.


 \paragraph{A Computationally Efficient Subsampling Stability Mechanism} In order to calculate the stability of our semidefinite relaxation of ML estimator in \eqref{eqn:SDP_relaxation_asymmetric} efficiently, we present an algorithm by deriving a lower bound on the distance to instability $d$ as a proxy for stability,  termed as ``\textit{subsampling stability-based estimator}'', see Fig. \ref{fig:subsampling_stability_mechanism_illustration} for an illustration. 
 
  \begin{figure}[t]
\centering
	\centering
	{\includegraphics[width=\columnwidth]{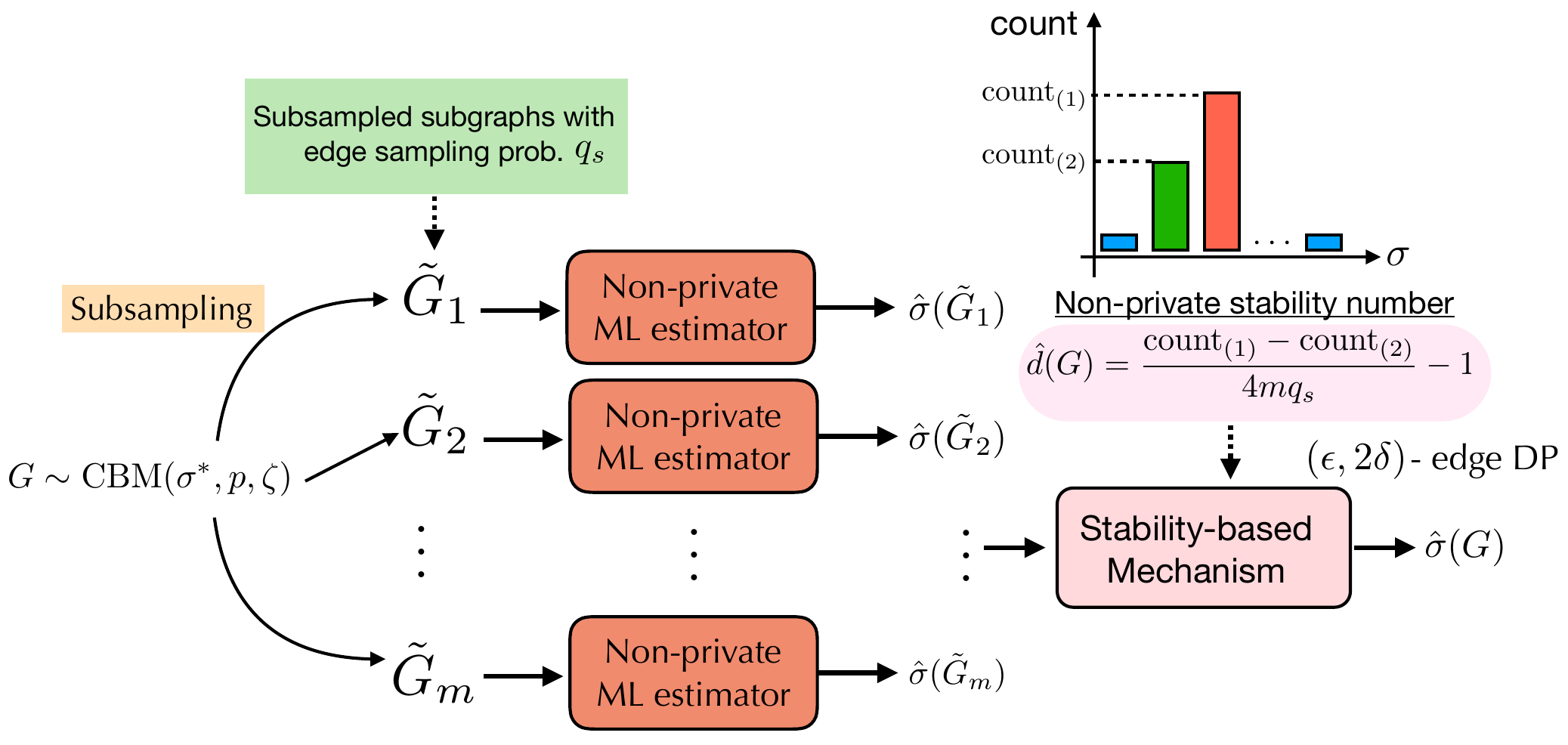}}	
	\caption{\small{Illustration of the subsampling stability-based mechanism:  From the original graph we generate \(m\) correlated subgraphs by
  including each edge independently with probability \(q_s\).
  A non-private community estimator is then applied to every subgraph,
  producing label vectors \(\{\hat{\bm{\sigma}}(\bA_k)\}_{k=1}^m\).
  These labelings are tallied in a histogram, and the modal labeling—
  provably stable under appropriately chosen
  \((q_s,m,\epsilon,\delta)\)—is returned as the private output.}}
\label{fig:subsampling_stability_mechanism_illustration}
    \vspace{-0.1in}
\end{figure}

The key idea in this algorithm is to create $m$ \textit{correlated} subgraphs $\{{\bA}_{1},{\bA}_{2},\ldots,{\bA}_{m}\}$ of the original graph ${\bA}$ where each subgraph ${\bA}_{k}$ is generated by randomly subsampling with replacement of the edges in ${\bA}$ with probability $q_{s}$. We then apply our \textit{non-private} ML-based estimator \eqref{eqn:SDP_relaxation_asymmetric} on each subgraph ${\bA}_{k}$. The labeling vectors $\hat{\bm{\sigma}}({\bA}_{k})$ are then represented on a histogram. {{Now, define $\text{count}{(\bm{\sigma})} \triangleq | \{k\in[m]: \hat{\bm{\sigma}}({\bA}_{k}) =  \bm{\sigma}\} | $.}} As shown in \cite{dwork2014algorithmic}, the stability of the histogram is proportional to the difference between the most frequent bin (i.e., the mode) and the second most frequent bin.  In other words, the most frequent outcome of the histogram agrees with the outcome of the original graph with high probability under an appropriate choice of the CBM parameters $(a, \zeta)$, the privacy budget $(\epsilon, \delta)$, the edge sampling probability $q_{s}$, and the number of subsampled weighted graphs $m$. We summarize the mechanism in Algorithm \ref{algo:sub_stability}.  

\begin{algorithm}[ht]
  \caption{Subsampling-based Efficient Private Community Recovery}
  \label{algo:sub_stability}
  \begin{algorithmic}[1]
     \STATE {\bfseries Input:} The graph $G({V}, {E})$ with adjacency matrix $\bA$.
     \STATE {\bfseries Output:} A private labelling vector $\hat{\bm\sigma}$.
     \STATE $q_{s} \leftarrow \epsilon/32 \log(n)$, $m \leftarrow \log(n/\delta)/q_{s}^{2}$.
     \STATE Subsample $m$ subgraphs ${\bA}_{1}, {\bA}_{2}, \cdots, {\bA}_{m}$.
     \STATE Compute $\bar{\bm{\sigma}} = (\hat{\bm{\sigma}}({\bA}_{1}), \hat{\bm{\sigma}}({\bA}_{2}), \cdots, \hat{\bm{\sigma}}({\bA}_{m}))$.
    \STATE $\hat{d}  \leftarrow (\text{count}_{(1)} - \text{count}_{(2)})/4 m q_{s} - 1$. 
    \STATE $\tilde{d} \leftarrow \hat{d}  + \operatorname{Lap}(0, 1/\epsilon)$.
    \IF{$\tilde{d}  > \frac{\log{1/\delta}}{\epsilon}$}
    \STATE Output $\hat{\bm{\sigma}}({\bA}) = \text{mode}(\bar{\bm{\sigma}})$ 
    \ELSE
    \STATE Output $\perp$ (a random label vector) 
    \ENDIF
  \end{algorithmic}
\end{algorithm}

\paragraph{Detection Procedure}

Given the stability-based community estimator $\hat{\bm\sigma}_{t-1}$ obtained from the last sample $\bA_{t-1}$, we calculate the {privatized} Adaptive CUSUM statistics $\tilde S_t$ as follows: 
\begin{equation}\label{eq:stat-wlcusum2}
\mathcal S_t = (\mathcal S_{t-1})^{+} + \log\frac{\Pr(\bA_t;\bm{\hat\sigma}_{t-1})}{\Pr(\bA_t;\bm{\sigma}^{\text{pre}})}, \quad \tilde S_t = \mathcal S_t + \rm{Lap}\left(\frac{4C}{\epsilon}\right),    
\end{equation}
where $C=2\log(({1-\zeta})/{\zeta})$ is the maximum difference $\big|\log\frac{\Pr(\bA;\bm{\hat\sigma}_{t-1})}{\Pr(\bA;\bm{\sigma}^{\text{pre}})}-\log\frac{p(\bA';\bm{\hat\sigma}_{t-1})}{p(\bA';\bm{\sigma}^{\text{pre}})}\big|$ for any pair of $(\bA,\bA')$ that differs in one edge. Note that the statistics in \eqref{eq:stat-wlcusum2} is different from \eqref{eq:stat-wlcusum} in two ways: (i) the estimator $\bm{\hat\sigma}_{t}$ is obtained differently due to the difference between LDP and CDP settings; (ii) the non-private data $\bA_t$ is utilized to update the statistics, consequently we add Laplace noise to the detection statistics as a common practice to ensure individual data privacy \cite{dwork2006calibrating}.
The stopping rule is constructed similarly as
\begin{equation}\label{eq:stop2}
T_{\rm C} :=\inf\{t: \tilde S_t \geq \tilde b\},    
\end{equation}
where $\tilde b = b + \rm{Lap}(\frac{2C}{\epsilon})$ is a randomized threshold based on a pre-specified constant $b$. 
By the DP property of the AboveThresh algorithm \cite{dwork2009complexity,dwork2006calibrating}, we have the above stopping rule is $\epsilon$-edge  CDP.

\section{Theoretical Analysis}\label{sec:theory}

In this section, we present a theoretical analysis for two proposed private detection frameworks. We first present sufficient and necessary conditions to ensure the exact recovery of the unknown community labels in Section \ref{sec:conditions}. Then we present theoretical guarantees for the average run length and detection delay in Section \ref{sec:detection_theory}. All proofs are deferred to the Appendices.

\subsection{Conditions for Exact Community Recovery }\label{sec:conditions}

\paragraph{Sufficient Condition under Graph Perturbation Mechanism}\label{sec:sufficient_recovery}

In the proposed detection method, we have an online estimate for the unknown community structure. We first analyze the sufficient conditions for exact community recovery under LDP privacy constraints.

\begin{definition}[Exact Recovery] 
 An estimator $\hat{\bm{\sigma}}(\mathbf{A}) = \{\hat{\sigma}_{1}, \hat{\sigma}_{2}, \cdots, \hat{\sigma}_{n}\}$ satisfies exact recovery (up to a global flip) if the probability of error behaves as  
\begin{equation}\label{eqn:exact_recovery_definition}
\operatorname{Pr}(\hat{\bm{\sigma}}(\mathbf{A})= {\bm{\sigma}^{*}} ) \geq  1 - o(1), 
\end{equation}
where the probability is taken over both the randomness of the observation $\bA$ and the estimator. 
\end{definition}

\begin{theorem} \label{thm:private_threshold_condition_one_time_instance} 
Suppose $\epsilon \geq \epsilon_{n} = \Omega(\log(n))$, $p = a \frac{\log(n)}{n}$, and $a > \frac{2 (n^{3/2} - n)}{(n-1) \log(n)}$. The estimator from \eqref{eqn:SDP_relaxation_asymmetric}, using perturbed graph $\tilde\bA$ under our Privacy Mechanism in Section \ref{sec:ldp-method}, satisfies exact recovery if  
\begin{align}\label{eq:recov_condi}
 a \big(\sqrt{1-\zeta} - \sqrt{\zeta} \big)^{2} > \left( \frac{\sqrt{n} }{\sqrt{n} -1} \right) \times  \left(\frac{e^{\epsilon} + 1}{e^{\epsilon} - 1} \right).
\end{align}
\end{theorem}
\noindent {{The proof of this theorem is presented in Appendix~\ref{app:proofIV1}.}} Observe that, if $\epsilon$ remains constant while $n$ increases, it follows that from \eqref{eqn:perturbed_parameters}, $\lim_{n\rightarrow \infty} \tilde{\zeta} = 1/2$. This suggests that with a consistent $\epsilon$, the characteristics of both inter- and intra-community edges converge in the asymptotic limit, rendering exact recovery infeasible. However, Theorem \ref{thm:private_threshold_condition_one_time_instance} demonstrates that exact recovery can be achieved when the leakage grows logarithmically with $n$.  Our analysis also reveals an interesting trade-off between the privacy scaling coefficient $c$, where $\epsilon = c \log(n)$ and the CBM parameter $a$ as highlighted in Fig. \ref{fig:phase_transition_mechanisms_app} in the Appendix.

\paragraph{Sufficient Condition under Stability Mechanism} We next present a sufficient exact recovery condition under CDP for the perturbation stability-based recovery algorithm in the following theorem. We also plot and compare the threshold conditions for exact recovery under both mechanisms, as shown in Fig. \ref{fig:phase_transition_mechanisms}.

\begin{theorem} \label{thm:perturbation_stability} The perturbation stability-based mechanism (described in Algorithm \ref{algo:perturb_stability}) satisfies $(\epsilon, \delta)$-edge CDP, and it satisfies exact recovery if 
\begin{align}\label{eq:cond2}
         a \big(\sqrt{1-\zeta} - \sqrt{\zeta} \big)^{2} >  1,
\end{align}
for $a > 3/\epsilon$, $p = a \log(n)/n$, $\epsilon>0$, and $\delta = n^{-2}$. 
\end{theorem}

\begin{figure}[t]
\centering
	\centering
	{\includegraphics[width= 0.75\columnwidth]{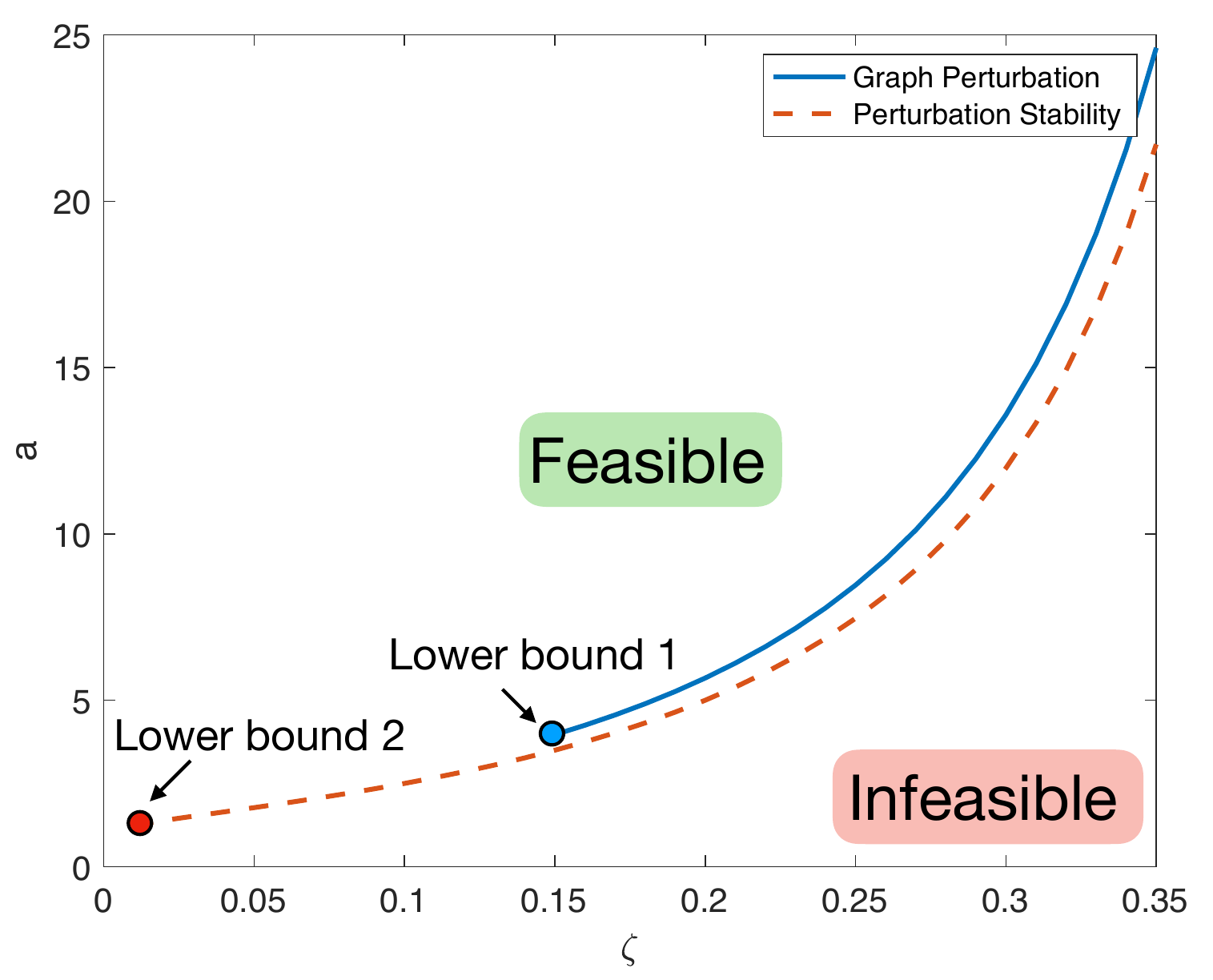}}	
	\caption{\small{Exact Recovery Regime that depends on $(a,\zeta)$, where $n = 100$ and $\epsilon = \log(n)$.}}
    \label{fig:phase_transition_mechanisms}
\end{figure}


We also present a sufficient condition for exact recovery for the computationally-efficient stability mechanism in the following theorem.

\begin{theorem} \label{thm:sub_stability} Algorithm \ref{algo:sub_stability} satisfies $(\epsilon, 2\delta)$-edge CDP, and it satisfies exact recovery if 
\begin{align}\label{eq:cond2}
         a \big(\sqrt{1-\zeta} - \sqrt{\zeta} \big)^{2} >  \max \left\{\frac{32 \log(n)}{\epsilon}, 1\right\},
\end{align}
for $p = a \log(n)/n$, $\epsilon>0$, and $\delta = n^{-2}$. 
\end{theorem}
\noindent {{The proof follows along lines similar to those in our work in \cite{seif2024differentially}.}

\begin{remark}[Comparison with Graph Perturbation Mechanism]
  A key advantage of the perturbation stability mechanism is its constant privacy leakage guarantee, i.e., $\epsilon = \mathcal{O}(1)$. However, it leads to a slightly higher running time and more strict recovery conditions. This highlights the fundamental tradeoffs between computational complexity, recovery conditions, and privacy guarantees. 
\end{remark}

{{
\begin{remark}[Memory Footprint and Scalability]
Efficient storage is critical when privacy mechanisms are applied
to large graphs; the following bounds summarise the peak RAM usage
of each step and how we scale beyond it.
\begin{itemize}
\item {Randomized‑response:}  A dense representation of the perturbed
      graph needs $\Theta(n^{2})$ bits, whereas a compressed edge list
      uses only $\Theta(m)$ words for sparse graphs.
   \item {Subsampling-stability mechanism:} 
    At any given time, only the current subsample \(S\subseteq E\) of size \(\lvert S\rvert = \beta\,m\) (for some \(0 < \beta < 1\)) is held in RAM. Hence, the peak memory cost is $\Theta\bigl(n + \beta\,m\bigr)$,
    which is linear in the graph size and strictly below the cost of storing the full edge list whenever \(\beta < 1\).
\end{itemize}
For massive graphs, we can further reduce resources by sketching \cite{seif2024differentially} the graph to
$\tilde{O}(n^{\alpha})$ nodes/edges with $\alpha<1$, running community
detection on the sketch, and propagating the labels back.  
This lowers both time and memory by a factor of $n^{1-\alpha}$ with negligible loss in accuracy, and the same idea accelerates the SDP
formulation even under the extra separation constraint induced by subsampling.
\end{remark}
}
}

\paragraph{Necessary Condition for Exact Community Recovery}

We subsequently introduce an information-theoretic lower bound that delineates the minimum necessary privacy levels for achieving exact community recovery in $\epsilon$-edge CDP frameworks. This lower bound serves as a critical benchmark for assessing the effectiveness and limitations of different privacy-preserving mechanisms.

\begin{theorem} [Converse Result]\label{thm:converse} Suppose there exists an $\epsilon$-edge CDP mechanism such that for any labeling vector $\bm{\sigma} \in \{\pm 1\}^{n}$, on $\mathbf{A} \sim \CBM(\bm\sigma,p,\zeta)$, output $\hat{\bm{\sigma}} \in \{\pm 1\}^{n}$ satisfying the exact recovery \eqref{eqn:exact_recovery_definition}. Then, the privacy leakage must scale at least as
\begin{align}
\epsilon & \geq  \frac{1}{2} \log \left[ 1 + \frac{2 \log(n) - \log(8 e)}{ p'  (4n-32)} \right], \label{eqn:lower_bound_privacy_CBM}
\end{align} 
where $p' = 2p^2\zeta(\zeta - 1) - (p - 1)^2 + 1$, $ p = a \log(n)/n$, and $\zeta \in (0, 1/2)$.
\end{theorem}
\noindent {{The proof of this theorem is presented in Appendix~\ref{app:proofIV4}.}}  The lower bound derived in this context reveals the direct connection between the CBM parameters \( a \) and \( \zeta \), and the privacy budget \( \epsilon \) required for exact community recovery. Distinctly, this lower bound decays slower than the bound derived from the traditional packing lower bound argument method, as discussed in \cite{vadhan2017complexity}. In contrast to the constant nature of our derived lower bound, the packing argument suggests a more dynamic relationship, specifically, \( \epsilon = \Omega \left({\log(n)}/{{n \choose 2}} \right) \), highlighting a less robust link with \( n \). For smaller values of \( \epsilon \leq 1 \), the lower bound on the privacy budget in \eqref{eqn:lower_bound_privacy_CBM} scales as \( \epsilon = \Omega \left({\log(n)}/{n}\right) \), which indicates a nuanced scaling behavior in comparison to the previously discussed bound. It is worth highlighting that the proposed graph perturbation mechanism requires $\epsilon = \Omega(\log(n))$. There remains an interesting opportunity to close this gap by devising more sophisticated mechanisms.

\subsection{Guarantees on Average Run Length and Detection Delay}\label{sec:detection_theory}
In this subsection, we provide guarantees for the average run length and detection delay of both the procedure $T_{\rm L}$ under LDP setting in Equation \eqref{eq:stop_time} and $T_{\rm C}$ under CDP setting in Equation \eqref{eq:stop2}.

\begin{lemma}[Average Run Length of $T_{\rm L}$]\label{lem:arl}
For a given detection threshold $b>0$, we have the test $T_{\rm L}$ in method \eqref{eq:stop_time} satisfies $\bE_\infty[T_{\rm L}]\geq e^{b}$.
\end{lemma}
\noindent {{The proof of this lemma is presented in Appendix~\ref{app:prooflemIV1}.}} In online change detection problems, a pre-specified large constant $\gamma$ is typically imposed as a lower bound on the average run length to prevent frequent false alarms. Based on Lemma \ref{lem:arl}, we may choose the detection threshold $b=\log\gamma$ to guarantee $\bE_\infty[T_{\rm L}]\geq \gamma$ is satisfied. We next present the theoretical guarantee of the detection effectiveness.  

\begin{theorem}[Detection Delay of $T_{\rm L}$]\label{th:upper_wadd}
Assume the exact recovery condition \eqref{eq:recov_condi} is satisfied, we have the following worst-case average detection delay of the detection procedure \eqref{eq:stop_time}, when $b=\log\gamma$ and $\tilde I_0=o(\log\gamma)$ as $\gamma\to\infty$:
\begin{equation}
\WADD[T_{\rm L}]=\frac{\log\gamma }{\tilde I_0} (1+o(1)),
\label{eq:th1}
\end{equation}
where 
\[
\tilde{I}_0=\frac12 (\log\frac{1-\tilde\zeta}{\tilde\zeta})\tilde p(1-2\tilde\zeta)\left({n\choose 2} - \sum_{i<j}{\sigma}^\pre_i{\sigma}^\pre_j{\sigma}_i^\post{\sigma}_j^\post\right).
\]
\end{theorem}
\noindent {{The proof of this theorem is presented in Appendix~\ref{app:proofIV5}.}} 
Here $\tilde I_0$ is the KL divergence between the post-change privatized distribution $\CBM(\bm\sigma^\post,\tilde p,\tilde\zeta)$ and pre-change privatized distribution $\CBM(\bm\sigma^\pre,\tilde p,\tilde\zeta)$ (we refer to Appendix \ref{sec:llr} for detailed derivations). It is worthwhile mentioning that the delay in \eqref{eq:th1} matches the well-known information-theoretic lower bound for any detection procedures applied to the privatized data \cite{tartakovsky2014sequential}. This implies that our detection procedure maintains the asymptotic optimality in terms of detection delay.



\begin{remark}[Comparison with Non-Private Detection]\label{rem:non-private}
As shown in Theorem~\ref{thm:delay_lower} in the Appendix~\ref{app:proofIV5}, the information-theoretic lower bound on the WADD is given by $\frac{\log\gamma}{I_0}(1+o(1))$, where $I_0$ denotes the KL divergence between the post- and pre-change data distributions \cite{tartakovsky2014sequential}. This bound is known to be achieved by the classical CUSUM procedure and thus serves as the optimal performance benchmark for non-private detection.
In contrast, the detection delay of our private method above is $\frac{\log\gamma}{\tilde{I}_0}(1+o(1))$, where $\tilde{I}_0 < I_0$ due to the information loss induced by privacy-preserving graph perturbations. This highlights the inherent trade-off between privacy guarantees and detection efficiency. As the privacy requirement becomes weaker (i.e., larger $\epsilon$), $\tilde{I}_0$ approaches $I_0$, and the performance of our private detection method approaches that of the optimal non-private baseline.
\end{remark}

Similar to the results above, we also present guarantees on the test $T_{\rm C}$ under CDP constraints.

\begin{lemma}[Average Run Length of $T_{\rm C}$]\label{lem:arl2}
Assume $\epsilon > 8\log\frac{1-\zeta}{\zeta}$. For a given detection threshold $b>0$, we have the test \eqref{eq:stop2} satisfies 
$
\bE_\infty[T_{\rm C}]\geq \frac{1-(4C/\epsilon)^2}{1-(2C/\epsilon)^2} e^b
$, where $C=2\log(({1-\zeta})/{\zeta})$.
\end{lemma}
\noindent {{The proof of this lemma is presented in Appendix~\ref{app:prooflemIV2}.}} 
Again, we may choose the threshold $b=\log\gamma + \log\frac{1-(2C/\epsilon)^2}{1-(4C/\epsilon)^2}$ to guarantee $\mathbb E_\infty[T_{\rm C}]\geq \gamma$ is satisfied. For a constant privacy leakage $\epsilon$, we have the resulted threshold $b = \log\gamma \cdot (1+o(1))$.

\begin{theorem}[Detection Delay of $T_{\rm C}$]\label{thm:wadd2} 

Assume the exact recovery condition \eqref{eq:cond2} is satisfied and $\epsilon > 8\log\frac{1-\zeta}{\zeta}$, we have the following worst-case average detection delay of the detection procedure \eqref{eq:stop2}, when $b=\log\gamma+\log\frac{1-(2C/\epsilon)^2}{1-(4C/\epsilon)^2}$ and $I_0=o(\log\gamma)$ as $\gamma\to\infty$:
\begin{equation}
\WADD[T_{\rm C}]=\frac{\log\gamma }{I_0} (1+o(1)),
\label{eq:th2}
\end{equation}
where 
\[
{I}_0=\frac12 (\log\frac{1-\zeta}{\zeta}) p(1-2\zeta)({n\choose 2} - \sum_{i<j}{\sigma}^\pre_i{\sigma}^\pre_j{\sigma}_i^\post{\sigma}_j^\post).
\]
\end{theorem}
\noindent {{The proof of this theorem is presented in Appendix~\ref{app:proofIV6}.}}  We note that ${I}_0$ is the KL divergence of the post- and pre-change $\CBM$ distributions without perturbation with ${I}_0>\tilde I_0$. Consequently, the delay in \eqref{eq:th2} is smaller than that in \eqref{eq:th1} (and $\tilde I_0\approx I_0$ for $\epsilon$ sufficiently large). This is consistent with the intuition that CDP is less stringent than LDP constraints, thus having less impact on the degradation of detection performance. On the other hand, the detection procedure under CDP setting requires more strict recovery conditions and a slightly higher running time caused by computation of stability.

\subsection{Delay Lower Bound with Private Constraint}\label{sec:info-lower-bound}

\subsubsection{Delay Lower Bound under $\epsilon$-edge LDP}


In this section, we investigate the minimax lower bound to the detection delay for {\it arbitrary} private algorithms that satisfies $\epsilon$-edge LDP,
\begin{equation}\label{eq:minimax}
  \begin{aligned}
   &\inf_{Q\in \mathcal Q_{\epsilon}} \inf_{\tau\in \mathcal C(\gamma)} \underset{\nu \geq 1}\sup \esssup \ \mathbb E_{\nu,Q}\left[(\tau-\nu+1)^+| \mathcal F_{\nu-1}\right],  
   \end{aligned}
\end{equation}
where $\mathcal Q_{\epsilon}$ is the set of all privacy mechanisms that satisfy the $\epsilon$-edge LDP constraint, $\tau$ ranges over all stopping rules in the set $\mathcal C(\gamma)$ satisfying the ARL constraint $\mathbb E_\infty[\tau]\geq \gamma$; and the inner supremum denotes the worst-case detection delay over all possible change-point locations, with $\mathbb E_{\cdot,Q}$ denoting the expectation under the privacy mechanism $Q$. We aim to provide insights regarding a general lower bound to \eqref{eq:minimax} and compare with the delay of our proposed method in the end.


To solve the above minimax lower bound, we first analyze the KL divergence after the privacy mechanism. We denote $p_\infty(\bA)$ and $p_0(\bA)$ as the probability distribution for the adjacency matrix $\bA$ in the pre-change and post-change regime, respectively. And denote $\tilde p_\infty(\bA)$ and $\tilde  p_0(\bA)$ as the corresponding two distributions for the privatized data distribution.
%
Following \cite{duchi2018minimax}, we have the following upper bound for the KL divergence after {\it any} $\epsilon$-edge LDP perturbation.
\begin{lemma}[KL divergence after perturbation]\label{lemma:KL_after_pertub} For any privacy mechanism $Q$ satisfying the $\epsilon$-edge LDP in Definition \ref{def:edgeLDP}, we have the KL divergence after perturbation is upper bounded as follows:
   \[
   \KL(\tilde  p_0 || \tilde  p_\infty)\leq C_\epsilon(e^\epsilon -1)^2 p^2(1-2\zeta)^2\bigg[{n \choose 2} - C_{\bm\sigma^\pre,\bm\sigma^\post}\bigg],
   \] 
   with $C_\epsilon=\min\{4,e^{2\epsilon}\}$,  $C_{\bm\sigma^\pre,\bm\sigma^\post}:=\sum_{i<j}{\sigma}^\pre_i{\sigma}^\pre_j{\sigma}_i^\post{\sigma}_j^\post$.
\end{lemma}
\noindent {{The proof of this theorem is presented in Appendix~\ref{app:prooflemIV3}.}} 
Since the optimal worst-case average detection delay, under ARL constraint $\gamma$, for detecting distribution shifts from $\tilde p_\infty$ to $\tilde  p_0$ is known to be $({\log\gamma})/{\KL(\tilde  p_0 || \tilde  p_\infty)}\cdot (1+o(1))$ \cite{tartakovsky2014sequential}, the above Lemma implies that 
\begin{multline}\label{eq:lower_tentative}
   \inf_{Q\in \mathcal Q_{\epsilon}} \inf_{\tau\in\mathcal C(\gamma)} \underset{\nu \geq 1}\sup \esssup \ \mathbb E_{\nu,Q}\left[(\tau-\nu+1)^+| \mathcal F_{\nu-1}\right] \geq \\
 \frac{\log\gamma}{C_\epsilon(e^\epsilon -1)^2 p^2(1-2\zeta)^2 ({n \choose 2} - C_{\bm\sigma^\pre,\bm\sigma^\post})} (1+o(1)).   
\end{multline}

Meanwhile, recall that in our privacy mechanism, we have that the KL divergence after the randomized perturbation becomes $\tilde I_0=\left(\log\frac{1 + p(1-\zeta)(e^\epsilon-1)}{1 + p\zeta(e^\epsilon-1)}\right)\cdot p(1-2\zeta)\frac{e^\epsilon-1}{e^\epsilon+2}({n \choose 2} - C_{\bm\sigma^\pre,\bm\sigma^\post}) \approx p(1-2\zeta)(e^\epsilon -1)\cdot p(1-2\zeta)\frac{e^\epsilon-1}{e^\epsilon+2}({n \choose 2} - C_{\bm\sigma^\pre,\bm\sigma^\post})$ for small $\epsilon$, which matches with the denominator in \eqref{eq:lower_tentative} up to a constant factor. However, for large $\epsilon$ values, there still remains an opportunity to close the gap between the delay of our proposed algorithm and the lower bound in \eqref{eq:lower_tentative} by developing more efficient private perturbation mechanisms.

\subsubsection{Delay Lower Bound under $(\epsilon, \delta)$-edge CDP}

We analyze the lower bound on detection delay under CDP settings, focusing on a specific scenario where the detection problem is performed based on the output of binary hypothesis tests on individual observations at each time. More specifically, we consider a hypothesis test $\mathcal T$ with binary output $\{0,1\}$, where $\mathcal T(\bA) = 1$ implies there is a change (i.e., accept the alternative hypothesis $H_{1}$ that the change has happened) and vice versa, based on the observed graph $\bA$. Now, for any two adjacent graphs $\bA,\bA'$, an $(\epsilon, \delta)$-edge CDP test $\mathcal T$ should satisfies
\[
{\operatorname{Pr}(\mathcal T(\mathbf{A}') = a)} \leq e^{\epsilon} { \operatorname{Pr}(\mathcal T({\mathbf{A}}) = a)} + \delta, \text{ for } a=0,1.
\]
Then we have the following Lemma that quantifies the performance of the $(\epsilon, \delta)$-edge CDP test $\mathcal T$.
\begin{lemma} \label{lem:private_hypo_test}
For an $(\epsilon, \delta)$-edge CDP binary hypothesis test $\mathcal T$, we have 
    \begin{align}\label{eq:hypo-test-diff-cdp}
     & { \operatorname{Pr}(\mathcal T({\mathbf{A}}) = 1 | p_0)} -  { \operatorname{Pr}(\mathcal T({\mathbf{A}}) = 1 | p_\infty)} \nonumber \\
     & \leq \frac{1}{\sqrt{2}} \times \left( \frac{e^{R\epsilon} - 1}{e^{R\epsilon} + 1} \right) \times \left(1 + \frac{2 \delta}{e^{\epsilon} - 1}\right) \times \sqrt{\KL (p_{0} || p_{\infty} )},
\end{align}
where $p_\infty(\bA)$ and $p_0(\bA)$ denote the probability distribution for the adjacency matrix $\bA$ in the pre-change and post-change regime, respectively, and $R = 2^{{n \choose 2}}$.
\end{lemma}

\noindent {{The proof of this lemma is presented in Appendix~\ref{app:prooflemIV4}.}} 

We note that for non-private binary hypothesis test, we have the following bound
    \[
     { \operatorname{Pr}(\mathcal T({\mathbf{A}}) = 1 | p_0)} -  { \operatorname{Pr}(\mathcal T({\mathbf{A}}) = 1 | p_\infty)} \leq \TV (p_{0} || p_{\infty} ). 
    \]
Indeed, the summation of type-I and type-II errors of any test is lower bounded by $1-\TV (p_{0} || p_{\infty} )$, with this minimum value obtained by the Neyman-Pearson test. Compared with the above Equation \eqref{eq:hypo-test-diff-cdp}, we observe that the additional constant $\frac{1}{\sqrt{2}}\frac{e^{R\epsilon} - 1}{e^{R\epsilon} + 1}  \left(1 + \frac{2 \delta}{e^{\epsilon} - 1}\right)$ can be viewed as the performance loss due to the privacy constraints.

We then present the general analysis of the information bound on the detection delay. Based on the converse result above for private hypothesis testing for a single observation, we consider the scenario that at each time $k$, an $(\epsilon, \delta)$-edge CDP test $\mathcal T$ is performed and the binary output is $O_k = \mathcal T(\bA_k)\in \{0,1\}$. The detection is performed via the sequence of binary outputs $\{O_k,k\in\mathbb N\}$.

\begin{proposition}\label{prop-delay-lower-cdp}
 When using the sequence of binary outputs $\{O_k,k\in\mathds{N} \}$ for detection, with each $O_k$ resulted from an $(\epsilon,\delta)$-edge CDP hypothesis test $\mathcal T$ applied on sample $\bA_k$, the detection delay is lower bound by,
 \[
\WADD \geq \frac{\log\gamma}{\frac{1}{\alpha_0} \left( \frac{e^{R\epsilon} - 1}{e^{R\epsilon} + 1} \right)^2 \left(1 + \frac{2 \delta}{e^{\epsilon} - 1}\right)^2 \KL (p_{0} || p_{\infty} )} (1+o(1)),
\]
as $\gamma\to\infty$, where $\alpha_0:=\operatorname{Pr}(\mathcal T({\mathbf{A}}) = 1 | p_\infty)>0$.
\end{proposition}
\noindent {{We prove this proposition in Appendix~\ref{app:proofpropIV1}.}} 
Recall that the information-theoretic lower bound to the asymptotic delay in the non-private setting is $(\log\gamma)/\KL (p_{0} || p_{\infty} )$ \cite{tartakovsky2014sequential}. Compared with the delay lower bound above, we note that the privacy impact is reflected in the constant term in the denominator that depends on $\epsilon$ and $\delta$.


\subsection{Converse: How Many Graphs are Needed for Detection with Privacy Constraint?}


In this subsection, we present a necessary condition for the number of graphs required to achieve the exact recovery condition under 
$\epsilon$-edge DP for arbitrary $\epsilon$. It is worth highlighting that the community recovery problem becomes easier as more graphs are observed and utilized together within the recovery algorithm.



\begin{theorem} \label{lemma:required_window_with_privacy} For the CBM model, the number of observed graphs needed for exact recovery while satisfying $\epsilon$-edge DP should be at least 
\begin{align}
     w \geq \max\left\{ \tilde{c}_{1} \times \frac{e^{\epsilon}}{e^{\epsilon}-1}, \tilde{c}_{2} \times \log(n) \right\}, \label{eqn:condition_window_size}
\end{align}
for an arbitrary  $\epsilon > 0$, $\tilde{c}_{1} = 1 - n^{-\Omega(1)}$ and $\tilde{c}_{2} \triangleq {(4 - 4 n^{-\Omega(1)} )}/{\left(1-2 - 2 n^{-\Omega(1)} \right) ^{2} }$.
\end{theorem}
\noindent {{The proof of this theorem is presented in Appendix~\ref{app:proofIV7}.}} 
Here, the first term in the argument of \eqref{eqn:condition_window_size} represents the edge privacy impact on the required window size, while the second term reflects the statistical limit required for exact recovery and follows from applying the basic Chernoff-Hoeffding concentration inequality. The above condition shows that the privacy impact is dominant in the high-privacy regime, where $\epsilon$ is small. Conversely, in the low-privacy regime, the lower bound on $w$ is mainly governed by the fundamental statistical limit, i.e., $w = \Omega(\log(n))$. This dual dependency ensures that the window size $w$ adapts appropriately to balance both privacy and recovery accuracy requirements.

The above theorem demonstrates that exact recovery is possible for arbitrary $\epsilon > 0$ when 
$w >1$, unlike the previous case where $w =1$ (i.e., a single observed graph) in Theorem \ref{thm:converse}. Consequently, in order to achieve exact recovery, we may use the most recent $w$ graphs to estimate the unknown community and then calculate the detection statistics. We would like to highlight that while using additional graphs can be beneficial for improving recovery, it may be costly to store the most recent $w$ observations in online change detection especially when the required window size $w$ is large.


\section{Numerical Results}\label{sec:numerical}

\begin{figure*}[tb]
    \centering
    \begin{tabular}{cccc}
  \includegraphics[width=0.22\textwidth]{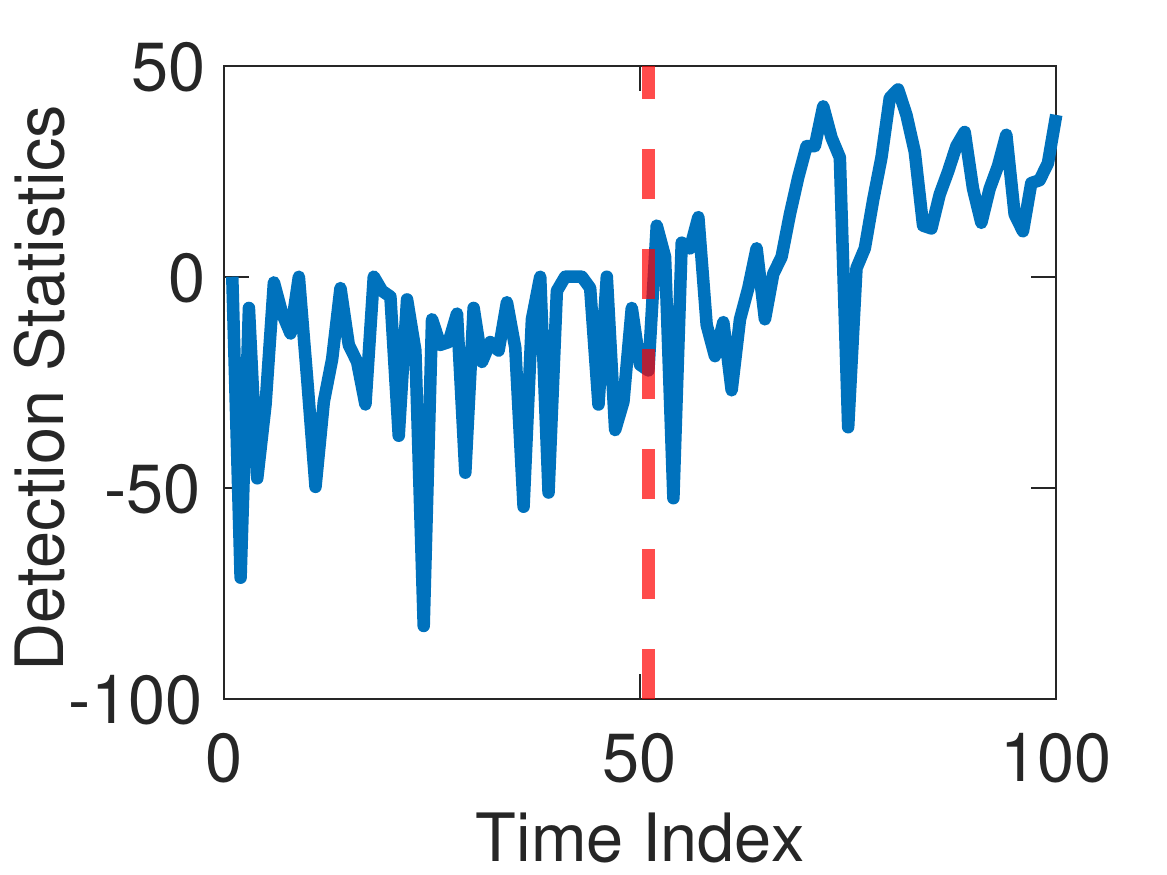}  & \includegraphics[width=0.22\textwidth]{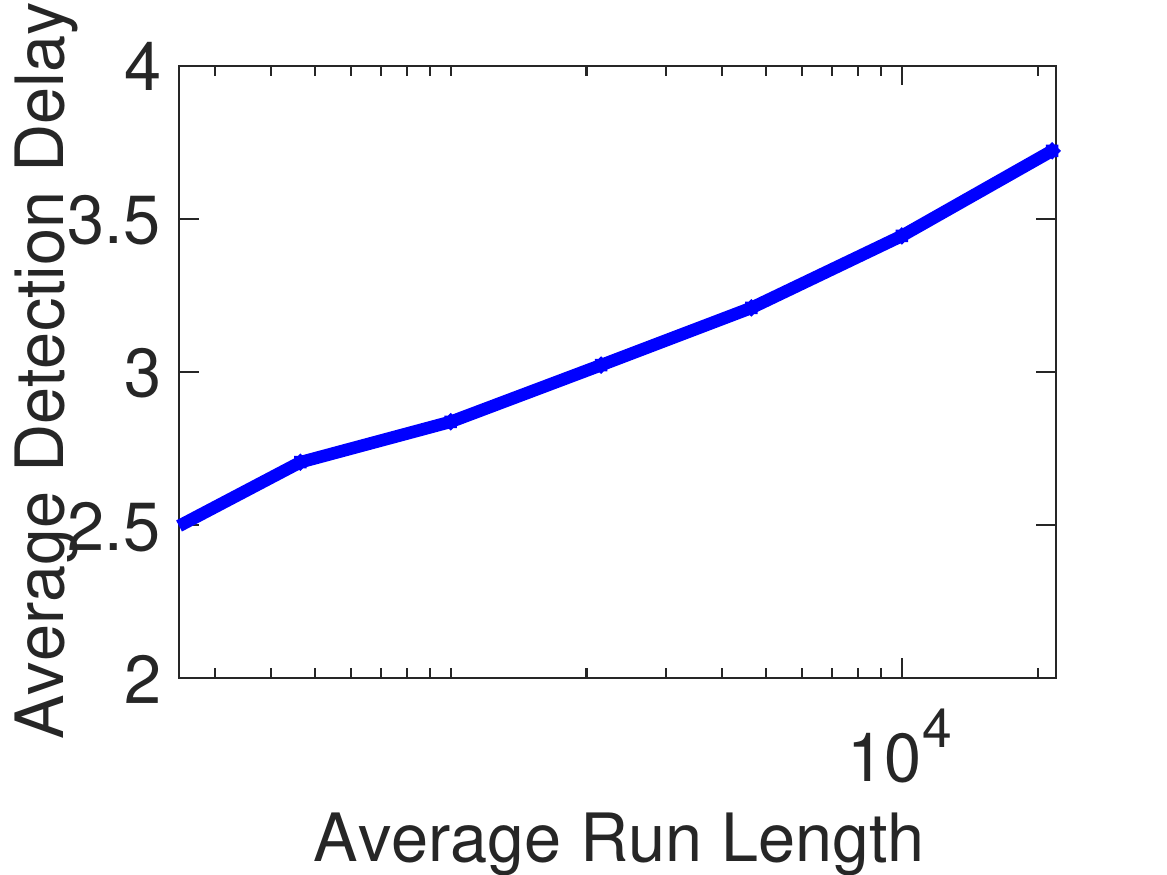} &  \includegraphics[width=0.22\textwidth]{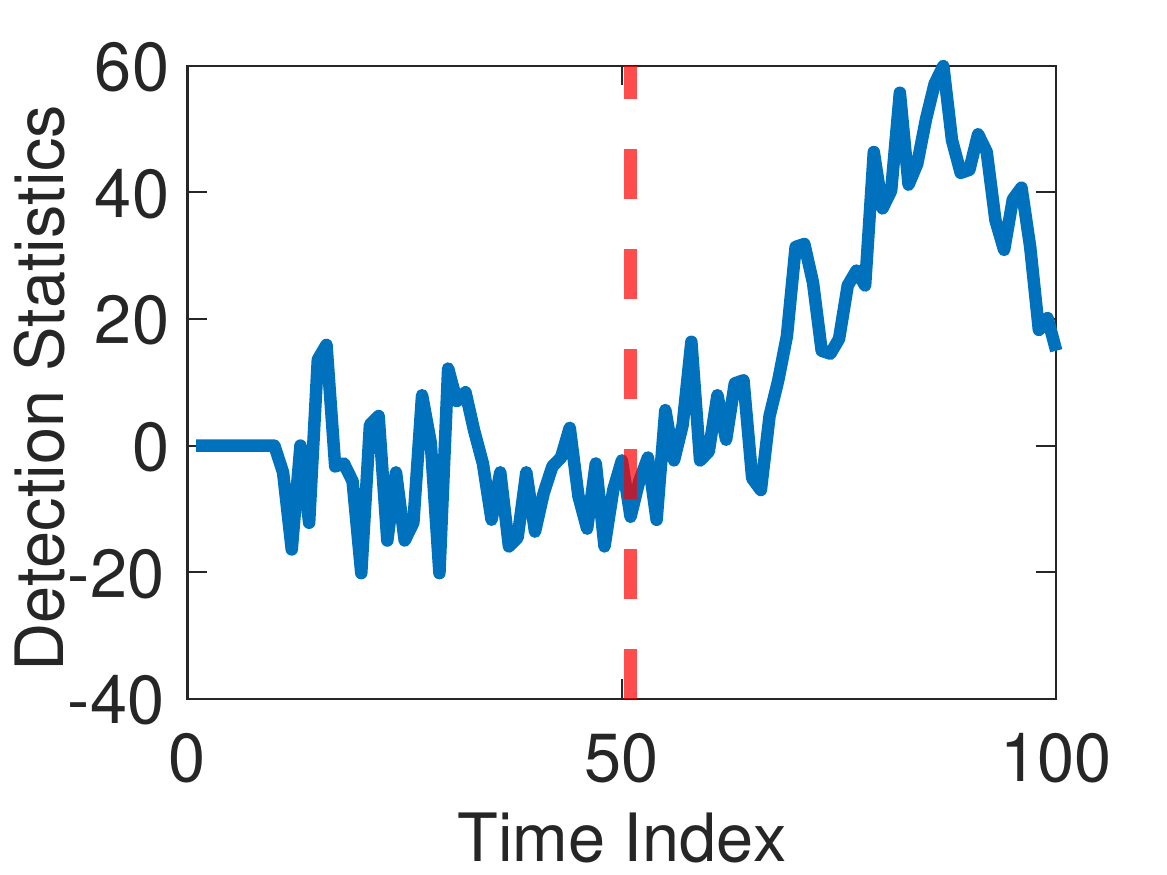} &\includegraphics[width=0.22\textwidth]{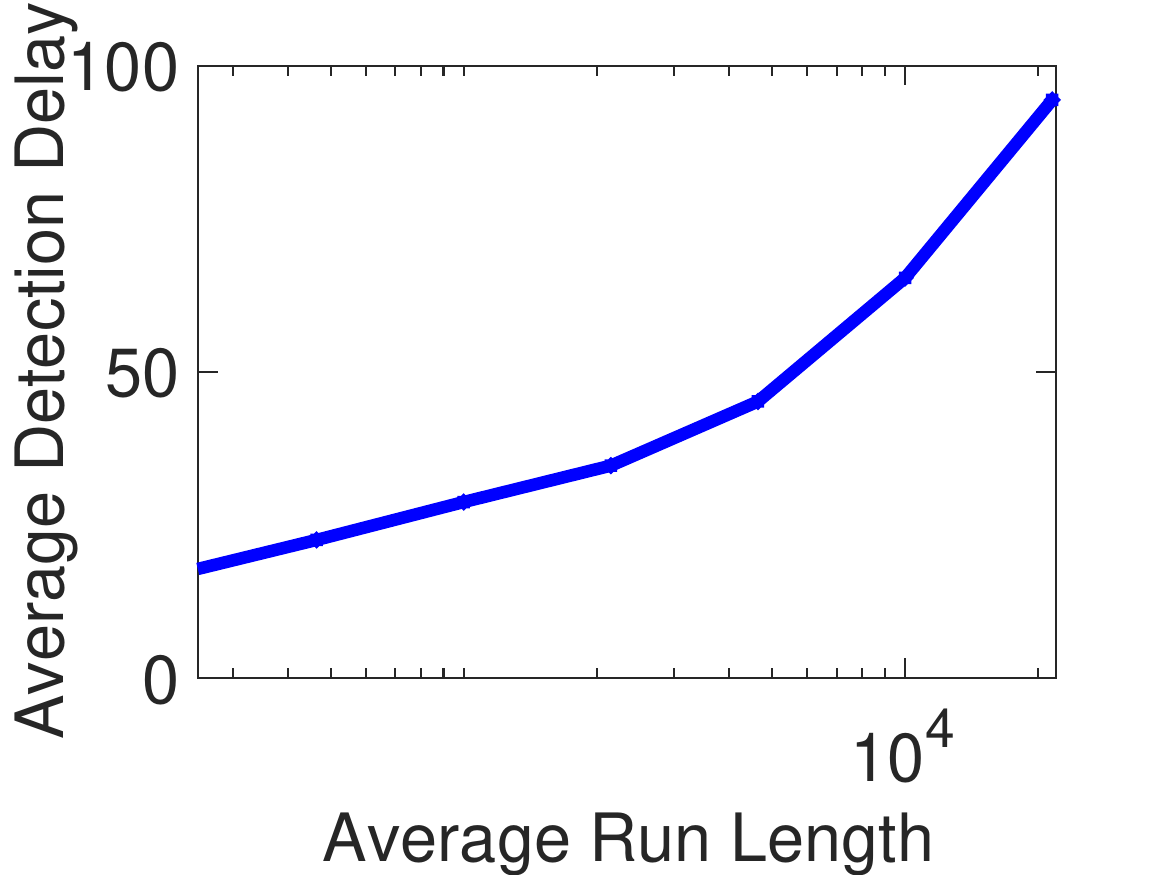}  \\
 \multicolumn{2}{c}{Case 1} & \multicolumn{2}{c}{Case 2}
    \end{tabular}
    \caption{Sample statistics trajectories and detection delay of the method $T_{\rm L}$ in \eqref{eq:stop_time} for both cases.}
    \vspace{-0.1in}
    \label{fig:stat}
\end{figure*}

\begin{figure*}[tb]
    \centering
    \begin{tabular}{cccc}
  \includegraphics[width=0.22\textwidth]{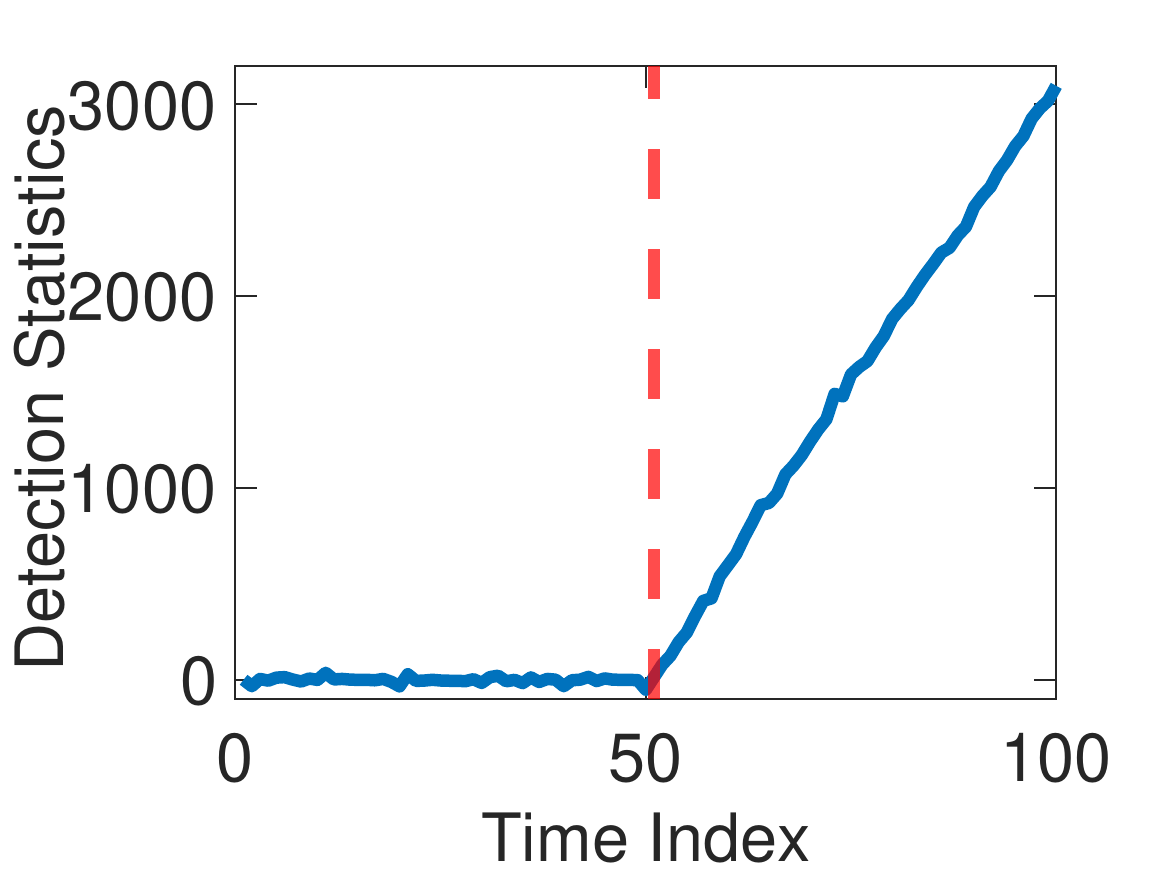}  & \includegraphics[width=0.22\textwidth]{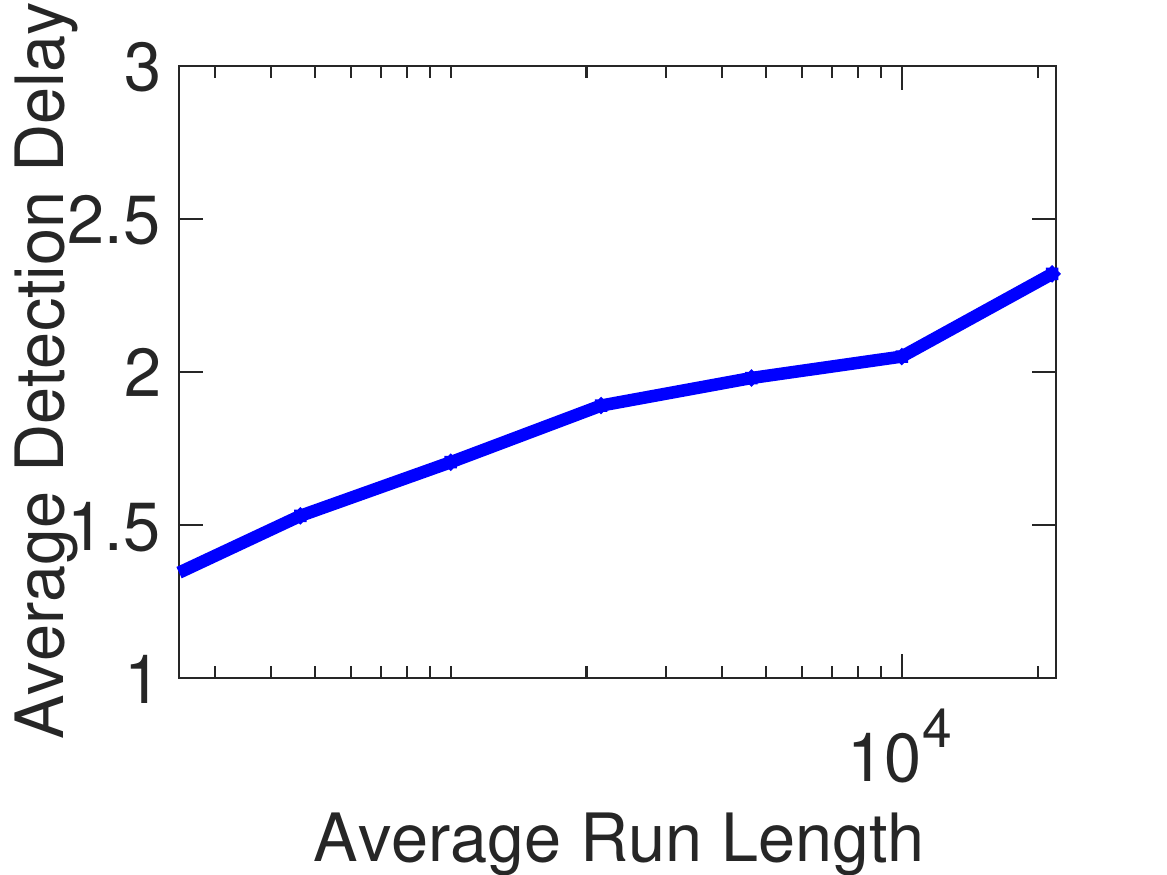} &  \includegraphics[width=0.22\textwidth]{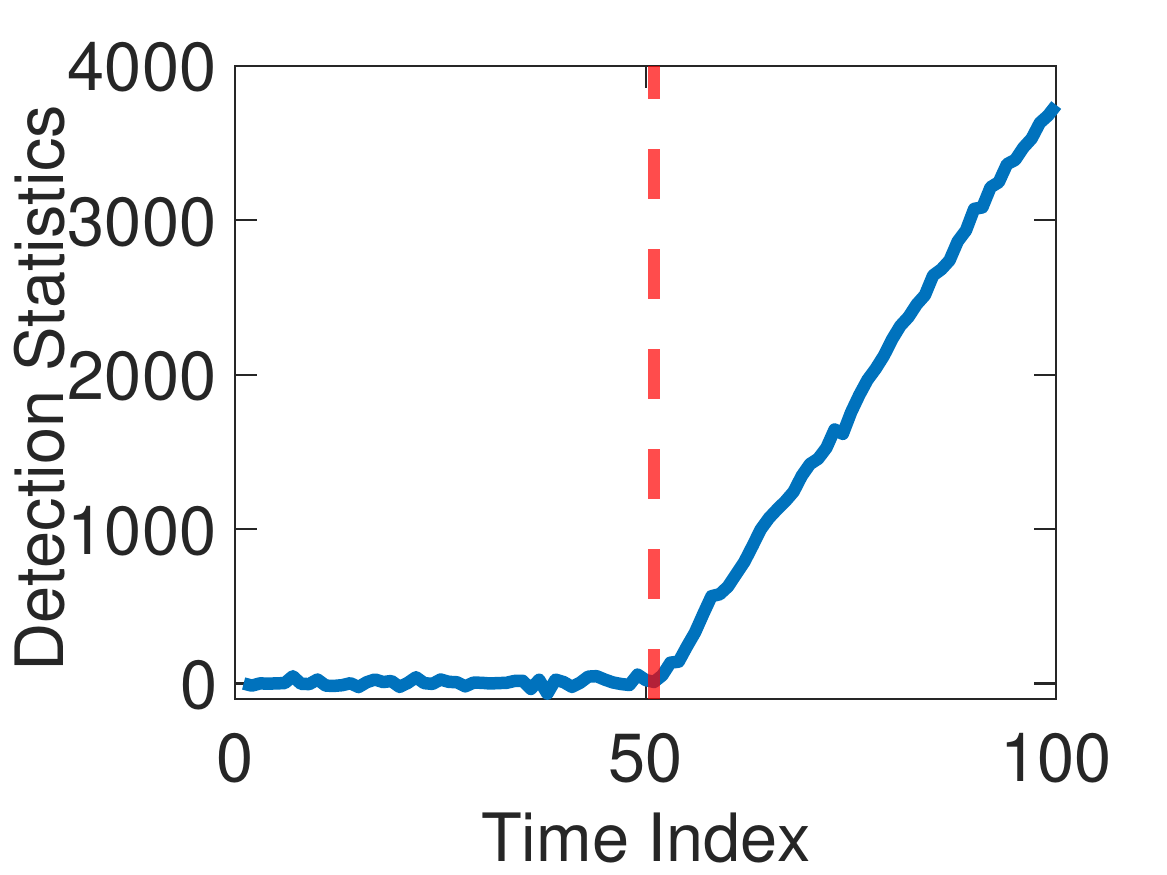} &\includegraphics[width=0.22\textwidth]{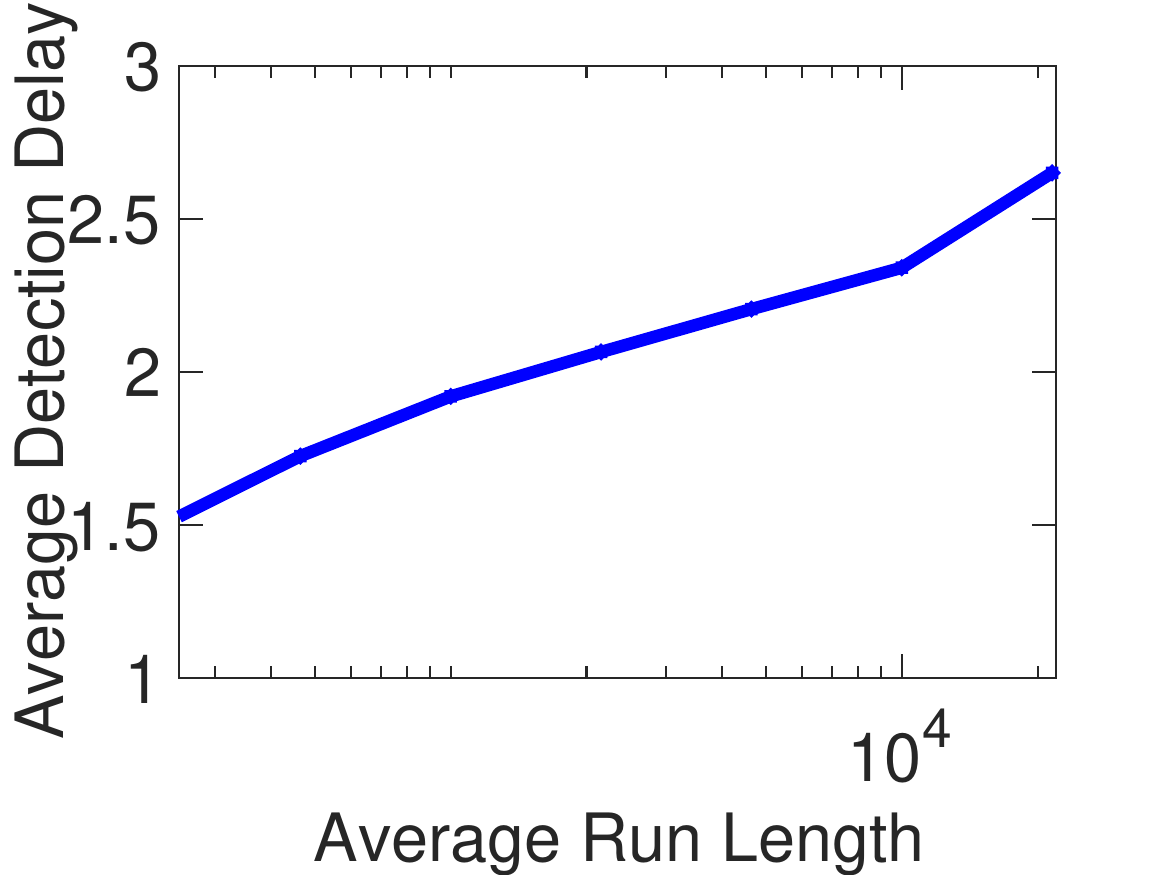}  \\
 \multicolumn{2}{c}{Case 1} & \multicolumn{2}{c}{Case 2}
    \end{tabular}
    \caption{Sample statistics trajectories and detection delay of the method $T_{\rm C}$ in \eqref{eq:stop2} for both Cases.}
    \label{fig:stat2}
\end{figure*}

In this section, we present numerical results to assess the performance of proposed detection frameworks. We emphasize that the examples presented are primarily illustrative. Given the theoretical focus of our work, the simulations are not extensive but aim to demonstrate the fundamental principles and potential applications of our approaches. The proposed method is implemented in MATLAB, and the optimization (via SDP method) is done through the CVX solver \cite{grant2009cvx}.

\subsection{Simulation Study}

In this simulation study, we evaluate two performance metrics: 1) the detection delay simulated by assuming the change happens at $\nu=1$; and 2) the average run length $\mathbb E_\infty[T]$ for the test (stopping time) $T$. {We consider a moderate network size of $n=50$, and choose $\bm\sigma^\pre$ and $\bm\sigma^\post$ such that $\operatorname{Ham}({\bm{\sigma}^{\pre}},{\bm{\sigma}^{\post}}) = 2$, corresponding to a scenario where only a small number of nodes change their community labels. The distribution parameter is set as $p = a\log(n)/n$ for some $a>0$ and $\zeta=0.1$.} The private parameter is set as $\epsilon$. We consider two scenarios: (1) Case 1 ($a=5,\epsilon=1.5$) and (2) Case 2 ($a=6,\epsilon=1$). {These settings are chosen such that the sufficient condition for exact recovery in Theorem~\ref{thm:private_threshold_condition_one_time_instance} is satisfied in Case 1 but violated in Case 2. This allows us to compare detection performances under both a favorable scenario with exact recovery guarantees and a more challenging scenario where exact recovery after perturbation is not assured.} It is worth mentioning that the KL divergence between the post- and pre-change distributions after graph perturbation in Case 1 is larger, meaning it is easier to be detected. We also tried multi-window with $w=10$ (see Footnote \ref{myfootnote}) in the more challenging Case 2 to enhance the detection performance. 

We demonstrate the sample trajectories and the trade-off curve of delay vs. ARL, for the detection procedure $T_{\rm L}$ in \eqref{eq:stop_time} under LDP settings, in Fig. \ref{fig:stat} for Case 1 (left) and Case 2 (right). We observe that the detection statistics under Case 1 are more stable under the pre-change regime (mostly below zero) and become mostly positive after the change, yielding a relatively small delay. On the other hand, the detection statistic for Case 2 is less stable and exhibits several upward trends before the true change point. These could be caused by the stringent privacy constraint, as a smaller $\epsilon$ value in this scenario will lead to significant alterations in data distributions, and the exact recovery sufficient condition \eqref{eq:recov_condi} is violated. 

For completeness, we also include the detection results of the procedure $T_{\rm C}$ in Fig. \ref{fig:stat2}. It is obvious that the performance of the detection procedure here is better than the procedure under LDP settings, as demonstrated by the detection statistics and small detection delay under both cases. This is expected and consistent with our intuition since the LDP constraint is more stringent than the CDP requirement. Under the LDP constraint, we have to privatize the data at the source, which could significantly reduce the divergence between post- and pre-change distributions, especially for small $\epsilon$ values. On the other hand, under the CDP setting, we calculate the detection statistics mostly under the true graph observations. Moreover, the Laplace noise added to the detection statistics is relatively small compared with the recursive statistics $S_t$, thus the detection performance could be much better. We highlight that the difference in performance is caused by the distinct nature of LDP and CDP mechanisms, which do not necessarily imply which detection procedure is better than the other.

{We further comment on the non-private baseline for change detection, which applies the classic CUSUM test directly to the original data without any privacy protection. As commented in Remark \ref{rem:non-private}, the non-private CUSUM method provides the information-theoretic lower bound, $\log\gamma / I_0(1+o(1))$, on the detection delay and is guaranteed to incur a smaller delay.
In our two scenarios, since the KL divergence between post- and pre-change models $I_0$ is relatively large, the corresponding theoretical and empirical delays are both small (approximately one for the ARL values smaller than $10^4$). Therefore, we omit this baseline from the plots. Nonetheless, as seen in Fig. \ref{fig:stat} and Fig. \ref{fig:stat2}, the detection delays under both LDP and CDP settings (particularly in Case 1) are generally less than 4, indicating highly efficient detection performance even with privacy constraints.}

\begin{figure}[ht!]
    \centering
 \begin{tabular}{cc}
   \includegraphics[width=0.45\linewidth]{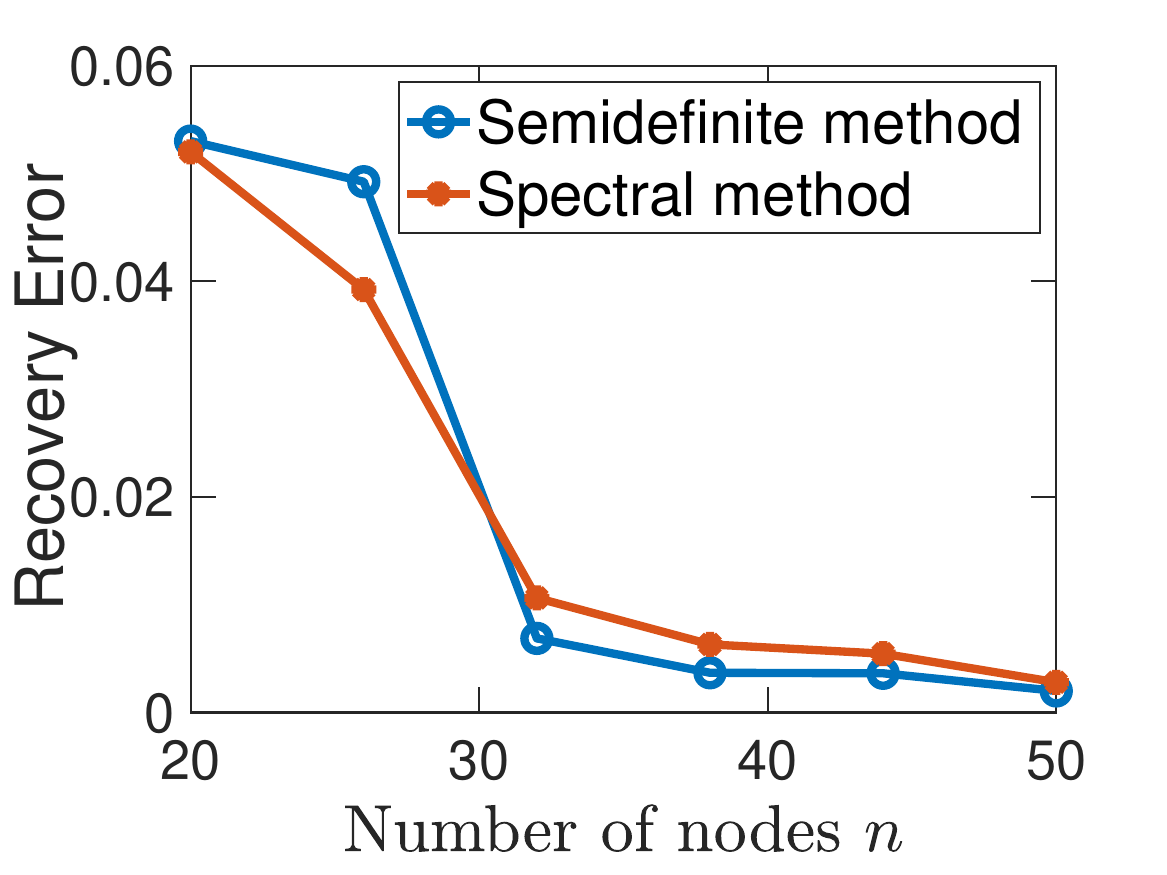}   & \includegraphics[width=0.45\linewidth]{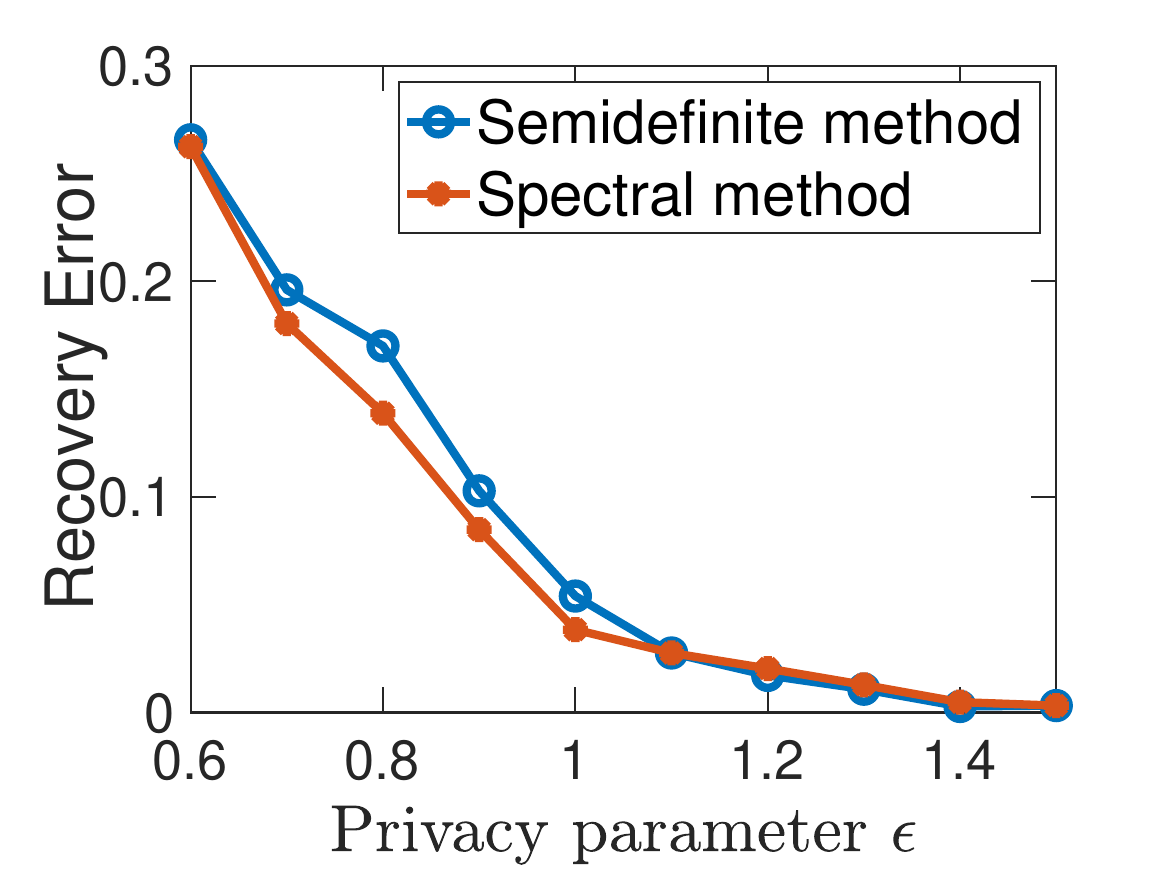} 
 \end{tabular}   
    \caption{Comparison of community label recovery errors for the semidefinite relaxation method and the spectral method. Each plot shows the average recovery error over 50 repetitions. The distributional parameters are fixed at $p=0.8$ and $\zeta=0.1$. In the left plot $\epsilon=1.5$ is held constant, while in the right plot, the number of nodes is fixed at $n=50$. The true community label vector $\bm\sigma$ consists of half $+1$ entries and half $-1$ entries.}
    \label{fig:recovery}
    \vspace{-0.15in}
\end{figure}

\begin{remark}[Recovery Error]
As a sanity check, we evaluate the community label recovery error of the semidefinite relaxation method and compare it with the spectral method proposed in \cite{dhara2022spectral}. Specifically, we apply the privacy mechanism described in Section~\ref{sec:private} via graph perturbation, then estimate the underlying community label vector using both the semidefinite relaxation formulation in \eqref{eqn:SDP_relaxation_asymmetric} and the spectral method. The recovery error is measured by the normalized Hamming distance,  $\operatorname{Ham}({\bm{\sigma}},{\hat{\bm{\sigma}}})/n$, where $\bm{\sigma}$ is the true community label and $\hat{\bm{\sigma}}$ is the estimated value. As shown in Fig.~\ref{fig:recovery}, both methods exhibit comparable recovery performance across various scenarios, with node size $n$ from 20 to 50 and privacy budget $\epsilon$ from 0.6 to 1.5. The parameter values are selected such that the exact recovery condition in \eqref{eq:recov_condi} is satisfied. Based on this observation, we adopt the semidefinite relaxation method in our work. It is worthwhile mentioning that our proposed detection frameworks are flexible, and they can readily incorporate the spectral method as an alternative for estimating post-change community structures as well.
\end{remark}

\subsection{Real Case Study}

We also implement the detection algorithm on a real-world agricultural trade dataset\footnote{This dataset is available at {\it https://www.fao.org/faostat/en/\#data/TM}.} and a U.S. air transportation network dataset\footnote{This dataset is available at \href{https://www.transtats.bts.gov/DL_SelectFields.aspx?gnoyr_VQ=GDL&QO_fu146_anzr=Nv4}{here}.}.

\paragraph{Agricultural Trade Dataset} The dataset includes the quantity and value of all food and agricultural products imported/exported annually by all the countries in the world. We select annual data from the years 1991 to 2015 and $n = 50$ countries with the most significant importing/exporting relationship. For each country pair $(i,j)$, we let $A_{i, j}= 0$ if there is no import-export relation between them, $A_{i, j} = 1$ if the total trade column of the top $4$ agricultural products exceeds a pre-specified threshold, and $A_{i, j}=-1$ otherwise. Consequently, each node corresponds to a country, and each edge represents the export-import relationship between countries. We use the data of the year 1991 to estimate the pre-change community label, parameters $p$, and $\zeta$. 

\begin{figure*}[ht!]
    \centering
    \begin{tabular}{cccc}
     \includegraphics[width=0.22\textwidth]{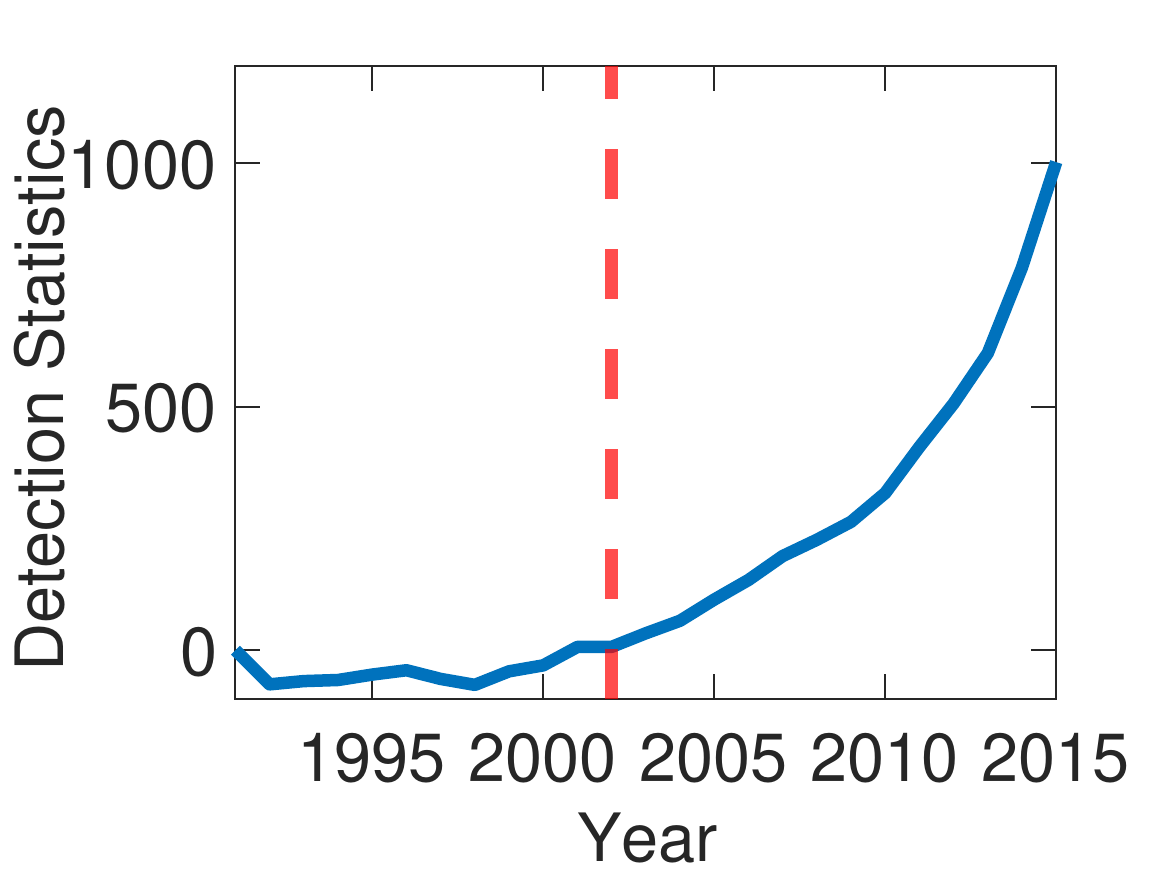} & \includegraphics[width=0.22\textwidth]{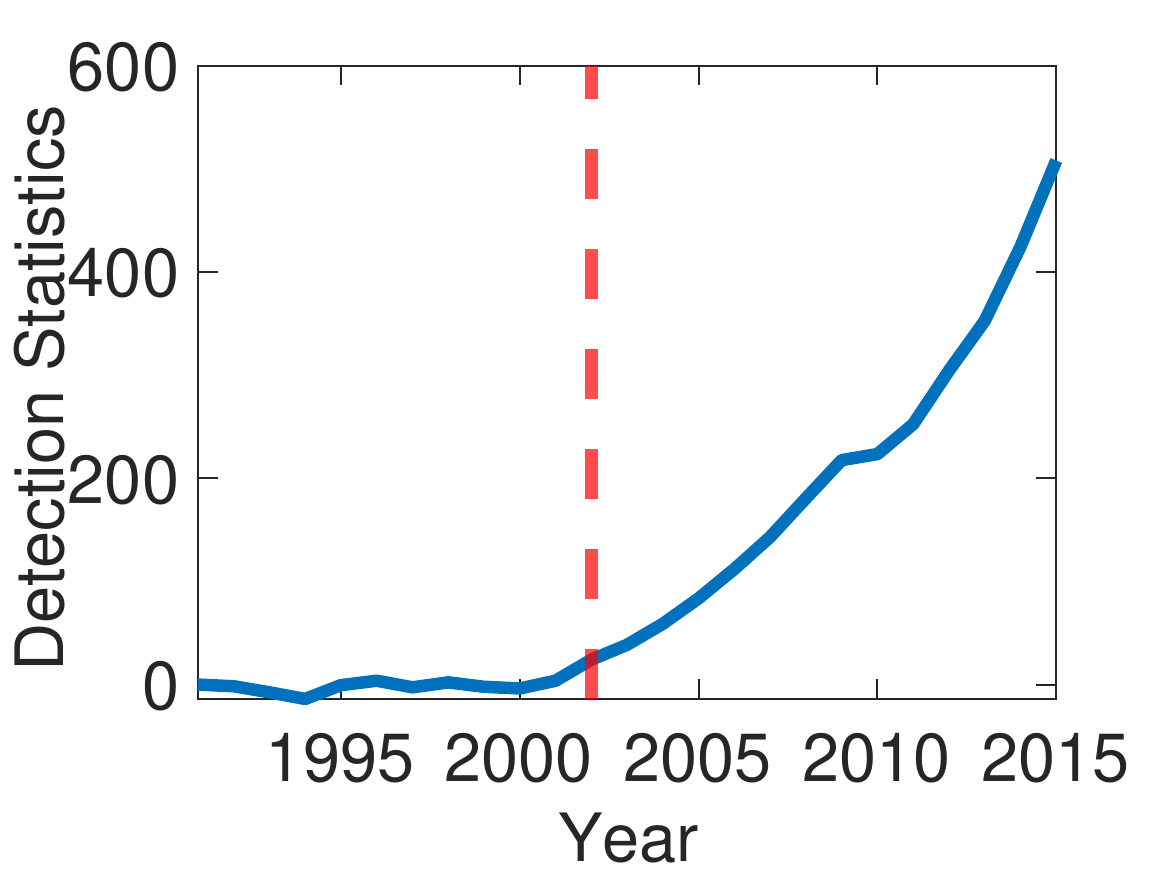} & \includegraphics[width=0.22\textwidth]{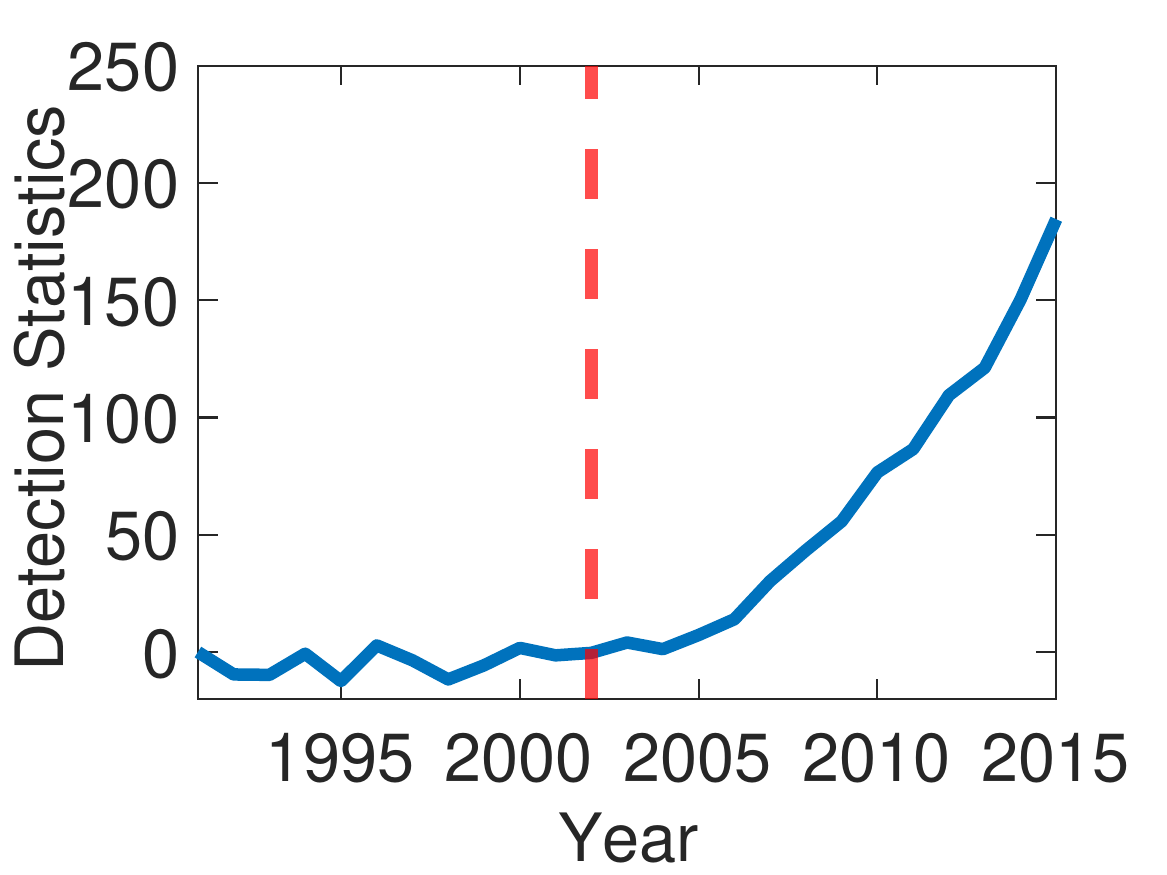} & \includegraphics[width=0.22\textwidth]{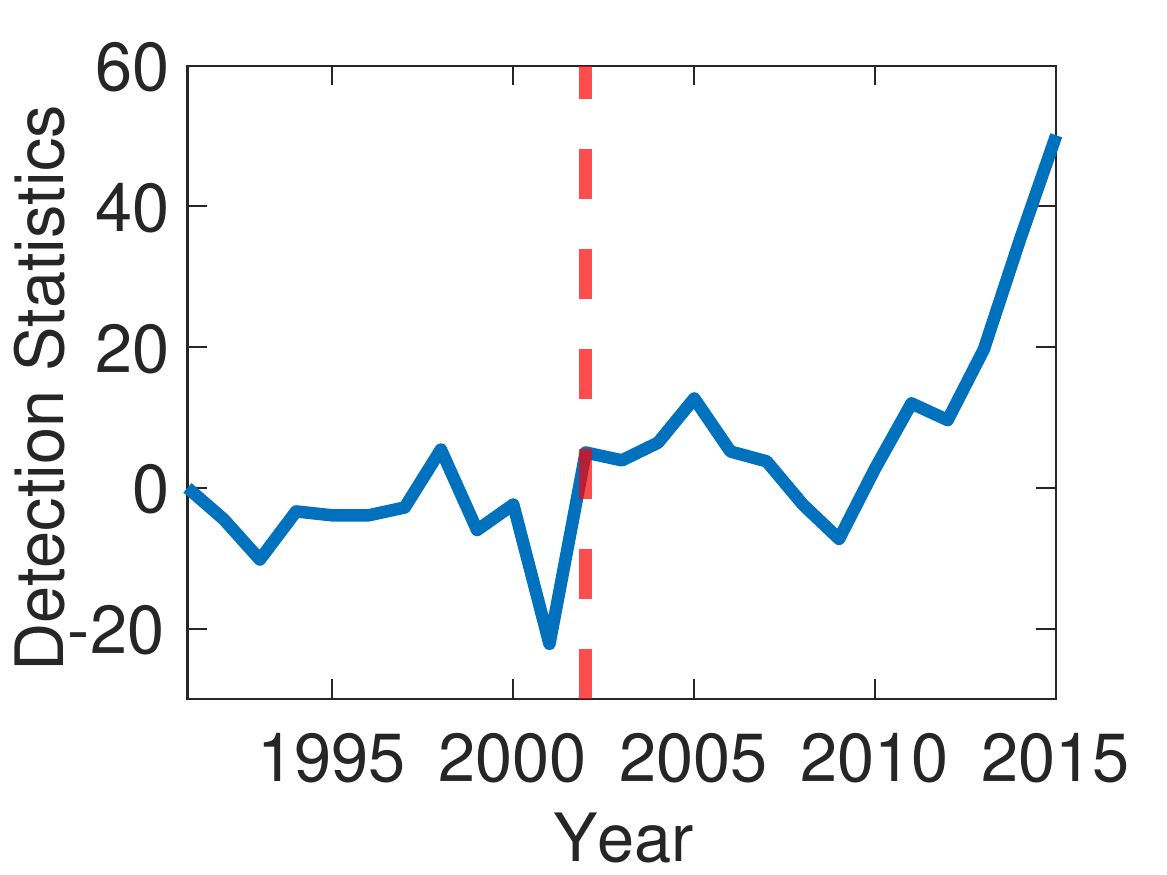} \\
       Raw Network ($\epsilon=+\infty$) & $\epsilon=1.5$ & $\epsilon=1$ &  $\epsilon=0.8$ 
    \end{tabular}   
    \caption{Trajectory of the detection statistics under LDP on the agricultural trade dataset, for the raw network (no perturbation), the privatized/perturbed network with $\epsilon=1.5$, 1, and 0.8, representing different privacy levels.}
    \label{fig:real}
\end{figure*}

Since private detection under LDP settings is more challenging than in CDP for the same privacy parameter $\epsilon$, we visualize the detection statistics from \eqref{eq:stat-wlcusum} under LDP in Fig. \ref{fig:real}. The detection statistics successfully identify the change-point in 2002, as noted in previous studies \cite{wang2023multilayer}. This year 2002 marked a significant turning point in global commodity prices, which likely influenced trade patterns of various countries \cite{wang2023multilayer}. 
Furthermore, from the trajectory of detection statistics after applying randomized perturbation with {\it different levels of $\epsilon$ values}, we can see the {\it tradeoff between privacy and utility} (detection effectiveness): stronger privacy constraints (lower $\epsilon$ values) make it more challenging to detect the change.

\paragraph{U.S. Air Transportation Network Dataset} We use monthly data from February 2018 to July 2022 (54 months). Each node represents an airport, and each edge encodes the presence of direct flights between two airports. We select $n=50$ airports with the highest combined numbers of departing and arriving flights. 
For each airport pair $(i,j)$, we let $A_{i, j}= 0$ if there is no direct flight recorded, $A_{i, j} = 1$ if the number of airlines operating direct flights between them exceeds a pre-specified threshold, and $A_{i, j}=-1$ otherwise. Consequently, the resulting graph represents the density of airline connectivity among major U.S. airports. 
Data from February 2016 to June 2018 are used as the historical data for estimating the pre-change community structure and parameters, with the rest being the dataset for conducting change-point detection. 

In this real data example, it is important to note that the distribution parameters $p$ and $\zeta$ may also change over time. Therefore, in addition to computing detection statistics under the assumption of constant $p$ and $\zeta$, we also implement a version that allows these parameters to change. Specifically, we estimate their post-change values jointly with $\bm{\hat\sigma}_{t-1}$ using the maximum likelihood approach detailed in Appendix~\ref{sec:llr}.

\begin{figure}[ht!]
    \centering
    \begin{tabular}{cc}
     \includegraphics[width=0.22\textwidth]{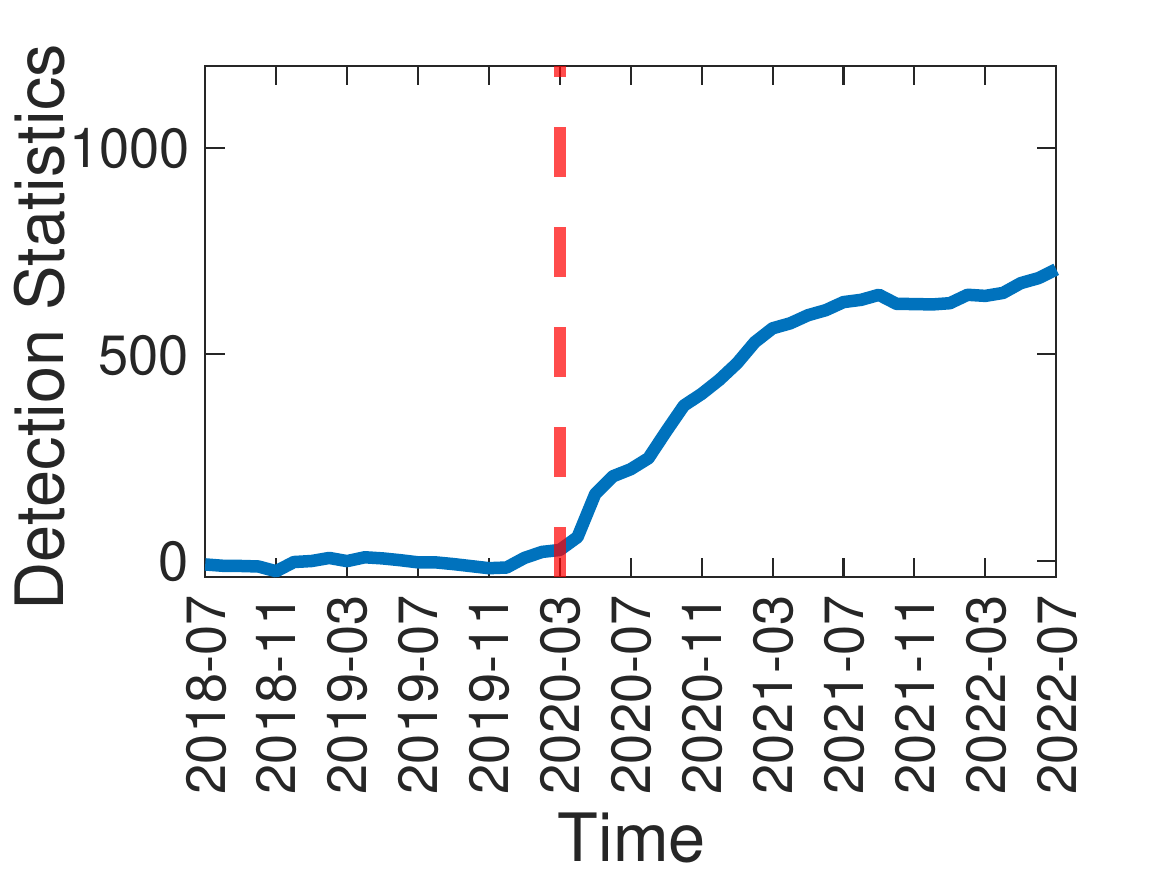} & \includegraphics[width=0.22\textwidth]{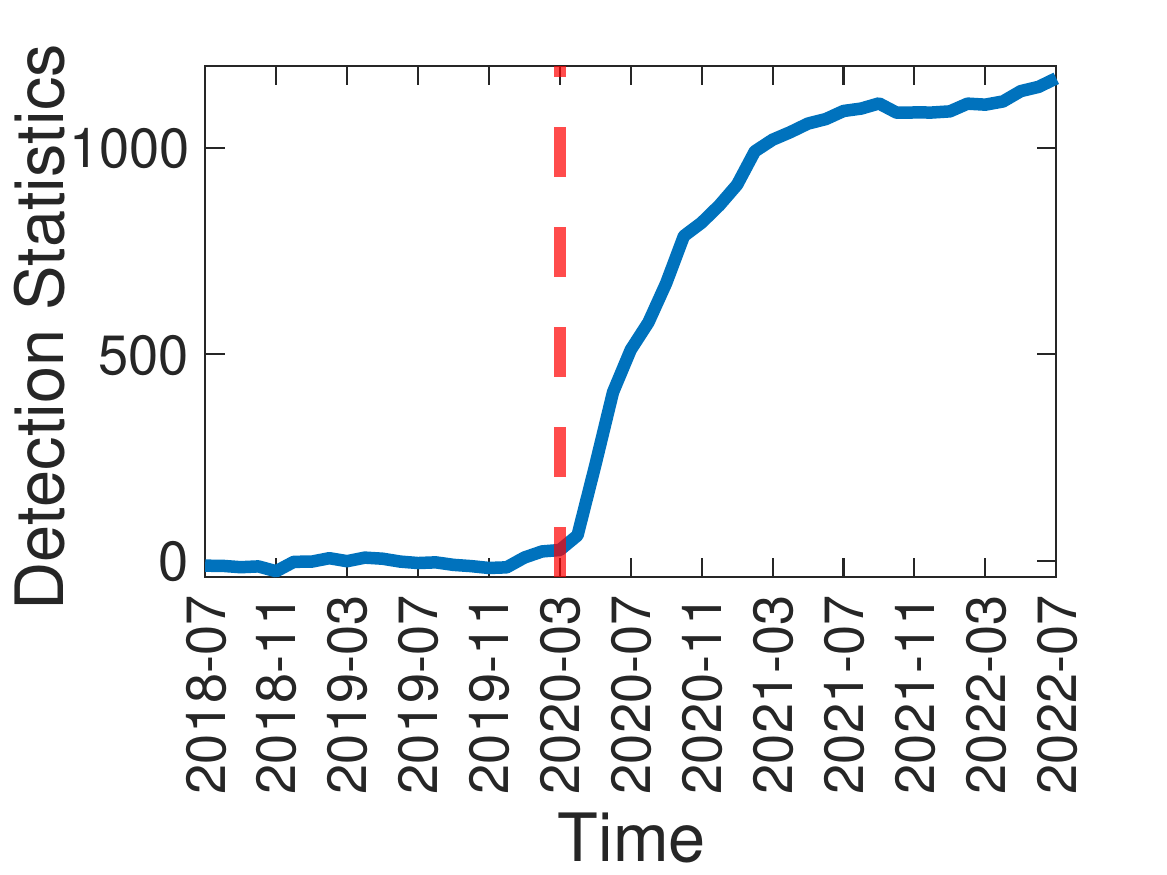} \\
     \makecell{Raw Network ($\epsilon=+\infty$); \\ Assume unchanged $p,\zeta$}   & \makecell{Raw Network ($\epsilon=+\infty$); \\ Assume changed $p,\zeta$}  \\
     \includegraphics[width=0.22\textwidth]{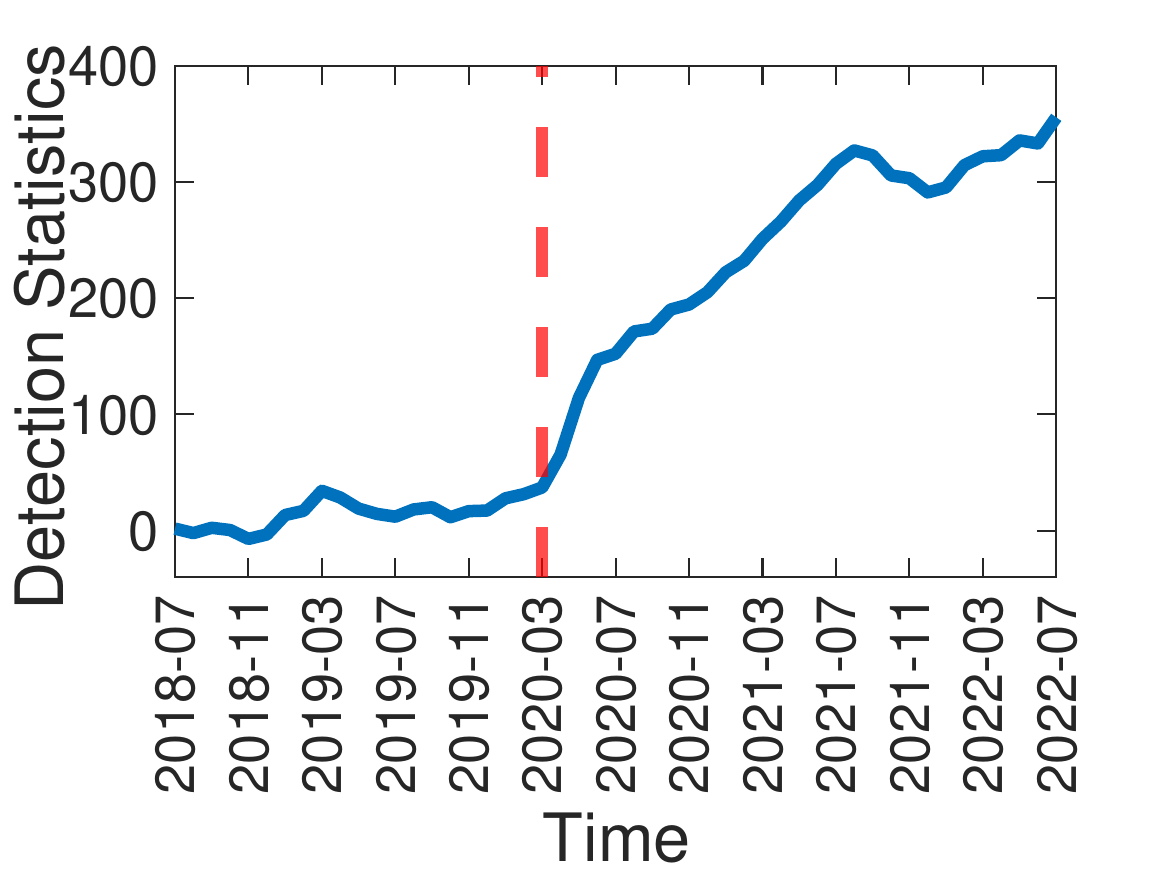} & \includegraphics[width=0.22\textwidth]{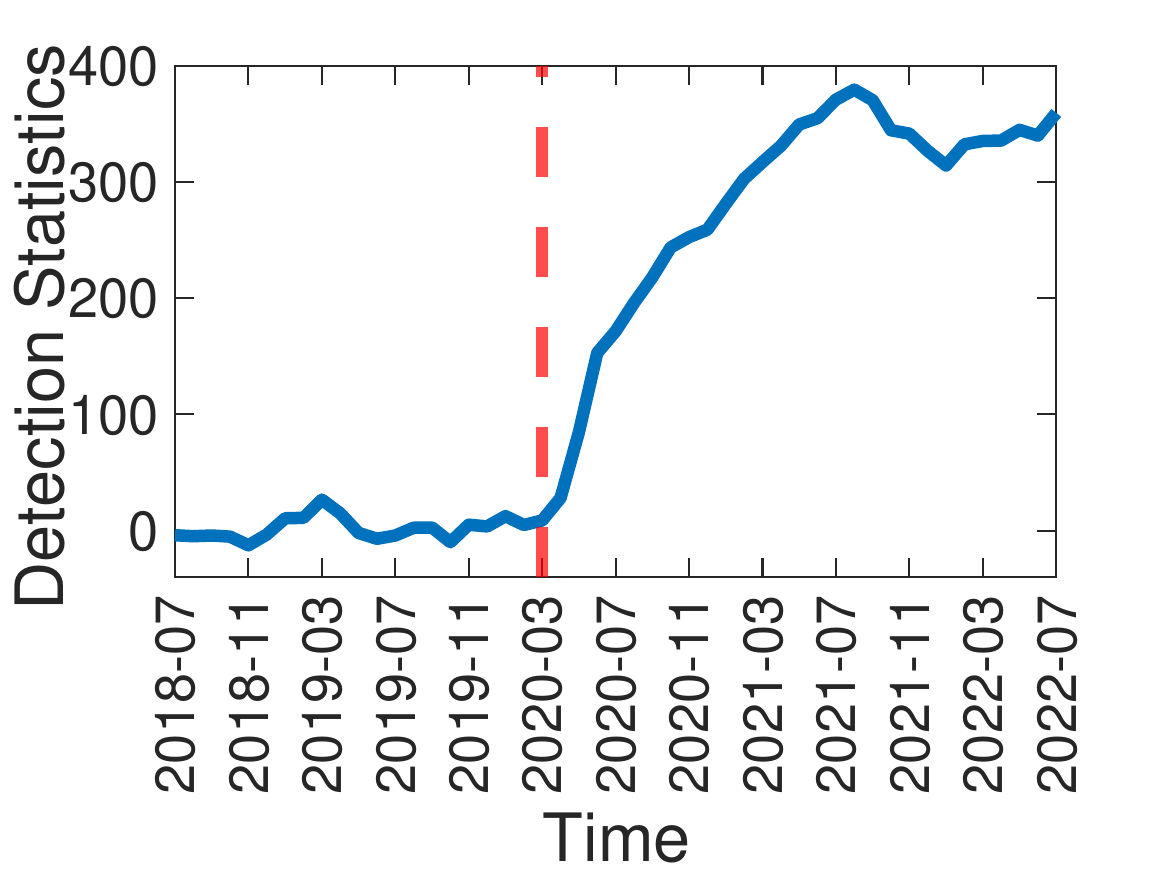} \\
       \makecell{ $\epsilon=1.5$; \\ Assume unchanged $p,\zeta$} &  \makecell{$\epsilon=0.8$; \\ Assume changed $p,\zeta$} 
    \end{tabular}   
    \caption{Trajectory of the detection statistics under LDP on the U.S. air transportation dataset, for the raw network (no perturbation, corresponding to $\epsilon=\infty$) and the privatized/perturbed network (with $\epsilon=1.5$). For each case, we report two versions of the detection statistic: one assuming the parameters $p$ and $\zeta$ remain unchanged after the change-point, and another allowing them to change, in which case the post-change values of $p$ and $\zeta$ are also estimated during detection.}
    \label{fig:real-air}
\end{figure}

We evaluate the more challenging private change detection under local differential privacy (LDP) constraints. Fig.~\ref{fig:real-air} shows the trajectory of the detection statistics. The statistic successfully identifies a change-point in March 2020, coinciding with the disruption in air travel caused by the COVID-19 outbreak, also consistent with findings in previous studies~\cite{wang2023multilayer}. 
From the trajectory of detection statistics after applying randomized perturbation, we can see that the detection becomes less effective due to privacy requirements, as evidenced in the slower increase in detection statistics. Furthermore, the detection statistics that take into account the potential change in $p$ and $\zeta$ improve detection performance in both raw and privatized settings, as reflected by the faster growth of the corresponding detection statistics.

\section{Conclusion and Discussion}

{This work has studied online change detection of dynamic community structures under the censored block model. Detection algorithms have been presented under two differential privacy settings. The proposed methods are applicable to a wide range of network types and sizes. Beyond network data, the detection framework can also be extended to other data types, such as spatio-temporal data, where spatial dependencies are captured by an underlying connectivity graph.}
Building on these findings, we can identify several avenues for future research. Firstly, we may consider the detection of general changes, including changes in both the community label vector $\bm\sigma$ and the parameters $p,\zeta$ of the CBM model. Changes in $p$ and $\zeta$ will represent practical scenarios in which the local connectivity within a community is strengthened or weakened after the change. Secondly, we leave the more involved discussion on the impact of the multi-view stochastic model \cite{zhang2024community} on our joint detection and estimation scheme, as well as the corresponding impact of the window size $w$, for future investigation. We conjecture that there is an optimal tradeoff between the detection performance and computational/memory efficiency by carefully choosing the window size, especially in the non-asymptotic regime where the network size $n$ is limited. Finally, this work opens numerous opportunities for exploring more advanced differential privacy mechanisms to further enhance detection performance.


\bibliographystyle{plain}
\bibliography{myreferences}


\appendices

\section{More Technical Details and Discussions}

\subsection{Tradeoff between Privacy and Separation Condition for Graph Perturbation}

From Theorem \ref{thm:private_threshold_condition_one_time_instance}, we have a sufficient condition for exact recovery being 
\begin{align*}
    a \big(\sqrt{1-\zeta} - \sqrt{\zeta} \big)^{2} > \frac{1}{\alpha} \times \left(\frac{e^{\epsilon} + 1}{e^{\epsilon} - 1} \right), 
\end{align*}
where $\alpha=1-\frac{1}{\sqrt{n}}$. 
  As can be shown in Fig. \ref{fig:phase_transition_mechanisms_app}, the analysis of the lower bounds on the parameter \( a \) for \( \alpha \), specifically \( \alpha = 1 - \frac{1}{\sqrt{n}} = 1 - o(1) \).  This functional form is desirable as it allows for more flexibility and potentially less stringent requirements on the parameter \( a \) while maintaining the desired privacy level defined by \( \epsilon = c \log(n) \). Thus, in scenarios where achieving an optimal balance between privacy and parameter constraints is crucial, opting for \( \alpha = 1 - o(1) \) would be a more strategic choice.

\begin{figure}[ht!]
\centering
	\centering
	{\includegraphics[width=0.75\columnwidth]{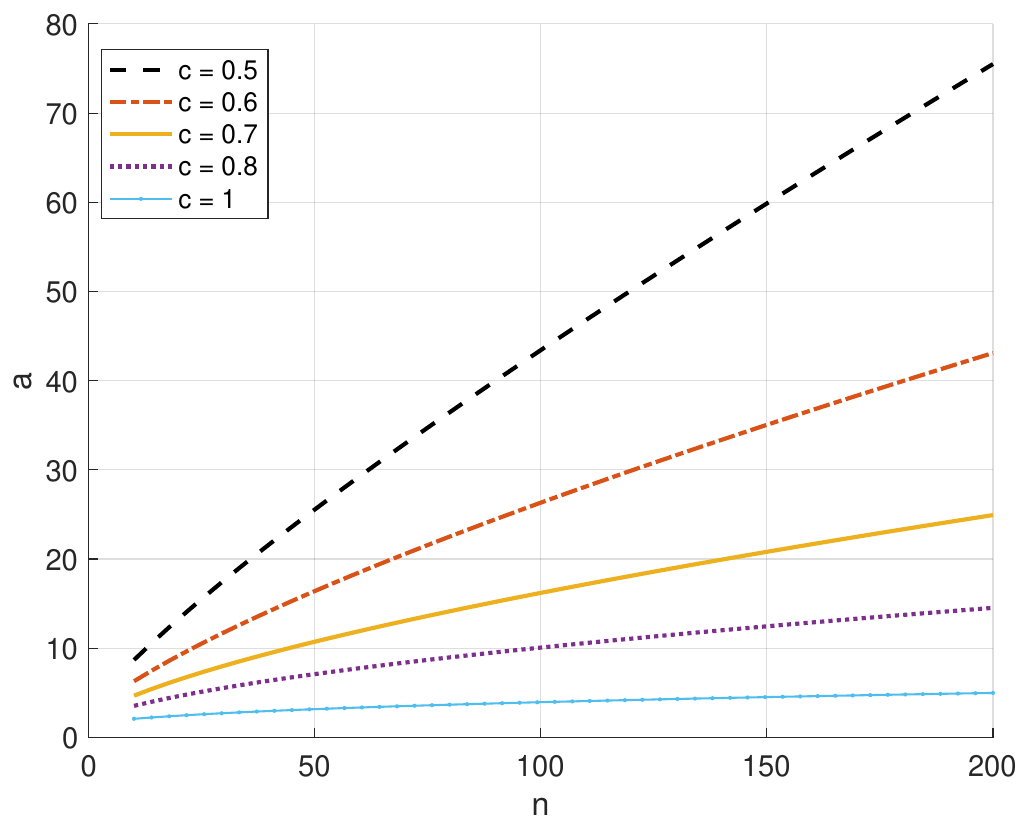}}	
	\caption{\small{\textit{Impact of the Privacy Scaling Coefficient \( c \) on the Required Lower Bound of \( a \)}: This plot illustrates the relationship between the lower bound of the parameter \( a \) and the number of samples \( n \), under the privacy condition \( \alpha = 1 - \frac{1}{\sqrt{n}} \) and the privacy budget \( \epsilon = c \log(n) \), where \( c > 0 \). The analysis reveals that an increase in the value of \( c \) leads to a reduction in the minimum required value of \( a \), indicating a trade-off between the privacy scaling coefficient and the parameter \( a \).}}
    \label{fig:phase_transition_mechanisms_app}
\end{figure}


\subsection{Derivation of the Log-Likelihood and the KL Divergence}\label{sec:llr}

For a given graph $G_t$ at time $t$ with adjacency matrix $\bA_t$, denote $E_{t,1}:=\#\{i<j:\sigma_i A_{t,i,j}\sigma_j=1\}$ as the number of edges within the same community with edge label $1$, or in different communities with edge label $-1$; similarly $E_{t,-1}:=\#\{i<j:\sigma_iA_{t,i,j}\sigma_j=-1\}$. Let $E_t = E_{t,1} + E_{t,-1}$ be the total number of edges at time $t$. Then we have the log-likelihood function is
\[
p(\mathbf{A};\bm\sigma,p,\zeta)= [p(1-\zeta)]^{E_{t,1}}[p\zeta]^{E_{t,-1}}(1-p)^{{n \choose 2} - E_t}.
\]
Note that 
\[
\begin{aligned}
& \frac12\bm\sigma^T \mathbf{A}_t \bm\sigma = E_{t,1} -  E_{t,-1} = 2 E_{t,1} - E_t \\ 
 \Rightarrow &   E_{t,1} = \frac14 \bm\sigma^T \mathbf{A}_t \bm\sigma + \frac12 E_t. 
\end{aligned}
\]
Substitute this into the log-likelihood we obtain
\begin{align}\label{eq:app_ll}
  & \log p(\mathbf{A}_{t};\bm\sigma,p,\zeta)  = \left(\frac14 \log\frac{1-\zeta}{\zeta}\right)\bm{\sigma}^T \mathbf{A}_t \bm{\sigma} \nonumber \\ 
 & +  {n \choose 2} \log(1-p) + E_t \log \left(\frac{p}{1-p}\sqrt{\zeta(1-\zeta)}\right).  
\end{align}

\paragraph{Maximum Likelihood Estimator (MLE)}\label{app:mle}
By the log-likelihood function \eqref{eq:app_ll}, we have that the MLE for $\bm\sigma$ is $\hat{\bm{\sigma}}_{\text{MLE}} = \arg \max_{\bm{\sigma}} \{\bm{\sigma}^{T}\mathbf{A} \bm{\sigma}\}$. And the MLE for $p$ and $\zeta$ are $\hat p_{\text{MLE}}= E_t/{n \choose 2}$ and $\hat \zeta_{\text{MLE}}=\frac12 - \frac{\hat{\bm{\sigma}}_{\text{MLE}}^{T}\mathbf{A} \hat{\bm{\sigma}}_{\text{MLE}}}{4E_t}$, respectively.

\paragraph{KL Divergence} Based on \eqref{eq:app_ll}, the KL divergence between $\CBM(\bm\sigma,p,\zeta)$ and $\CBM(\bm\sigma',p,\zeta)$ can be calculated as
\begin{align}\label{eq:app_KL}
    I & = \bE_{\bm\sigma',p,\zeta}[\log p(\bA;\bm\sigma',p,\zeta) - \log p(\bA;\bm\sigma,p,\zeta) ] \nonumber \\
    & = \left(\frac12 \log\frac{1-\zeta}{\zeta}\right)p(1-2\zeta)\left[{n\choose 2} - C_{\bm\sigma,\bm\sigma'}\right],
\end{align}    
where $C_{\bm\sigma,\bm\sigma'}:=\sum_{i<j} {\sigma}_i {\sigma}_j {\sigma}_i' {\sigma}_j'$ represent the similarity level between the community label vectors $\bm\sigma$ and $\bm\sigma'$ and it satisfies $|C_{\bm\sigma,\bm\sigma'}|\leq {n \choose 2}$.

\subsection{Analysis of the Graph Perturbation Mechanism}\label{sec:pert_calc}

In this subsection, we provide analysis for the Ternary Randomized Response Mechanism used in Section \ref{sec:ldp-method}. 
In this mechanism, when \( A_{i,j} = x \) where \( x \in \{0, +1, -1\} \), the mechanism outputs \( \tilde{A}_{i,j} = y \), where \( y \in \{0, +1, -1\} \) with probability \( q_{x}(y) \triangleq \operatorname{Pr}(\tilde{\bA}_{i,j} = y | \bA_{i,j} = x) \). We set $q_{x}(y)=c_2$, $\forall x\neq y$, and $q_{x}(y)=c_1=1-2c_2$, $\forall x=y$. We denote $\tilde{p}$ as the distribution of each edge $\tilde{A}_{i,j}$ for the pair $(i,j)$ such that $\bm\sigma_i=\bm\sigma_j$ (in the same community), and denote $\tilde{q}$ as the distribution of each edge $\tilde{A}_{i,j}$ for the pair $(i,j)$ such that $\bm\sigma_i\neq \bm\sigma_j$ (in different communities). We have the following: 
\begin{align*}
    \tilde{p}({+1}) &  = c_{1} p (1 -\zeta) + c_{2} p \zeta + c_{2} (1-p), \nonumber \\ 
    \tilde{p}({-1}) & = c_{2} p (1-\zeta) + c_{1} p \zeta + c_{2} (1-p), \nonumber \\ 
  \tilde{p}({0}) & = c_{2} p (1-\zeta) + c_{2} p \zeta + c_{1} (1-p), \nonumber \\ 
   \tilde{q}({+1}) & =   c_{1} p \zeta + c_{2} p (1 - \zeta) + c_{2} (1-p), \nonumber \\ 
   \tilde{q}({-1}) &  = c_{2} p \zeta + c_{1} p (1-\zeta) + c_{2} (1-p), \nonumber  \\ 
    \tilde{q}({0}) & = c_{2} p (1-\zeta) + c_{2} p \zeta + c_{1} (1-p). 
\end{align*}

We note that after the above ternary-randomized mechanism, the distribution of $\tilde{\bA}$ is still a CBM model. More specifically, for the original CBM model $\CBM(\bm\sigma,p,\zeta)$, the perturbed $\tilde\bA$ follows the CBM model $\CBM(\bm\sigma,\tilde p,\tilde \zeta)$  with the same community label vector $\bm\sigma$ and new parameters
\begin{align}\label{eq:para_perturb}
\tilde p &= 1-c_1(1- p) - c_2 p = 2c_2 + p(1-3c_2), \nonumber \\
\tilde\zeta &=  \frac{c_2(1-p\zeta) + c_1p\zeta}{1-c_1(1- p) - c_2 p} = \frac{c_2 + p\zeta(1-3c_2)}{2c_2 + p(1-3c_2)}.
\end{align}

We further analyze the change in the KL divergence between the pre- and post-change distributions after the perturbation. The KL divergence for the raw data distribution (before perturbation) can be derived as (see \eqref{eq:app_KL}), 
\begin{align}
&I(\bm{\sigma}^{\text{pre}};\bm{\sigma}^{\text{post}})=\frac12 \log\frac{1-\zeta}{\zeta} \cdot p(1-2\zeta)\left[{n\choose 2} - C_{\bm{\sigma}^{\text{pre}},\bm{\sigma}^{\text{post}}}\right].
\end{align}
After the perturbation, we have similarly
\begin{align}
&\tilde I(\bm{\sigma}^{\text{pre}};\bm{\sigma}^{\text{post}}) =\frac12 \log\frac{1-\tilde\zeta}{\tilde\zeta}\cdot \tilde p(1-2\tilde\zeta)\left[{n\choose 2} - C_{\bm{\sigma}^{\text{pre}},\bm{\sigma}^{\text{post}}}\right].
\end{align}
Substitute the parameter $\tilde p$ and $\tilde \zeta$ in \eqref{eq:para_perturb}, and using the setting that $c_1=1-2c_2$, we have 
\begin{align}\label{eq:tentative}
&\left(\log\frac{1-\tilde\zeta}{\tilde\zeta}\right)\tilde p(1-2\tilde\zeta) \nonumber \\ 
 = &\left(\log\frac{c_2 + p(1-\zeta)(1-3c_2)}{c_2 + p\zeta(1-3c_2)}\right)\cdot p(1-2\zeta)(1-3c_2).    
\end{align}

We note that $\log\frac{c_2 + p(1-\zeta)(1-3c_2)}{c_2 + p\zeta(1-3c_2)}$ is a positive and decreasing function of $c_2$ when $\zeta<\frac12$ and a negative and increasing function when $\zeta>\frac12$, the same is the function $p(1-2\zeta)(1-3c_2)$. Moreover, when $c_2=0$ we have the KL divergence remaining unchanged. Combining these together, we have \eqref{eq:tentative} is a nonnegative and decreasing function of $c_2$, with the largest value obtained when $c_2=0$ and the largest value equal to $I(\bm{\sigma}^{\text{pre}};\bm{\sigma}^{\text{post}})$. This demonstrates the {\it tradeoff between detection and privacy}. This being said, for a given privacy constraint, we may select the smallest value of $c_2$ value that satisfies this constraint; such choice ensures the best detection performance after the perturbation. 

Note that, under $\epsilon$-edge LDP, we require $c_2/c_1\in[e^{-\epsilon},e^{\epsilon}]$, which yields the smallest possible value for $c_2$ being $c_2=\frac{1}{e^\epsilon +2}$, which is the value we choose in \ref{sec:ldp-method}.

\section{Fundamentals of Technical Analysis}

\subsection{Ancillary Lemmas}
We provide auxiliary results and definitions that are used throughout the proofs.

\begin{lemma}[\cite{hajek2016achieving}] \label{lemma_tail_bound_censored} Let $\{X_{i}\}_{i=1}^{m}$ be a collection of independent  random variables, where $p=a \frac{\log n}{n}$, $m = \rho n + o(n)$ for some $\rho>0$ as $n\to\infty$, and each $X_{i} \overset{i.i.d.} \sim p (1-\zeta) \delta_{+1} + p \zeta \delta_{-1} + (1-p) \delta_{0}$. Then, 
\begin{align*}
\operatorname{Pr} \bigg(\sum_{i=1}^{m} X_{i} \leq  k_{n} \bigg) \leq n^{- \rho  a  (\sqrt{1 - \zeta} - \sqrt{\zeta})^{2} + o(1)},
\end{align*}
where $k_{n} = (1 + o(1)) \frac{\log(n)}{\log(\log(n))}$. 
\end{lemma}



We present a high probability result for upper bounding the spectral norm of $\mathbf{A} - \mathds{E}(\mathbf{A})$.
\begin{lemma} (Matrix concentration inequality for dense regime \cite{hajek2016achieving}) \label{lem:lem-spectral-concentrate}
For a matrix $\mathbf{A} \sim \CBM(\bm\sigma,p,\zeta)$ and $p = \Omega (\log(n)/n)$, we have for any $c>0$, there exists $c'>0$ such that
\begin{align*}
   \operatorname{Pr}\bigg( \|\mathbf{A} - \mathds{E}(\mathbf{A})\| \leq c' \sqrt{np}\bigg) \geq 1-n^{-c}.
\end{align*}
\end{lemma}

\subsection{Deriving Recovery Condition for Exact Recovery}\label{app:recovery}

In this subsection, we review the derivation of necessary conditions, $a (\sqrt{1 - \zeta} - \sqrt{\zeta})^{2} > 1$, for exact community recovery using SDP estimate \eqref{eqn:SDP_relaxation_asymmetric}, under non-private scenarios, as shown in \cite{hajek2016achieving}. This derivation forms the foundational basis for our subsequent proofs. We also shed light on the necessary condition when multiple graphs are used. 

First of all, for the semidefinite program in \eqref{eqn:SDP_relaxation_asymmetric}, the Lagrangian function is written as follows,
\begin{align*}
    & \mathcal{L}({\mathbf{A}}, \mathbf{Y}, \mathbf{S}, \mathbf{D}) = \operatorname{tr}({\mathbf{A}} \mathbf{Y}) + \operatorname{tr}(\mathbf{S} \mathbf{Y}) - \operatorname{tr}(\mathbf{D} (\mathbf{Y}-\mathbf{I})),
\end{align*}
where $\mathbf{S} \succcurlyeq \mathbf{0}$ and $\mathbf{D} = \operatorname{diag}(d_{i})$. Then, 
\begin{align*}
    \nabla_{\mathbf{Y}} \mathcal{L} = {\mathbf{A}} + \mathbf{S} - \mathbf{D} = \mathbf{0}. 
\end{align*}
In order to satisfy the first order stationary condition, we must satisfy the (deterministic) condition
\begin{align*}
    \mathbf{S}^{*} = \mathbf{D}^{*} - {\mathbf{A}},  
\end{align*}
 From the KKT conditions, we have the following, 
\begin{align}
     & \mathbf{D}^{*} \mathbf{Y}^{*} = \mathbf{D}^{*} \mathbf{I} \Rightarrow \operatorname{tr}(\mathbf{D}^{*} (\mathbf{Y}^{*} -\mathbf{I})) = 0, \nonumber \\
     &\mathbf{S}^{*} {\bm{\sigma}}^{*} = \mathbf{0} \Rightarrow  \operatorname{tr}(\mathbf{S}^{*} \mathbf{Y}^{*}) = 0 , \label{eqn:null_space_of_S}
\end{align}
where in Eqn. \eqref{eqn:null_space_of_S}, ${\bm{\sigma}}^{*}$ is the null space of $\mathbf{S}^{*}$. In order to ensure that $\mathbf{Y}^{*}$ is the unique solution, we require that $\lambda_{2}(\mathbf{S}^{*}) > 0$, i.e., the second smallest eigenvalue of $\mathbf{S}$ is positive. This comes from the rank-nullity theorem, i.e., 
\begin{align*}
    \text{rank}(\mathbf{S}^{*}) + \text{null}(\mathbf{S}^{*}) = n \Rightarrow \text{rank}(\mathbf{S}^{*}) = n -1. 
\end{align*}
To this end, we have the following,
\begin{align*}
      \operatorname{tr}({\mathbf{A}} \mathbf{Y}) & \leq \mathcal{L}({\mathbf{A}}, \mathbf{Y}^{*}, \mathbf{S}^{*}, \mathbf{D}^{*})  \nonumber \\
     & =  \operatorname{tr}((\mathbf{S}^{*} - \mathbf{D}^{*} + {\mathbf{A}}) \mathbf{Y} ) +  \operatorname{tr}(\mathbf{D}^{*} \mathbf{I})  \nonumber\\
     & =  \operatorname{tr}(\mathbf{D}^{*} \mathbf{I})  = \operatorname{tr}(\mathbf{D}^{*} \mathbf{Y}^{*}) \nonumber \\ 
     & = \operatorname{tr}((\mathbf{S}^{*} + {\mathbf{A}}) \mathbf{Y}^{*}) = \operatorname{tr}({\mathbf{A}} \mathbf{Y}^{*}). 
\end{align*}
Now, our goal is to prove that w.h.p. $\mathbf{S}^{*} \succcurlyeq 0$ with $\lambda_{2}(\mathbf{S}^{*}) > 0$. More specifically, we want to show that 
\begin{align*}
    \operatorname{Pr} \left[\inf_{\mathbf{x}: \|\mathbf{x}\| = 1, \mathbf{x} \perp \bm{\sigma}^{*}} \mathbf{x}^{T} \mathbf{S}^{*} \mathbf{x} > 0 \right] \geq 1 - o(1).
\end{align*}
\noindent Before we proceed, we note that 
\begin{align*}
    \mathds{E}[{\mathbf{A}}] & = (1- 2 \zeta) p (\mathbf{Y}^{*} - \mathbf{I}).
\end{align*}

 Now, for any $\mathbf{x}$ such that $\|\mathbf{x}\| = 1, \mathbf{x} \perp \bm{\sigma}^{*}$ (i.e., $\mathbf{x}^{T} \bm{\sigma}^{*}$ = 0), we proceed with the following set of steps,
   \begin{align}
    & \mathbf{x}^{T} \mathbf{S}^{*} \mathbf{x}   = \mathbf{x}^{T} \mathbf{D}^{*} \mathbf{x} -  \mathbf{x}^{T} \mathds{E}[{\mathbf{A}}] \mathbf{x} - \mathbf{x}^{T} ({\mathbf{A}} - \mathds{E}[{\mathbf{A}}]) \mathbf{x} \nonumber \\ 
      &  \geq  \mathbf{x}^{T} \mathbf{D}^{*} \mathbf{x} - (1- 2 \zeta) p \mathbf{x}^{T} \mathbf{Y}^{*} \mathbf{x} + (1 - 2 \zeta) p -   \|{\mathbf{A}} - \mathds{E}[{\mathbf{A}}] \| \nonumber \\ 
    & \overset{(a)} =   \mathbf{x}^{T} \mathbf{D}^{*} \mathbf{x} + (1 - 2 \zeta) p -   \|{\mathbf{A}} - \mathds{E}[{\mathbf{A}}] \| \nonumber \\ 
      & \geq \min_{i \in [n]} d_{i}^{*} + (1 - 2 \zeta) p -   \|{\mathbf{A}} - \mathds{E}[{\mathbf{A}}] \| \nonumber \\ 
      & \overset{(b)} \geq \min_{i \in [n]} d_{i}^{*} + (1 - 2 \zeta) p -  c' \sqrt{\log(n)}, \label{eq:single-condition}
      \end{align}       
where step (a) holds from the unique optimality condition in \eqref{eqn:null_space_of_S}. Note that $d_{i}$ is equal in distribution to $\sum_{i=1}^{n-1} X_{i}$, where $X_{i} \overset{i.i.d.}\sim p (1-\zeta) \delta_{+1} + p \zeta \delta_{-1} + (1-p) \delta_{0}$. Step (b) follows that $\|{\mathbf{A}} - \mathds{E}[{\mathbf{A}}] \| \leq c' \sqrt{\log(n)}$ with high probability by Lemma \ref{lem:lem-spectral-concentrate} when $p=a\frac{\log n}{n}$. Using Lemma \ref{lemma_tail_bound_censored} and the subsequent application of the union bound yields that with probability at least $1-n^{1-a(\sqrt{1 - {\zeta}}  - \sqrt{{\zeta}})^{2}+o(1)}$, we have 
\[
\min_i d_i^\ast \geq \frac{\log n}{\log \log n},
\]
and $\frac{\log n}{\log \log n} + (1-2\zeta)p - c'\sqrt{a\log n}>0$ for $n$ sufficiently large. Therefore if $a(\sqrt{1 - {\zeta}}  - \sqrt{{\zeta}})^{2}>1$, the model satisfies the exact recovery condition.

\begin{remark}[The Synergy of Combining Multiple Adjacency Matrices]
    In the case when we estimate $\bm\sigma$ via multiple observed adjacancy matrices, i.e., given $\bA_1,\ldots,\bA_w\sim \CBM(\bm\sigma,p,\zeta)$, we have the MLE for $\bm\sigma$ resulted from SDP is
\[
\begin{aligned}
\hat{\mathbf{Y}} = \max_{\mathbf{Y}}   & \ \text{tr}\left(\sum_{t=1}^{w} \bA_{t} \mathbf{Y}\right) \  \text{s.t.} \   \mathbf{Y} \succcurlyeq \mathbf{0}, Y_{i, i} = 1, \forall i \in [n].   
\end{aligned}
\]
The Lagrangian function is similarly written as 
\begin{align*}
    & \mathcal{L}(\bA_1,\ldots,\bA_w, \mathbf{Y}, \mathbf{S}, \mathbf{D}) = \operatorname{tr}\left(\sum_{t=1}^{w} \bA_{t} \mathbf{Y}\right) + \operatorname{tr}(\mathbf{S} \mathbf{Y}) \nonumber \\ 
    & \hspace{120pt} - \operatorname{tr}(\mathbf{D} (\mathbf{Y}-\mathbf{I})),
\end{align*}
where $\mathbf{S} \succcurlyeq \mathbf{0}$ and $\mathbf{D} = \operatorname{diag}(d_{i})$. Then, 
\begin{align*}
    \nabla_{\mathbf{Y}} \mathcal{L} = \sum_{t=1}^{w} \bA_{t} + \mathbf{S} - \mathbf{D} = \mathbf{0}. 
\end{align*}
In order to satisfy the first order stationary condition, we must satisfy the following (deterministic) condition 
\begin{align*}
    \mathbf{S}^{*} = \mathbf{D}^{*} - \sum_{t=1}^{w} \bA_{t}.  
\end{align*}
Similar to the single observation case, we have that when $\bS^\ast \succcurlyeq \mathbf{0}$, $\lambda_2(\bS^\ast)>0$, and $\bS^\ast \bm\sigma^\ast = 0$, 
then $\hat{\mathbf{Y}} = \mathbf{Y}^\ast=\bm\sigma^\ast(\bm\sigma^\ast)^T$ is the unique solution. We let 
\[
\bD^\ast = \text{diag}\{d_i^\ast\}, \quad d_i^\ast = \sum_{t=1}^w\sum_{j=1}^n \bA_{t}(i,j)\sigma_i^\ast\sigma_j^\ast.
\]
By definition we have $d_i^\ast \sigma_i^\ast = \sum_{t=1}^w\sum_{j=1}^n \bA_{t}(i,j)\sigma_j^\ast$, i.e., $\bD^\ast\bm\sigma^\ast =(\sum_{t=1}^{w} \bA_{t})\bm\sigma^\ast$, thus $\bS^\ast \bm\sigma^\ast = 0$.
It remains to show 
\begin{align*}
    \operatorname{Pr} \left[\inf_{\mathbf{x}: \|\mathbf{x}\| = 1, \mathbf{x} \perp \bm{\sigma}^{*}} \mathbf{x}^{T} \mathbf{S}^{*} \mathbf{x} > 0 \right] \geq 1 - o(1).
\end{align*}
\noindent Before we proceed, we note that 
\begin{align*}
    \mathds{E}[\sum_{t=1}^{w} \bA_{t}] & = w(1- 2 \zeta) p (\mathbf{Y}^{*} - \mathbf{I}).
\end{align*}

 Now, for any $\mathbf{x}$ such that $\|\mathbf{x}\| = 1, \mathbf{x} \perp \bm{\sigma}^{*}$ (i.e., $\mathbf{x}^{T} \bm{\sigma}^{*}$ = 0), we proceed similarly as \eqref{eq:single-condition},
\begin{align}
     & \mathbf{x}^{T} \mathbf{S}^{*} \mathbf{x}   = \mathbf{x}^{T} \mathbf{D}^{*} \mathbf{x} -  \mathbf{x}^{T} \mathds{E}[\sum_{t=1}^{w} \bA_{t}] \mathbf{x} \nonumber\\
     & \hspace{50pt} - \mathbf{x}^{T} (\sum_{t=1}^{w} \bA_{t} - \mathds{E}[\sum_{t=1}^{w} \bA_{t}]) \mathbf{x} \nonumber \\ 
      & \geq \min_{i \in [n]} d_{i}^{*} + w (1 - 2 \zeta) p -   \|\sum_{t=1}^{w} \bA_{t} - \mathds{E}[\sum_{t=1}^{w} \bA_{t}] \|. \label{eq:random}
      \end{align}
Note that each $d_{i}^\ast$ is equal in distribution to $\sum_{j=1}^{w(n-1)} X_{j}$, where $X_{j} \overset{i.i.d.}\sim p (1-\zeta) \delta_{+1} + p \zeta \delta_{-1} + (1-p) \delta_{0}$. By Lemma \ref{lemma_tail_bound_censored} and the union bound, we have with probability at least $1-n^{1-wa(\sqrt{1 - {\zeta}}  - \sqrt{{\zeta}})^{2}+o(1)}$, $\min_i d_i^\ast \geq \frac{\log n}{\log \log n}$.
Second, we have 
\[
\|\sum_{t=1}^{w} \bA_{t} - \mathds{E}[\sum_{t=1}^{w} \bA_{t}] \| \leq \sum_{t=1}^{w} \| \bA_{t} - \mathds{E}[\bA_{t}] \|, 
\]
by Lemma \ref{lem:lem-spectral-concentrate} we have $\operatorname{Pr} \{ \| \bA_{t} - \mathds{E}[\bA_{t}] \| \leq c'\sqrt{a\log n} \} \geq 1-n^c$ for $p=a\frac{\log n }{n}$, thus we have with probability at least $1-wn^{-c}$,
\[
\|\sum_{t=1}^{w} \bA_{t} - \mathds{E}[\sum_{t=1}^{w} \bA_{t}] \| \leq wc'\sqrt{a\log n}.
\]

Combing these together, with probability at least $1-n^{1-wa(\sqrt{1 - {\zeta}}  - \sqrt{{\zeta}})^{2}+o(1)} -wn^{-c}$, we have \eqref{eq:random} $\geq \frac{\log n}{\log \log n} + w(1-2\zeta)p - wc'\sqrt{a\log n}>0$ for $n$ sufficiently large and $w=o(\sqrt{\log n})$. And therefore the model satisfies the exact recovery condition if $wa(\sqrt{1 - {\zeta}}  - \sqrt{{\zeta}})^{2}>1$ and $w=o(\sqrt{\log n})$.

\end{remark}

\section{Proofs for Section \ref{sec:conditions}} \label{app:proofA}

\subsection{Proof of Theorem \ref{thm:private_threshold_condition_one_time_instance}}\label{app:proofIV1}
In this subsection, we derive the condition for exact recovery under the privacy mechanism.

\begin{lemma} \label{lemma_tail_bound_censored_with_graph_perturbation} (Chernoff Bound for the Perturbed Graph via Randomized Response) Let $\{X_{i}\}_{i=1}^{m}$ be a collection of independent random variables, where $m = n + o(n)$ and each $X_{i} \overset{i.i.d.} \sim \tilde{p}_{+1} \delta_{+1} + \tilde{p}_{-1}  \delta_{-1} + \tilde{p}_{0} \delta_{0}$. Then, 
\begin{align*}
\operatorname{Pr} \bigg(\sum_{i=1}^{m} X_{i} \leq  k_{n} \bigg) \leq \exp(- m \ell(k_{n}/m)),
\end{align*}
where $k_{n} = (1 + o(1)) \frac{\log(n)}{\log(\log(n))}$ and $\ell(k_{n}/m) = \sup_{t \geq 0} - tx - \log \mathds{E} \left[ e^{- t X_{1}} \right] $.
\end{lemma}
 The term $\mathds{E} \left[ e^{- t X_{1}} \right]$ is obtained as
\begin{align}
    \mathds{E} \left[ e^{- t X_{1}} \right] & = 1 + \tilde{p} \left[ e^{-t} (1 - \tilde{\zeta})  + e^{t} \tilde{\zeta} -1 \right]. \label{eqn:moment_generating_function}
\end{align}
Note that the above equation is a concave function of the parameter $t$. The optimum $t^{*}$ that maximizes \eqref{eqn:moment_generating_function} is given as follows:
\begin{align*}
    - x +  \frac{\tilde{p} \left[e^{-t^{*}} (1 - \tilde{\zeta}) - e^{t^{*}} \tilde{\zeta} \right]}{1 + \tilde{p} \left[ e^{-t^{*}} (1 - \tilde{\zeta}) + e^{t^{*}} \tilde{\zeta} -1 \right]} = 0.
\end{align*}

For $x = {k_{n}}/{m}$, the optimal parameter $t^{*}$ for the Chernoff bound is obtained as
\begin{align*}
    t^{*} & = \frac{1}{2} \log \frac{1 - \tilde{\zeta}}{\tilde{\zeta}} + o(1).
\end{align*}
Plugging the value of $t^{*}$, it gives us the following: 
\begin{align*}
   & \ell(k_{n}/m)  = - t^{*} \frac{k_{n}}{m} - \log     \mathds{E} \left[ e^{- t^{*} X_{1}} \right] \nonumber \\ 
        &  = - \frac{k_{n}}{2 m}  \log \frac{1 - \tilde{\zeta}}{\tilde{\zeta}} - \log \left[ 1 - \tilde{p} \bigg(\sqrt{1 - \tilde{\zeta}}  - \sqrt{\tilde{\zeta}}\bigg)^{2} \right] \\
        & \quad + o(k_{n}/m) \nonumber \\ 
 & \overset{(a)} = \tilde{p} \bigg(\sqrt{1 - \tilde{\zeta}}  - \sqrt{\tilde{\zeta}}\bigg)^{2} + o(\log(n)/n), \nonumber
\end{align*}
where step $(a)$ follows due to the Taylor expansion of $\log (1 - x)$ around $x = 0$, and $p = a \log(n)/n$ where the first term in step (a) behaves as $o(\log(n)/n)$. Substitute $\tilde p = \frac{2 + p(e^\epsilon-1)}{e^\epsilon+2}$ and $\tilde\zeta =  \frac{1+p\zeta(e^\epsilon-1)}{2 + p(e^\epsilon-1)}$, we obtain
\begin{align*}
   & \ell(k_{n}/m)  \geq \left[ \frac{e^{\epsilon} - 1}{e^{\epsilon} + 1} \times a +  \frac{2}{e^{\epsilon}+2} \times \frac{n}{\log(n)} \right] 
 \nonumber \\ 
 & \hspace{0.5in} \times \frac{p (e^{\epsilon}-1)}{2 + p (e^{\epsilon}-1)} \times \bigg(\sqrt{1-\zeta} - \sqrt{\zeta} \bigg)^{2} \log(n)/n \nonumber \\
 & \hspace{0.5in}  + o(\log(n)/n) \nonumber \\ 
    & \overset{(b)} \geq \left[ \frac{e^{\epsilon} - 1}{e^{\epsilon} + 1} \times a +  \frac{2}{e^{\epsilon}+2} \times \frac{n}{\log(n)} \right] 
 \nonumber \\ 
 & \hspace{0.2in} \times \alpha \times \bigg(\sqrt{1-\zeta} - \sqrt{\zeta} \bigg)^{2} \log(n)/n  + o(\log(n)/n), 
\end{align*}
where step (b) holds true whenever $a > \frac{2 (n^{3/2} - n)}{(n^{c} - 1) \log(n)}$, where $\alpha = 1 - \frac{1}{\sqrt{n}}$, $\epsilon = c \log(n)$ and $c > 0$. After applying the union bound, we get the following sufficient conditions for exact recovery is: 
\begin{align}
    & a \big(\sqrt{1-\zeta} - \sqrt{\zeta} \big)^{2} > \frac{e^{\epsilon} + 1}{e^{\epsilon} - 1} \nonumber \\
   & \hspace{0.3in} \times \bigg[ 1/\alpha - \frac{2}{e^{\epsilon} + 2} \times \frac{n}{\log(n)} \times \big(\sqrt{1-\zeta} - \sqrt{\zeta} \big)^{2}  \bigg].
\end{align}
From the above condition, we conclude that the privacy leakage $\epsilon$ should behave as $\Omega(\log(n))$. A more stringent condition yields 
\begin{align*}
    a \big(\sqrt{1-\zeta} - \sqrt{\zeta} \big)^{2} > \frac{1}{\alpha} \times \left(\frac{e^{\epsilon} + 1}{e^{\epsilon} - 1} \right), 
\end{align*}
thereby we prove \eqref{eq:recov_condi} in Theorem \ref{thm:private_threshold_condition_one_time_instance}.

\subsection{Proof of Theorem \ref{thm:converse}} \label{app:proofIV4}

The proof technique is inspired by the packing argument \cite{vadhan2017complexity} in which we provide a lower bound for estimation error under $\epsilon$-edge DP.

We establish a rigorous lower bound for all differentially private community recovery algorithms operating on graphs generated from the Censored Block Models (CBMs). Our methodology closely follows the frameworks outlined in \cite{vadhan2017complexity, chen2023private}, focusing on the notion of edge DP. Furthermore, we precisely define the classification error rate as
\begin{align}
  & {\operatorname{err}} \operatorname{rate} (\hat{\bm{\sigma}}(\mathbf{A}), {\bm{\sigma}^{*}}) \nonumber \\ 
  & \hspace{0.1in} =  \frac{1}{n} \times \min \{\operatorname{Ham}(\hat{\bm{\sigma}}(\mathbf{A}), {\bm{\sigma}^{*}}), \operatorname{Ham}(- \hat{\bm{\sigma}}(\mathbf{A}), {\bm{\sigma}^{*}}) \}.
\end{align}



Let us consider a series of pairwise disjoint sets \(\mathcal{S}_{i}, i\in[m]\). Each set \(\mathcal{S}_{i}\) contains vectors \(\bm{u} \in \{\pm 1\}^n\), where \(n\) is the vector dimension. A vector \(\bm{u}\) is included in \(\mathcal{S}_{i}\) if the error rate 
\(\operatorname{err}\operatorname{rate}(\bm{u}, \bm{\sigma}^{i})\) with a fixed vector \(\bm{\sigma}^{i}\) does not exceed the threshold \(\beta\). This is formally expressed as: 
\begin{align}
\mathcal{S}_{i} = \{ \bm{u} \in \{\pm 1\}^n : \operatorname{err}\operatorname{rate}(\bm{u}, \bm{\sigma}^{i} ) \leq \beta \}, \quad i = 1, 2, \ldots, m,
\end{align}
where $\mathcal{S}_{i}$'s are pairwise disjoint sets.

We next derive the necessary conditions for 
\begin{align}
    {\operatorname{Pr}}(\hat{\bm{\sigma}}(\mathbf{A}) \in \mathcal{S}_{i}) \geq 1 - \eta \label{eqn:utility_condition_recovery}
\end{align}
as a function of the CBM parameters and the privacy budget. Note that the randomness here is taken over the randomness of graph $G$ that is generated from $\operatorname{CBM} (\bm\sigma^{i},p,\zeta)$. Without loss of generality, let us consider a graph $\mathbf{A} \sim \operatorname{CBM} (\bm\sigma^{1},p,\zeta)$ that is generated from the ground truth labeling vector $\bm{\sigma}^{*} = \bm{\sigma}^{1}$. Further, for this case, we want to derive the necessary conditions for any $\epsilon$-edge DP recovery algorithms that
\begin{align}
        & {\operatorname{Pr}}(\hat{\bm{\sigma}}(\mathbf{A}) \in \mathcal{S}_{1}) \geq 1 - \eta 
        \Rightarrow  \sum_{i = 2}^{m}   {\operatorname{Pr}}(\hat{\bm{\sigma}}(\mathbf{A}) \in \mathcal{S}_{i}) \leq \eta, \label{eqn:equation_necessary_conditions}
\end{align}
where $\mathbf{A} \sim \operatorname{CBM} (\bm\sigma^{1},p,\zeta)$.

We next individually lower bound each term in Eqn. \eqref{eqn:equation_necessary_conditions}. To do so, we first invoke the group privacy property of DP \cite{dwork2014algorithmic} and show that for any two adjacency matrices $\mathbf{A}$ and $\mathbf{A}'$, we have
\begin{align}
    \operatorname{Pr}(\hat{\bm{\sigma}}(\mathbf{A}) \in \mathcal{S}) \leq e^{\epsilon \operatorname{Ham}(\mathbf{A}, \mathbf{A}')}     \operatorname{Pr}(\hat{\bm{\sigma}}(\mathbf{A'}) \in \mathcal{S}), 
\end{align}
for any measurable set $\mathcal{S} \subseteq \{\pm 1\}^{n}$. For each $i = 1,2, \cdots, m$, taking the expectation with respect to the coupling distribution $\Pi(\mathbf{A}, \mathbf{A}')$ between $\mathbf{A}$ and $\mathbf{A}'$ and setting $\mathcal{S} = \mathcal{S}_{i}$, yields the following:
\begin{align}
     & \mathds{E}_{\mathbf{A}, \mathbf{A}' \sim \Pi(\mathbf{A}, \mathbf{A}')} \left[ \operatorname{Pr}(\hat{\bm{\sigma}}(\mathbf{A}) \in \mathcal{S}_{i}) \right] \nonumber \\
     & \hspace{0.2in} \leq \mathds{E}_{\mathbf{A}, \mathbf{A}' \sim \Pi(\mathbf{A}, \mathbf{A}')} \left[ e^{\epsilon \operatorname{Ham}(\mathbf{A}, \mathbf{A}')}     \operatorname{Pr}(\hat{\bm{\sigma}}(\mathbf{A'}) \in \mathcal{S}_{i})\right] \nonumber \\
     & \Rightarrow  \operatorname{Pr}(\hat{\bm{\sigma}}(\mathbf{A}) \in \mathcal{S}_{i})  \nonumber \\
     & \hspace{0.2in} \leq \mathds{E}_{\mathbf{A}, \mathbf{A}' \sim \Pi(\mathbf{A}, \mathbf{A}')} \left[ e^{\epsilon \operatorname{Ham}(\mathbf{A}, \mathbf{A}')}     \operatorname{Pr}(\hat{\bm{\sigma}}(\mathbf{A'}) \in \mathcal{S}_{i})\right] \nonumber \\ 
 & \overset{(a)} \leq \left( \mathds{E}_{\mathbf{A}, \mathbf{A}' \sim \Pi(\mathbf{A}, \mathbf{A}')} \left[ e^{2 \epsilon \operatorname{Ham}(\mathbf{A}, \mathbf{A}')} \right] \right)^{1/2} \nonumber \\ 
 & \hspace{1in} \times \left(\mathds{E}_{\mathbf{A}, \mathbf{A}' \sim \Pi(\mathbf{A}, \mathbf{A}')} \left[ \operatorname{Pr}^{2}(\hat{\bm{\sigma}}(\mathbf{A'}) \in \mathcal{S}_{i}) \right] \right)^{1/2} \nonumber \\ 
  & \leq \left( \mathds{E}_{\mathbf{A}, \mathbf{A}' \sim \Pi(\mathbf{A}, \mathbf{A}')} \left[ e^{2 \epsilon \operatorname{Ham}(\mathbf{A}, \mathbf{A}')} \right] \right)^{1/2} \nonumber \\
  & \hspace{1in}\times \left( \operatorname{Pr}(\hat{\bm{\sigma}}(\mathbf{A'}) \in \mathcal{S}_{i}) \right)^{1/2}, \nonumber \\ 
    &\Rightarrow  1-\eta  \overset{(b)} \leq \left( \mathds{E}_{\mathbf{A}, \mathbf{A}' \sim \Pi(\mathbf{A}, \mathbf{A}')} \left[ e^{2 \epsilon \operatorname{Ham}(\mathbf{A}, \mathbf{A}')} \right] \right)^{1/2} \nonumber \\
    & \hspace{1in} \times \left( \operatorname{Pr}(\hat{\bm{\sigma}}(\mathbf{A'}) \in \mathcal{S}_{i}) \right)^{1/2}, \nonumber \\ 
     & \Rightarrow  (1-\eta)^{2}  \leq \left( \mathds{E}_{\mathbf{A}, \mathbf{A}' \sim \Pi(\mathbf{A}, \mathbf{A}')} \left[ e^{2 \epsilon \operatorname{Ham}(\mathbf{A}, \mathbf{A}')} \right] \right) \nonumber \\
     & \hspace{1in} \times  \operatorname{Pr}(\hat{\bm{\sigma}}(\mathbf{A'}) \in \mathcal{S}_{i}) , 
  \label{eqn:privacy_utility_condition}
\end{align}
where step (a) follows from applying Cauchy-Schwartz inequality. In step (b), we invoked the condition in Eqn. \eqref{eqn:utility_condition_recovery}.


We first note that the Hamming distance between two labeling vectors $\bm{\sigma}$ and $\bm{\sigma}'$ (each of size $n$) directly determines how many rows in the adjacency matrices $\mathbf{A}$ and $\mathbf{A}'$ are generated from the same versus different distributions. More precisely, we have two cases: 
case $(1)$: $\operatorname{Ham}(\bm{\sigma}, \bm{\sigma}')$ rows in $\mathbf{A}$ and $\mathbf{A}'$ have elements generated from different distributions, and case $(2)$: $n - \operatorname{Ham}(\bm{\sigma}, \bm{\sigma}')$ rows have elements from the same distribution. \\

\indent Case $(1)$: In this case, the probability the corresponding elements in the two matrices $\mathbf{A}$ and $\mathbf{A}'$ are different is $\bar{q} =1 - ((1-\zeta) \zeta p^{2} + (1-\zeta) \zeta p^{2} + (1-p)^{2})$. \\ 

\indent Case $(2)$:  In this case, the probability the corresponding elements in the two matrices $\mathbf{A}$ and $\mathbf{A}'$ are same is  $\bar{p} =1 - ((1-\zeta)^{2} p^{2} + \zeta^{2} p^{2} + (1-p)^{2})$.

We next focus in calculating the term $M_{\operatorname{Ham}(\mathbf{A}, \mathbf{A}')}(2 \epsilon) \triangleq \mathds{E}_{\mathbf{A}, \mathbf{A}' \sim \Pi(\mathbf{A}, \mathbf{A}')} \left[ e^{2 \epsilon \operatorname{Ham}(\mathbf{A}, \mathbf{A}')} \right] $ with the following set of steps: 
\begin{align}
    M_{\operatorname{Ham}(\mathbf{A}, \mathbf{A}')}(2 \epsilon) & = \left(M_{\operatorname{Ham}(\mathbf{A}, \mathbf{A}'): \text{same dist.}}(2 \epsilon) \right)^{(n - \operatorname{Ham}(\bm{\sigma}, \bm{\sigma}'))}  \nonumber \\ 
    & \hspace{0.1in} \times  \left(M_{\operatorname{Ham}(\mathbf{A}, \mathbf{A}'): \text{different dist.}}(2 \epsilon) \right)^{\operatorname{Ham}(\bm{\sigma}, \bm{\sigma}')}, \label{eqn:MGF_expression}
\end{align}
where, 
\begin{align}
    M_{\operatorname{Ham}(\mathbf{A}, \mathbf{A}'): \text{same dist.}}(2 \epsilon) & = e^{2 \epsilon} \bar{p} + (1 - \bar{p}), \\
    M_{\operatorname{Ham}(\mathbf{A}, \mathbf{A}'): \text{different dist.}}(2 \epsilon) & =  e^{2 \epsilon} \bar{q} + (1 - \bar{q}).
\end{align}
We then can readily show that,
\begin{align}
    & M_{\operatorname{Ham}(\mathbf{A}, \mathbf{A}')}(2 \epsilon)  \nonumber\\
    & \leq   \left(M_{\operatorname{Ham}(\mathbf{A}, \mathbf{A}'): \text{same dist.}}(2 \epsilon) \right)^{(n - \operatorname{Ham}(\bm{\sigma}, \bm{\sigma}')) \cdot \operatorname{Ham}(\bm{\sigma}, \bm{\sigma}')} \nonumber \\ 
    & = (1 - p' + p' e^{2\epsilon})^{(n - \operatorname{Ham}(\bm{\sigma}, \bm{\sigma}')) \cdot \operatorname{Ham}(\bm{\sigma}, \bm{\sigma}')}, 
    \label{eqn:MGF_expression_upper_bound}
\end{align}
where $p' = 2p^2\zeta(\zeta - 1) - (p - 1)^2 + 1$. Plugging \eqref{eqn:MGF_expression_upper_bound} in \eqref{eqn:privacy_utility_condition} yields the following:
\begin{align}
    \operatorname{Pr}(\hat{\bm{\sigma}}(\mathbf{A}') \in \mathcal{S}_{i}) & \geq \frac{(1-\eta)^{2}}{ M_{\operatorname{Ham}(\mathbf{A}, \mathbf{A}')}(2 \epsilon)}.
\end{align}



Finally, we lower bound the packing number $m$ with respect to $\operatorname{err} \operatorname{rate}$. 
Building upon the framework established in \cite{chen2023private}, we can readily demonstrate that
\begin{align}
    m \geq \frac{1}{2} \cdot \frac{| \mathcal{B}_{\operatorname{Ham}} (\bm{\sigma}^{*}, 4\beta n) |}{| \mathcal{B}_{\operatorname{Ham}} (\bm{\sigma}^{*}, 2 \beta n) |},
\end{align}
where $\mathcal{B}_{\operatorname{Ham}} (\bm{\sigma}^{*}, t \beta n) = \{\bm{\sigma} \in \{\pm 1\}^{n}: \operatorname{Ham}(\bm{\sigma}, \bm{\sigma}^{*}) \leq \beta \} $ and $| \mathcal{B}_{\operatorname{Ham}} (\bm{\sigma}^{*}, t \beta n) |$ is the number of vectors within the Hamming distance of $t \beta n$ from $\bm{\sigma}^{*}$ for $t > 0$. It includes all vectors that can be obtained by flipping any $t \beta n$ elements of $\bm{\sigma}^{*}$. Thus, we can further lower bound $m$ as 
\begin{align}
    m \geq \frac{1}{2} \cdot \frac{{n \choose 4 \beta n} }{{n \choose 2 \beta n}} \geq \frac{1}{2} \cdot \left(\frac{1}{8 e \beta} \right)^{2 \beta n}.
\end{align}
Recall that, we have
\begin{align}
   (m-1) \cdot \frac{(1-\eta)^{2}}{ M_{\operatorname{Ham}(\mathbf{A}, \mathbf{A}')}(2 \epsilon)} \leq \eta. \label{eqn:lower_bound_bad_events}
\end{align}
Taking the logarithm for both sides of \eqref{eqn:lower_bound_bad_events}, it yields
\begin{align}
   & 8 \beta n (n - 8 \beta n) \log(1 - p' + p' e^{2 \epsilon})  \geq 2 \beta n \log \left(\frac{1}{8 e \beta} \right) + \log \left( \frac{1}{\eta}\right) \nonumber \\ 
   \Rightarrow & \log(1 - p' + p' e^{2 \epsilon})  \geq \frac{  \log \left(\frac{1}{8 e \beta} \right)}{4  (n - 8 \beta n)} + \frac{\log \left( \frac{1}{\eta}\right)}{8 \beta n (n - 8 \beta n) } \nonumber \\ 
     \Rightarrow & \left(e^{2\epsilon} -1 \right) p' \geq \frac{  \log \left(\frac{1}{8 e \beta} \right)}{4  (n - 8 \beta n)} + \frac{\log \left( \frac{1}{\eta}\right)}{8 \beta n (n - 8 \beta n) } \nonumber \\ 
      \overset{(a)}\Rightarrow & \left(e^{2\epsilon} -1 \right) p' \geq  \frac{(\Delta + 1) \log(n) - \log(8 e)}{4 (n-8)} \nonumber \\ 
          \Rightarrow & \epsilon \geq   \frac{1}{2} \log \left[ 1 + \frac{(\Delta + 1) \log(n) - \log(8 e)}{ p'  (4n-32)} \right],
\end{align}
where $p' =   2 a \times \frac{\log(n)}{n} + a^{2} (2 \zeta^{2} - 2 \zeta -1) \times \frac{\log^{2}(n)}{n^{2}} = 2 a \frac{\log(n)}{n} + \mathcal{O} \left(\frac{\log^{2}(n)}{n^{2}}\right) $ . In step (a), we set $\beta = n^{-1}$ and $\eta = n^{- \Delta}$, where $\Delta > 0$.





\section{Proofs for Section \ref{sec:detection_theory}}\label{sec:detection_analysis}

We begin by considering the average run length, which characterizes the average false alarm period when there is no change in the community structure. For notational simplicity we denote $\ell(\bA; \bm{\sigma}, \bm{\sigma}'):=\log\frac{\Pr(\bA;\bm{\sigma}')}{\Pr(\bA;\bm{\sigma})}$ as the log-likelihood ratio between two CBMs $\CBM(\bm{\sigma}',p,\zeta)$ and $\CBM(\bm{\sigma},p,\zeta)$, and denote $\tilde\ell(\tilde\bA; \bm{\sigma}, \bm{\sigma}'):=\log\frac{\Pr(\tilde\bA;\bm{\sigma}')}{\Pr(\tilde \bA;\bm{\sigma})}$ as the log-likelihood ratio between two CBMs $\CBM(\bm{\sigma}',\tilde p,\tilde \zeta)$ and $\CBM(\bm{\sigma},\tilde p,\tilde \zeta)$ after graph perturbation.

\subsection{Proof of Lemma \ref{lem:arl}}\label{app:prooflemIV1}

    We compute the average false alarm period using similar ideas as in Lemma 8.2.1 in \cite{tartakovsky2014sequential}.
Let us define a Shiryaev-Roberts like statistic $\{R_t\}$ as
$$
R_t=(R_{t-1}+1)e^{\tilde\ell(\tilde\bA_t; \bm{\sigma}^{\text{pre}}, \bm{\hat\sigma}_{t-1})},  t >1, \quad ~R_1=0.
$$
It is important to note that $\{R_t-t\}$ is a martingale with respect to the $\bE_\infty$ measure since
\begin{align*}
\bE_\infty[R_t-t|\calF_{t-1}]& =\bE_\infty\left[(R_{t-1}+1)e^{\tilde\ell(\tilde\bA_t; \bm{\sigma}^{\text{pre}}, \bm{\hat\sigma}_{t-1})}-t\Big|\calF_{t-1}\right] \nonumber \\ 
& =R_{t-1}-(t-1).
\end{align*}
The last equality is true because given $\calF_{t-1}$ we have $\bm{\hat\sigma}_{t-1}$ fixed and is independent of the observation $\bA_t$. Since $\tilde\ell(\tilde\bA_t; \bm{\sigma}^{\text{pre}}, \bm{\hat\sigma}_{t-1})$ is a valid log density ratio for $\tilde\bA_t$, we have that its expectation under $\bE_\infty$ regime equals to one. 
The martingale property of $\{R_t-t\}$ and usage of Optional Sampling allows us to write for any stopping time $T$ with finite expectation that
\begin{equation}
\bE_\infty[R_T-T]=\bE_\infty[R_1-1]=-1~\Rightarrow~\bE_\infty[T]-1=\bE_\infty[R_T].
\label{eq:A1}
\end{equation}

Recall that the detection statistics updated in \eqref{eq:stat-wlcusum}, after exponentiation, can be equivalently written as
$$
e^{S_t}=\max\{e^{S_{t-1}},1\}e^{\tilde\ell(\tilde\bA_t; \bm{\sigma}^{\text{pre}}, \bm{\hat\sigma}_{t-1})},~e^{S_1}=1.
$$
Using induction, the fact that $e^{S_{2}}=R_{2}$ and that for $x\geq0$,  $x+1\geq\max\{x,1\}$, it is straightforward to prove that for $t>1$ we have $R_t\geq e^{S_t}$. 
With this observation and using \eqref{eq:A1} we can now write
$$
e^b\leq \bE_\infty[e^{S_{T_{\rm{L}}}}]\leq\bE_\infty[R_{T_{\rm{L}}}]=\bE_\infty[T_{\rm{L}}]-1\leq\bE_\infty[T_{\rm{L}}],
$$
which proves the desired inequality.

\subsection{Proof of Lemma \ref{lem:arl2}}\label{app:prooflemIV2}
Similar to the proof of Lemma \ref{lem:arl}, we again define the Shiryaev-Roberts like statistic $\{R_t\}$. $\{R_t-t\}$ is again a martingale, thus $\bE_\infty[T]-1=\bE_\infty[R_T]$ for the stopping time $T$.

Denote $L_t\overset{iid}{\sim}\rm{Lap}(\frac{4C}{\epsilon})$, $t\geq 1$. Then the detection statistics updated in \eqref{eq:stat-wlcusum2}, after exponentiation, can be equivalently written as
$$
e^{\tilde S_t}=e^{\mathcal S_t}e^{L_t} \leq  R_t e^{L_t}, t>1,
$$
where the inequality is due to $R_t\geq e^{\mathcal S_t}$ according to the proof of Lemma \ref{lem:arl}. 

Write the randomized threshold as $\tilde b = b + L_0$ for $L_0 \sim \rm{Lap}(\frac{2C}{\epsilon})$, then using \eqref{eq:A1} we can now write
$$
e^b\cdot \mathbb E[e^{L_0}]\leq \bE_\infty[e^{\tilde S_T}]\leq\bE_\infty[R_{T}]\bE_\infty[e^{L_T}]\leq \bE_\infty[T]\bE_\infty[e^{L_T}].
$$
Note that for $L\sim\rm{Lap}(0,\beta)$, we have $\mathbb E[e^L] = \frac{1}{1-\beta^2}$ when $\beta<1$. Thus we have when $\epsilon > 4C$, 
\[
\bE_\infty[T] \geq \frac{\mathbb E[e^{L_0}]}{\bE[e^{L_T}]} e^b  = \frac{1-(4C/\epsilon)^2}{1-(2C/\epsilon)^2} e^b,
\]
which proves the desired inequality.




\subsection{Proof of Theorem \ref{th:upper_wadd}}\label{app:proofIV5}

Then we provide theoretical guarantees on the detection delay of the procedure \eqref{eq:stat-wlcusum} and \eqref{eq:stat-wlcusum2}. Recall $I_0=\frac12 (\log\frac{1-\zeta}{\zeta}) p(1-2\zeta)({n\choose 2} - \sum_{i<j}\bm{\sigma}^\pre_i\bm{\sigma}^\pre_j\bm{\sigma}_i^\post\bm{\sigma}_j^\post)=\bE_1[\ell(\bA;\bm{\sigma}^\pre,\bm{\sigma}^\post)]$ is the KL divergence for original CMBs, and $\tilde I_0=\frac12 (\log\frac{1-\tilde\zeta}{\tilde\zeta}) \tilde p(1-2\tilde\zeta)({n\choose 2} - \sum_{i<j}\bm{\sigma}^\pre_i\bm{\sigma}^\pre_j\bm{\sigma}_i^\post\bm{\sigma}_j^\post)=\bE_1[\ell(\tilde \bA;\bm{\sigma}^\pre,\bm{\sigma}^\post)]$ is the KL information number for the distribution after graph perturbation. 
Here $\bE_1$ represents the expectation taken under the post-change regime (i.e., change-point $\nu=1$).


We first present the well-known information-theoretic lower bound on the detection delay for online change detection.

\begin{theorem}[Information-Theoretic Lower Bound on WADD \cite{tartakovsky2014sequential}]\label{thm:delay_lower} A lower bound on the WADD of any test that satisfies the average run length constraint $\mathbb{E}_\infty\left[\tau\right] \geq \gamma$ is
\[
\WADD(\tau) \geq \frac{\log\gamma}{I}(1+o(1)),
\]
as $\gamma\to\infty$, where $I$ is the KL divergence of the post- and pre-change data distributions.
\end{theorem}

Based on Theorem \ref{thm:delay_lower}, when choosing $b=\log\gamma$, we have $\mathbb E_\infty[T_{\rm L
}]\geq \gamma$, thus we have 
\begin{equation}\label{eq:wadd1_lower}
\WADD(T_{\rm L
})\geq \frac{\log\gamma}{\tilde I_0}(1+o(1))   
\end{equation}
since by definition $\tilde I_0$ is the KL divergence between the post- and pre-change graph distributions after perturbation. 
We then list the following Lemma as a direct consequence of Theorem 1 of \cite{wlcusum2023}, observing that $\log\gamma$ dominates the numerator term under our assumption.
\begin{lemma}[WADD Upper Bound \cite{wlcusum2023}]
Under the assumptions of Theorem \ref{th:upper_wadd}, we have the following upper bound for the worst-case performance of the detection procedure \eqref{eq:stop_time},
\begin{equation}
\WADD[T_{\rm L
}]\leq\frac{\log\gamma}{\hat{I}_0}(1+o(1)),
\label{eq:th1_upper}
\end{equation}
where $\hat{I}_0=\bE_1[\ell(\tilde \bA;\bm{\sigma}^\pre,\bm{\hat\sigma}_{t-1})]$.
\end{lemma}



In the following, we analyze the empirical quantities $\hat{I}_0$ and establish the asymptotic result. Following the proof of Theorem \ref{thm:private_threshold_condition_one_time_instance}, we have under the post-change regime, the estimate $\bm{\hat\sigma}_t$ resulted from the most recent adjacency matrix satisfies, as $n\to\infty$,
\[
\operatorname{Pr}(\bm{\hat\sigma}_t = \bm{\sigma}^{\text{post}}) \geq 1- n^{-\Omega(1)}. 
\]
Then we have
\[
\hat I_0 = \tilde I_0 \Pr\{\bm{\hat\sigma}_t = \bm{\sigma}^{\text{post}}\} + \bE[\hat I_0|\bm{\hat\sigma}_t \neq \bm{\sigma}^{\text{post}}]  \Pr\{\bm{\hat\sigma}_t \neq \bm{\sigma}^{\text{post}}\} 
\]
note that $\mathbb E[\hat I_0|\bm{\hat\sigma}_t] = \mathbb E[\log \frac{\Pr(\tilde\bA_t;\hat{\bm\sigma}_t)}{\Pr(\tilde\bA_t;\bm\sigma^\pre)}] = \mathbb E[\log \frac{\Pr(\tilde \bA_t;\hat{\bm\sigma}_t)}{\Pr(\tilde \bA_t;\bm\sigma^\post)} + \log \frac{\Pr(\tilde \bA_t;\bm\sigma^\post)}{\Pr(\tilde \bA_t;\bm\sigma^\pre)} ] \in [0, \tilde I_0]$.
Therefore, 
\begin{align}
\frac{\hat I_0}{\tilde I_0} &=\frac{\tilde I_0 \Pr\{\bm{\hat\sigma}_t = \bm{\sigma}^{\text{post}}\} + \bE[\hat I_0|\bm{\hat\sigma}_t \neq \bm{\sigma}^{\text{post}}]  \Pr\{\bm{\hat\sigma}_t \neq \bm{\sigma}^{\text{post}}\}}{\tilde I_0} \nonumber \\ 
& = 1- O(n^{-c}),
\end{align}
for some constant $c$, and when combined with \eqref{eq:th1_upper}, we obtain the desired asymptotic result in Theorem \ref{th:upper_wadd}.

\subsection{Proof of Theorem \ref{thm:wadd2}}\label{app:proofIV6}

Based on Theorem \ref{thm:delay_lower}, we have $\WADD(T_{\rm C
})\geq \frac{\log\gamma}{I_0}(1+o(1)) $. It remains to prove $\frac{\log\gamma}{I_0}(1+o(1))$ is also an upper bound for $\WADD(T_{\rm C})$. To show this, we modify the proof of Theorem 1 in \cite{wlcusum2023} to incorporate the additional Laplace noise added to the detection statistics.

Define the process $\{U_t,t=1,2,\ldots\}$ as $
U_t=U_{t-1}+\log\frac{\Pr(\bA_t;\bm{\hat\sigma}_{t-1})}{\Pr(\bA_t;\bm{\sigma}^{\text{pre}})}$ with $U_1=0$, and the corresponding stopping time $T'=\inf\{t>1:U_t +L_t\geq \tilde b\}$ for the threshold $\tilde b = b+L_0$. Obviously we have $U_t \leq \mathcal S_t$, thus $U_t+L_t \leq \mathcal S_t+L_t=\tilde S_t$ and $\mathbb E[T_{\rm C}]\leq \mathbb E[T']$. In the following, we derive an upper bound for $\mathbb E[T']$.

For simplicity let us further denote with $\ell_t=\log\frac{\Pr(\bA_t;\bm{\hat\sigma}_{t-1})}{\Pr(\bA_t;\bm{\sigma}^{\text{pre}})}$ and $\hat I_0 :=\bE_1[\ell_t]$, then we observe that for $t>1$ we can write
\begin{equation*}
\begin{aligned}
\bE[U_{T'}] & =\bE[\sum_{t=2}^{T'}\ell_t]  =\bE[\sum_{t=2}^{T'+1}\ell_t]
-\bE_0[\ell_{T'+1}] \\
& = \bE[\sum_{t=2}^{\infty}\ell_t\mathbb I\{T'\geq t-1\}
] -\bE_0[\ell_{T'+1}] \\
& =\bE[\sum_{t=2}^{\infty}\bE\left[\ell_t|\mathcal F_{t-2}\right]\mathbb I\{T'\geq t-1\}] -\bE_0[\ell_{T'+1}] \\
& = \hat I_0\bE[T'] - \bE_0[\ell_{T'+1}] \geq \hat I_0\bE[T'] - I_0. 
\end{aligned}
\end{equation*}
Then we conclude that
\begin{equation}
\begin{aligned}
&\bE[T']\leq\frac{\bE[U_{T'}]+I_0}{\hat{I}_0}=
\frac{\bE[U_{T'}+L_{T'}-\tilde b]+\bE[\tilde b - L_{T'}]+I_0}{\hat{I}_0}. 
\end{aligned}
\label{eq:AAA1}
\end{equation}

The next step involves the analysis of the expectation of the overshoot $U_{T'}+L_{T'}-\tilde b$. We borrow ideas from \cite{lorden1970excess} and modify the proof in \cite{wlcusum2023}. For any threshold $x>0$ define $T'_x$ to be the corresponding stopping time
$$
T'_x=\inf\{t>1:U_t+L_t \geq x\},
$$
and denote the overshoot function as $R_x=U_{T'_x} +L_{T'_x} -x$. Define a sequence of stopping times $\{\tau_j\}$ with $\tau_0=1$ and 
$$
\tau_j=\inf\left\{t>\tau_{j-1}:\sum_{s=\tau_{j-1}+1}^t\ell_s + L_t > L_{\tau_{j-1}}\right\},~j\geq1,
$$
and the corresponding ladder variables $z_j=\sum_{s=\tau_{j-1}+1}^{\tau_j}\ell_s + L_{\tau_j} - L_{\tau_{j-1}}>0$. We can now see that $U_{\tau_j} + L_{\tau_j}=\sum_{i=1}^jz_i$, in fact the sequence $\{U_t+L_t\}_{t\geq1}$ increases only at the stopping times $\{\tau_j\}$. For any given threshold $\nu$, in order to stop at $T_\nu'$ the statistic $U_t+L_t$ needs an increase at $T_\nu'$, which means that there exists a random index $j_\nu$ such that $\tau_{j_\nu}=T_\nu'$. Due to this fact, we can write $R_\nu=\sum_{j=1}^{j_\nu}z_j-\nu$. Following the same steps as in \cite{lorden1970excess}, we observe that $R_x$ is a piecewise linear function and all pieces having slope $-1$, thus we have
\begin{equation}
\begin{aligned}
\int_0^\nu R_x\,dx & =\int_0^{U_{T'_\nu}+L_{T'_\nu}} R_x \, dx - \int_\nu^{U_{T'_\nu}+L_{T'_\nu}} R_x \, dx \\
& =\frac{1}{2}\left\{\sum_{j=1}^{j_\nu}z_j^2-R_\nu^2\right\}. 
\end{aligned}
\label{eq:AAA2}
\end{equation}
By definition $0<z_j$ and $z_j=\sum_{t=\tau_{j-1}+1}^{\tau_j-1}\ell_t+\ell_{\tau_j} + L_{\tau_j} - L_{\tau_{j-1}}$ with $\sum_{t=\tau_{j-1}+1}^{\tau_j-1}\ell_t + L_{\tau_j-1} - L_{\tau_{j-1}}\leq0$, then we have $0<z_j\leq \ell_{\tau_j} + L_{\tau_j}  - L_{\tau_j-1}$, which implies $z_j^2\leq 2 \ell_{\tau_j}^2 + 2 (L_{\tau_j}  - L_{\tau_j-1})^2 \leq 2 \sum_{t=\tau_{j-1}+1}^{\tau_j} \left[\ell_t^2 + (L_{t}  - L_{t-1})^2 \right]$. Substituting in \eqref{eq:AAA2} yields
\[
\int_0^\nu R_x\,dx\leq \sum_{t=2}^{T'+1}[\ell_t^2 + (L_{t}  - L_{t-1})^2]- \frac12R_\nu^2.    
\]
If we take the expectation of above expression and use Jensen's inequality on the last term we obtain
\begin{align*}
&0\leq \int_0^\nu \bE[R_x]\,dx
\leq  \bE[\sum_{t=2}^{T'+1}[\ell_t^2 + (L_{t}  - L_{t-1})^2]]-\frac12(\bE[R_\nu])^2
 \\
 & =\frac{\hat{J}_0}{\hat{I}_0}\bE[\sum_{t=2}^{T'+1}\ell_t]-\frac12(\bE[R_\nu])^2 \\
& =\frac{\hat{J}_0}{\hat{I}_0}\bE\left[U_{T'}+\ell_{T'+1}\right]-\frac12(\bE[R_\nu])^2 \\
& =\frac{\hat{J}_0}{\hat{I}_0}\bE\left[R_\nu+\nu-L_{T'}+\ell_{T'+1}\right]-\frac12(\bE[R_\nu])^2 \\
& \leq \frac{\hat{J}_0}{\hat{I}_0}\left\{\bE[R_\nu]+\nu+I_0\right\}-\frac12(\bE[R_\nu])^2,
\end{align*}
where $\hat J_0 = \mathbb E[\ell_2^2 + (L_{2}  - L_{1})^2] = \mathbb E[\ell_2^2] + 64(\frac{C}{\epsilon})^2$; the first equality is true because $\bE[\sum_{t=2}^{T'+1}[\ell_t^2 + (L_{t}  - L_{t-1})^2]]=\hat J_0\bE[T']$ and
$\bE[\sum_{t=2}^{T'+1}\ell_t]=\hat{I}_0\bE[T']$. From the nonnegativity of the integral we have
$\frac12(\bE[R_\nu])^2\leq\frac{\hat{J}_0}{\hat{I}_0}\left\{\bE[R_\nu]+\nu+I_0\right\}$, from which we conclude that $\bE[R_\nu]\leq 2\frac{\hat{J}_0}{\hat{I}_0}+\big(2\frac{\hat{J}_0}{\hat{I}_0}\nu\big)^{1/2}+\big(2\frac{\hat{J}_0}{\hat{I}_0}I_0\big)^{1/2}$. Substitute $\nu = \tilde b= b+L_0$ as a random threshold, we obtain 
\begin{align*}
\mathbb E[R_{\tilde b}] & \leq 2\frac{\hat{J}_0}{\hat{I}_0}+\bE[\big(2\frac{\hat{J}_0}{\hat{I}_0}\nu\big)^{1/2}]+\big(2\frac{\hat{J}_0}{\hat{I}_0}I_0\big)^{1/2} \nonumber \\ 
& \leq 2\frac{\hat{J}_0}{\hat{I}_0}+\big(2\frac{\hat{J}_0}{\hat{I}_0}\big)^{1/2}(\sqrt{b}+\frac{\sqrt{4C\pi/\epsilon}}{2})+\big(2\frac{\hat{J}_0}{\hat{I}_0}I_0\big)^{1/2}.
\end{align*}
Substitute above into \eqref{eq:AAA1}, and plug in $b=\log\gamma+\log\frac{1-(2C/\epsilon)^2}{1-(4C/\epsilon)^2}$, we obtain
\begin{align*}
\mathbb E[T'] & \leq \frac{ \log\gamma+2\frac{\hat{J}_0}{\hat{I}_0}+\big(2\frac{\hat{J}_0}{\hat{I}_0}\big)^{1/2}(\sqrt{b}+\frac{\sqrt{4C\pi/\epsilon}}{2})}{\hat I_0} \nonumber \\
& \hspace{0.1in} + \frac{\big(2\frac{\hat{J}_0}{\hat{I}_0}I_0\big)^{1/2} + \log\frac{1-(2C/\epsilon)^2}{1-(4C/\epsilon)^2}+I_0}{\hat I_0}.
\end{align*}
We note that $\log\frac{1-(2C/\epsilon)^2}{1-(4C/\epsilon)^2} = o(\log\gamma)$ and $(\frac{\hat{J}_0}{\hat{I}_0})^{1/2} \sqrt{4C\pi/\epsilon} = o(\log\gamma)$. Similar to the proof of Theorem \ref{th:upper_wadd}, we also have $\frac{\hat I_0}{I_0} =1- O(n^{-c'})$ for some constant $c'$ under exact recovery condition. Then the above upper bound can be written as $\mathbb E[T']\leq \frac{\log\gamma}{I_0}(1+o(1))$. This completes the proof.

\section{Proofs for Section \ref{sec:info-lower-bound}}\label{sec:minimax_private}

\subsection{Proof of Lemma \ref{lemma:KL_after_pertub}} \label{app:prooflemIV3}

Recall that we denote $\bA$ as the original adjacency matrix, and $\tilde \bA$ as the perturbed matrix yielded by a $\epsilon$ edge-LDP randomized mechanism. We denote $\tilde p_\infty$ as the distribution for the perturbed matrix $\tilde \bA$ in the pre-change regime, and let $\tilde p_0$ be the distribution in the post-change regime. We have
\[
\begin{aligned}
    & \KL(\tilde  p_0 || \tilde  p_\infty)
    = \mathbb E_{\tilde p_0} \bigg[\sum_{1\leq i<j \leq n} \log \frac{\tilde p_0(\tilde A_{ij})}{\tilde  p_\infty(\tilde A_{ij})} \bigg] \\
    & = \sum_{i<j: \ \sigma^\post_i\sigma^\post_j\neq \sigma^\pre_i\sigma^\pre_j} \sum_z  \tilde p_0(\tilde \bA_{i,j}=z) \log \frac{\tilde p_0(\tilde \bA_{i,j}=z)}{\tilde p_\infty(\tilde \bA_{i,j}=z)}. 
\end{aligned}
\]
Denote $q_1(z) = q(\tilde A_{ij}=z|A_{ij}=+1)$ as the probability that the privatized edge equals to $z$ when the raw edge is $+1$, and similarly we define $q_0(z)$ and $q_{-1}(z)$. By the privacy constraint we have $\frac{q_x(z)}{q_y(z)} \in [e^{-\epsilon}, e^\epsilon]$ for any $x,y,z\in\{-1,0,+1\}$. We first present the following Lemma.

\begin{lemma} We note that for pair $i<j$ such that \(\sigma^\post_i=\sigma^\post_j,\sigma^\pre_i\neq \sigma^\pre_j\), we have
\[
\tilde p_0(\tilde \bA_{i,j}=z) = q_1(z) p(1-\zeta) + q_{-1}(z) p\zeta + q_0(z)(1-p),
\]
and 
\[
\tilde p_\infty(\tilde \bA_{i,j}=z) = q_1(z) p\zeta + q_{-1}(z) p(1-\zeta) + q_0(z)(1-p).
\]
For the perturbed data distribution $\tilde p_\infty$ and $\tilde p_0$, we have for each $z\in\{+1,-1,0\}$,
    \[
   \left|\log \frac{\tilde p_0(\tilde \bA_{i,j}=z)}{\tilde p_\infty(\tilde \bA_{i,j}=z)} \right| \leq c_\epsilon (e^\epsilon -1 ) p(1-2\zeta),
    \]
    where $c_\epsilon=\min\{2,e^\epsilon\}$.
\end{lemma}
\begin{proof}
    The proof follows from Lemma 1 of \cite{duchi2018minimax}. We first have $|\log\frac{a}{b} |\leq \frac{|a-b|}{\min\{a,b\}}$ for $a,b\in\mathbb R_{+}$. Indeed, for any $x>0$, $\log (x)\leq x-1$. Set $x=a/b$ and $x=b/a$ yields 
    \[
    \log\frac{a}{b}\leq \frac{a}{b}-1=\frac{a-b}{b}, \quad \log\frac{b}{a}\leq \frac{b}{a}-1=\frac{b-a}{a}.
    \]
    Using the first inequality for $a>b$ and the second for $a<b$ yields $|\log\frac{a}{b} | \leq \frac{|a-b|}{\min\{a,b\}}$. 

    Let $a=\tilde p_0(\tilde \bA_{i,j}=z)$ and $b=\tilde p_\infty(\tilde \bA_{i,j}=z)$, we have
    \[
    \left|\log \frac{\tilde p_0(\tilde \bA_{i,j}=z)}{\tilde p_\infty(\tilde \bA_{i,j}=z)} \right| \leq \frac{|\tilde p_0(\tilde \bA_{i,j}=z)-\tilde p_\infty(\tilde \bA_{i,j}=z)|}{\min\{\tilde p_0(\tilde \bA_{i,j}=z),\tilde p_\infty(\tilde \bA_{i,j}=z) \}}.
    \]

We further have
\[
|\tilde p_0(\tilde \bA_{i,j}=z)-\tilde p_\infty(\tilde \bA_{i,j}=z)| = |q_1(z) -  q_{-1}(z)|\cdot |p(1-2\zeta)|. 
\]
For any $x\in\{-1,+1,0\}$ we have
\begin{align}
|q_1(z) -  q_{-1}(z)|& = |q_1(z) - q_x(z) + q_x(z) -  q_{-1}(z)| \nonumber \\
&\leq |q_1(z) - q_x(z)| + |q_x(z) -  q_{-1}(z)|.
\end{align}
This means we have
\[
|q_1(z) -  q_{-1}(z)| \leq 2\min_x q_x(z) \max_{x'} \left| \frac{q_{x'}(z)}{q_{x}(z)} -1 \right|.
\]
Similarly we have
\begin{align}
|q_1(z) -  q_{-1}(z)| & \leq q_{-1}(z) \left| \frac{q_{1}(z)}{q_{-1}(z)} -1 \right|  \nonumber \\
&\leq e^\epsilon \min_x q_x(z) \left| \frac{q_{1}(z)}{q_{-1}(z)} -1 \right|.
\end{align}
Combing these together, and using the result that $e^\epsilon-1\geq 1-e^{-\epsilon}$, we have
\[
|q_1(z) -  q_{-1}(z)| \leq c_\epsilon  \min_x q_x(z) (e^\epsilon-1).
\]
Therefore we have
\[
\begin{aligned}
 \left|\log \frac{\tilde p_0(\tilde \bA_{i,j}=z)}{\tilde p_\infty(\tilde \bA_{i,j}=z)} \right| 
 & \leq \frac{|\tilde p_0(\tilde \bA_{i,j}=z)-\tilde p_\infty(\tilde \bA_{i,j}=z)|}{\min\{\tilde p_0(\tilde \bA_{i,j}=z),\tilde p_\infty(\tilde \bA_{i,j}=z) \}} \\
 & \leq \frac{c_\epsilon  \min_x q_x(z) (e^\epsilon-1)\cdot |p(1-2\zeta)|}{\min_x q_x(z)}\\
 & =c_\epsilon(e^\epsilon-1)\cdot |p(1-2\zeta)|.    
\end{aligned}
\]
\end{proof}

 From the above lemma, we have
\[
\begin{aligned}
 & \sum_z  \tilde p_0(\tilde \bA_{i,j}=z) \log \frac{\tilde p_0(\tilde \bA_{i,j}=z)}{\tilde p_\infty(\tilde \bA_{i,j}=z)} \nonumber \\
 & \leq  \sum_z  (\tilde p_0(\tilde \bA_{i,j}=z) - \tilde p_\infty(\tilde \bA_{i,j}=z)) \log \frac{\tilde p_0(\tilde \bA_{i,j}=z)}{\tilde p_\infty(\tilde \bA_{i,j}=z)} \\
 & \leq c_\epsilon^2 (e^\epsilon -1 )^2 p^2(1-2\zeta)^2 \sum_z \min_x q_x(z) \\
 & \leq c_\epsilon^2 (e^\epsilon -1 )^2 p^2(1-2\zeta)^2. 
\end{aligned}
\]
The same inequality hold for pair $i<j$ such that \(\sigma^\post_i\neq\sigma^\post_j,\sigma^\pre_i= \sigma^\pre_j\). Therefore we have
\[
\begin{aligned}
     &\KL(\tilde p_0 || \tilde  p_\infty ) \nonumber \\
     & \hspace{0.1in} \leq  c_\epsilon^2 (e^\epsilon -1 )^2 p^2(1-2\zeta)^2 \left[{n\choose 2} - \sum_{i<j}\bm{\sigma}_i^{\text{pre}}\bm{\sigma}_j^{\text{pre}}\bm{\sigma}_i^{\text{post}}\bm{\sigma}_j^{\text{post}}\right]. 
\end{aligned}
\]

\subsection{Proof of Lemma \ref{lem:private_hypo_test}} \label{app:prooflemIV4}

For any two adjacent graphs $\bA,\bA'$, an $(\epsilon, \delta)$-edge CDP test $\mathcal T$ should satisfies
\begin{align*}
    {\operatorname{Pr}(\mathcal T(\mathbf{A}') = 1)} \leq e^{\epsilon} { \operatorname{Pr}(\mathcal T({\mathbf{A}}) = 1)} + \delta,
\end{align*}
and ${\operatorname{Pr}(\mathcal T(\mathbf{A}') = 0)} \leq e^{\epsilon} { \operatorname{Pr}(\mathcal T({\mathbf{A}}) = 0)} + \delta$. 
This is equivalent to
\begin{align*}
    & { \operatorname{Pr}(\mathcal T({\mathbf{A}}) = 1)}  \geq e^{- \epsilon} \left[ {\operatorname{Pr}(\mathcal T(\mathbf{A}') = 1)} - \delta  \right] \nonumber  \\
     & \hspace{0.1in} = e^{- \epsilon} \left[ {\operatorname{Pr}(\mathcal T(\mathbf{A}') = 1)} + \frac{\delta}{e^{\epsilon} - 1} \right] - \frac{\delta}{e^{\epsilon} - 1}. 
\end{align*}
Generally, for any graph $\mathbf{A}$, using the group privacy property of CDP \cite{dwork2014algorithmic}, we have 
\begin{align*}
     & P_{\max} \geq { \operatorname{Pr}(\mathcal T({\mathbf{A}}) = 1)} \geq  e^{- R \epsilon} \left[ P_{\max} + \frac{\delta}{e^{\epsilon} - 1} \right] - \frac{\delta}{e^{\epsilon} - 1},
\end{align*}
where $P_{\max} \triangleq \sup_{\mathbf{A}} { \operatorname{Pr}(\mathcal T({\mathbf{A}}) = 1)} $. Similarly, we can show that 
\begin{align*}
     & 1- P_{\max}  \leq { \operatorname{Pr}(\mathcal T({\mathbf{A}}) = 0)},  \nonumber \\
     &{ \operatorname{Pr}(\mathcal T({\mathbf{A}}) = 0)}  \leq  e^{R \epsilon} \left[ (1 - P_{\max}) + \frac{\delta}{e^{\epsilon} - 1} \right] - \frac{\delta}{e^{\epsilon} - 1}.
\end{align*}
Combining both bounds, we have the following: 
\begin{align}
     & P_{\max} \geq { \operatorname{Pr}(\mathcal T({\mathbf{A}}) = 1)}, \nonumber \\ 
     & { \operatorname{Pr}(\mathcal T({\mathbf{A}}) = 1)} \geq \max \bigg\{ e^{- R \epsilon} \left[ P_{\max} + \frac{\delta}{e^{\epsilon} - 1} \right] - \frac{\delta}{e^{\epsilon} - 1}, \nonumber \\ 
     & \hspace{0.6in}  1 -  e^{R \epsilon} \left[ (1 - P_{\max}) + \frac{\delta}{e^{\epsilon} - 1} \right] - \frac{\delta}{e^{\epsilon} - 1} \bigg\}. \label{eqn:bounds_T_equal_1}
\end{align}
Solving for $P_{\max}$ to obtain the tightest upper bound, we get 
\begin{align*}
    P_{\max}^{*} = \frac{e^{R \epsilon}}{e^{R \epsilon} + 1} +  \left(1 + \frac{\delta}{e^{\epsilon} - 1}\right) \times \left(\frac{e^{R \epsilon} -1}{e^{R \epsilon} + 1}\right).
\end{align*}
Plugging the value of $P_{\max}^{*}$ into the right hand side of equation \eqref{eqn:bounds_T_equal_1}. Using the above lower bound on ${ \operatorname{Pr}(\mathcal T({\mathbf{A}}) = 1)}$, we get the following:
\begin{align}\label{eq:hypo-test-diff-cdp}
     & { \operatorname{Pr}(\mathcal T({\mathbf{A}}) = 1 | H_{1})} -  { \operatorname{Pr}(\mathcal T({\mathbf{A}}) = 1 | H_{0})} \nonumber \\
     & \overset{(a)} \leq \left( \frac{e^{R\epsilon} - 1}{e^{R\epsilon} + 1} \right) \times \left(1 + \frac{2 \delta}{e^{\epsilon} - 1}\right) \times \TV (p_{0} || p_{\infty} ) \nonumber \\ 
     & \overset{(b)} \leq \frac{1}{\sqrt{2}} \times \left( \frac{e^{R\epsilon} - 1}{e^{R\epsilon} + 1} \right) \times \left(1 + \frac{2 \delta}{e^{\epsilon} - 1}\right) \times \sqrt{\KL (p_{0} || p_{\infty} )},
\end{align}
where in step (a), $R = 2^{{n \choose 2}}$ which is the total number of undirected graphs. Step (b) follows from applying Pinsker's inequality. Here recall that $p_\infty(\bA)$ and $p_0(\bA)$ denote the probability distribution for the adjacency matrix $\bA$ in the pre-change and post-change regime, respectively.

\subsection{Proof of Proposition \ref{prop-delay-lower-cdp}}\label{app:proofpropIV1}

The proof follows the analysis in \cite{lai1998information}. Firstly, given that the ARL constraint $\mathbb E_\infty[T]\geq \gamma$ is satisfied, let $m$ be a positive integer less than $\gamma$, then for some $\nu \geq 1$, we have $
\Pr_\infty (T \geq \nu)>0$ and $\Pr_\infty(T<\nu+m \mid T \geq \nu) \leq m / \gamma$. 

Secondly, let $m$ be the largest integer $\leq(\log \gamma)^{2}$. Suppose $\mathbb E_{\infty}(T) \geq \gamma$. Then we choose a specific $\nu$ such that $
\Pr_\infty (T \geq \nu)>0$ and $\Pr_\infty(T<\nu+m \mid T \geq \nu) \leq m / \gamma$. 

Let $I=\frac{1}{\alpha_0} \left( \frac{e^{R\epsilon} - 1}{e^{R\epsilon} + 1} \right)^2 \times \left(1 + \frac{2 \delta}{e^{\epsilon} - 1}\right)^2 \times \KL (p_{0} || p_{1} )$, and $Z_i:=\log\frac{\Pr(O_i|H_1)}{\Pr(O_i|H_0)}$. Define the event $C_{\delta}=\{0 \leq T-\nu<(1-\delta) I^{-1} \log \gamma, \sum_{i=\nu}^{T} Z_{i}<\left(1-\delta^{2}\right) \log \gamma\}$, we first show that 
\[
{\Pr}_{(\nu)}\left\{C_{\delta} \mid T \geq \nu\right\} \rightarrow 0, \text{ as } \gamma \rightarrow \infty,
\]
for the chosen $\nu$ and every $0<\delta<1$. Indeed, we have
\begin{align}
{\Pr}_{(\nu)}\left(C_{\delta}\right) & =\int_{C_{\delta}} \exp \left(\sum_{i=\nu}^{T} Z_{i}\right) d P_{\infty} \nonumber \\
& \leq \exp \left\{\left(1-\delta^{2}\right) \log \gamma\right\} {\Pr}_{\infty}\left(C_{\delta}\right).
\end{align}
Then it follows that for {\it large} $\gamma$,
\begin{align*}
& {\Pr}_{(\nu)}\left\{C_{\delta} \mid T \geq \nu\right\} \leq \exp \left\{\left(1-\delta^{2}\right) \log \gamma\right\} \\
&\hspace{60pt} \cdot {\Pr}_{\infty} \left\{T-\nu<(1-\delta) I^{-1} \log \gamma \mid T \geq \nu\right\} \\
& \leq  \gamma^{1-\delta^{2}} {\Pr}_{\infty} (T<\nu+m \mid T \geq \nu) \leq  \gamma^{1-\delta^{2}}(\log \gamma)^{2} / \gamma \to 0,
\end{align*}
where the second inequality is because $m$ is the largest integer $\leq(\log \gamma)^{2}$. Hence, we have shown that as $\gamma \rightarrow \infty$, ${\Pr}_{(\nu)}\left\{C_{\delta} \mid T \geq \nu\right\} \rightarrow 0$.

Thirdly, we show that $\Pr_{(\nu)}\left\{\tilde C_{\delta} \mid T \geq \nu\right\} \rightarrow 0$ as well, where 
$\tilde C_{\delta}=\{0 \leq T-\nu<(1-\delta) I^{-1} \log \gamma, \sum_{i=\nu}^{T} Z_{i}\geq\left(1-\delta^{2}\right) \log \gamma\}$. To show this, we first observe that by the inverse Pinsker inequality,
\[
\mathbb E_0[Z_i] \leq \frac{2}{\alpha_0} ({ \operatorname{Pr}(\mathcal T({\mathbf{A}}) = 1 | H_{1})} -  { \operatorname{Pr}(\mathcal T({\mathbf{A}}) = 1 | H_{0})})^2\leq I,  
\]
where $\alpha_0:=\operatorname{Pr}(\mathcal T({\mathbf{A}}) = 1 | H_{0})>0$. Therefore, by the law of large numbers, we have for any $\delta>0$, $
\lim_{n\to\infty} \sup_{\nu\geq 1} \esssup P^{(\nu)} \{\max_{t\leq n} \sum_{i=\nu}^{\nu+t} Z_i\geq I(1+\delta)n|O_1,\ldots,O_{\nu-1}\}=0$. Substitute $n=(1-\delta) I^{-1} \log \gamma\to\infty$, we have ${\Pr}_{(\nu)} \left\{\tilde C_{\delta} \mid T \geq \nu\right\}\rightarrow0$. 


Combining this with $\Pr_{(\nu)}\left\{C_{\delta} \mid T \geq \nu\right\} \rightarrow 0$, we have 
$$
{\Pr}_{(\nu)}\left\{T-\nu \geq(1-\delta) I^{-1} \log \gamma \mid T \geq \nu\right\} \rightarrow 1, 
$$
which yields
$$
\mathbb E_\nu(T-\nu \mid T \geq \nu) \geq(1-\delta+o(1)) I^{-1} \log \gamma,
$$
as $\gamma \rightarrow \infty$. Since $\delta$ is arbitrary, it then follows that
\begin{equation*}
\WADD \geq\left(I^{-1}+o(1)\right) \log \gamma.
\end{equation*} 
Recall that $I=\frac{1}{\alpha_0} \left( \frac{e^{R\epsilon} - 1}{e^{R\epsilon} + 1} \right)^2 \times \left(1 + \frac{2 \delta}{e^{\epsilon} - 1}\right)^2 \times \KL (p_{0} || p_{\infty} )$, thereby we complete the proof.



\subsection{Proof of Theorem \ref{lemma:required_window_with_privacy}} \label{app:proofIV7}


 \paragraph{Step $1$: Lower Bounding the Mutual Information:}

Let us first define the random variable
\begin{align}
I \triangleq 
    \begin{cases}
        1, & \text{if} \hspace{0.05in} \mathcal{E} \hspace{0.05in} \text{occurs}, \\
        0, & \text{otherwise},
    \end{cases}
\end{align}
where the event $\mathcal{E}$ is defined as $\mathcal{E} \triangleq \big\{ {\operatorname{err}} (\hat{\bm{\sigma}}_{t}, {\bm{\sigma}}_{t}^{*}) = 0 \big\}$. Now, we have the following set of steps:
\begin{align}
    & H(\bm{\sigma}_{t}^{*} | \hat{\bm{\sigma}}_{t}) \leq H(I, \bm{\sigma}_{t}^{*}| \hat{\bm{\sigma}}_{t}) \nonumber \\ 
    & = H(I | \hat{\bm{\sigma}}_{t}) +  H(\bm{\sigma}_{t}^{*}|\hat{\bm{\sigma}}_{t}, I) \nonumber \\ 
    & \leq H(I) + H(\bm{\sigma}_{t}^{*}|\hat{\bm{\sigma}}_{t}, I = 0) \times \operatorname{Pr}(I = 0) \nonumber \\ 
    & \hspace{0.3in} +  H(\bm{\sigma}_{t}^{*}|\hat{\bm{\sigma}}_{t}, I = 1) \times \operatorname{Pr}(I = 1) \nonumber \\ 
    & \leq 1 + n \times \operatorname{Pr}(I = 0) +  H(\bm{\sigma}_{t}^{*}|\hat{\bm{\sigma}}_{t}, I = 1) \times \operatorname{Pr}(I = 1) \nonumber \\ 
    & \overset{(a)} = 1 + n - \operatorname{Pr}(I = 1) \times (n -  H(\bm{\sigma}_{t}^{*}|\hat{\bm{\sigma}}_{t}, I = 1) )  \nonumber \\ 
    & = 1 + n - \operatorname{Pr}(\hat{\bm{\sigma}}_{t} = {\bm{\sigma}}_{t}^{*}) \times  (n -  H(\bm{\sigma}_{t}^{*}|\hat{\bm{\sigma}}_{t}, I = 1) ), \label{upper_bound_cond_entropy}
\end{align}
where in step (a), we used fact that $H(\bm{\sigma}_{t}^{*}|\hat{\bm{\sigma}}_{t}, I = 1) \leq  n$. To this end, we have
\begin{align}
    I(\bm{\sigma}_{t}^{*}; \hat{\bm{\sigma}}_{t}) & = H(\bm{\sigma}_{t}^{*}) - H(\bm{\sigma}_{t}^{*} | \hat{\bm{\sigma}}_{t}) \nonumber \\ 
    & \overset{(a)}  = n -  H(\bm{\sigma}_{t}^{*} | \hat{\bm{\sigma}}_{t}) \nonumber \\ 
    & \overset{(b)} \geq  \operatorname{Pr}(\hat{\bm{\sigma}}_{t} = {\bm{\sigma}}_{t}^{*}) \times (n -  H(\bm{\sigma}_{t}^{*}|\hat{\bm{\sigma}}_{t}, I = 1) ) - 1  \nonumber \\
    & \overset{(c)} =  \operatorname{Pr}(\hat{\bm{\sigma}}_{t} = {\bm{\sigma}}_{t}^{*}) \times  n - 1,
\end{align}
where step $(a)$ follows that the labeling vector is generated uniformly at random, while step $(b)$ we invoked the upper bound on the conditional entropy $H(\bm{\sigma}_{t}^{*} | \hat{\bm{\sigma}}_{t}) $ in Eqn. \eqref{upper_bound_cond_entropy}. Finally, step (c) follows that $ H(\bm{\sigma}_{t}^{*}|\hat{\bm{\sigma}}_{t}, I = 1)  = 0$. 

\paragraph{Step 2: Upper Bounding the Mutual Information}

We upper bound the mutual information between $\bm{\sigma}_{t}^{*}$ and $\bm{\hat{\sigma}}_{t}$.  Let $\hat{\mathbf{Y}}_{t}$ be the labelling estimator via SDP (before the rounding procedure). Now, we have 
{
   \allowdisplaybreaks
    \begin{align}
     I(\bm{\sigma}_{t}^{*}; \hat{\bm{\sigma}_{t}}) & \overset{(a)} =    I(\bm{\sigma}_{t}^{*}; \hat{\mathbf{Y}}_{t}) \nonumber \\ 
    & \leq I(\bm{\sigma}_{t}^{*}; \hat{\mathbf{Y}}_{t}, \hat{\mathbf{Y}}_{t-1}, \cdots, \hat{\mathbf{Y}}_{t-(w-1)}) \nonumber \\ 
      &= H(\hat{\mathbf{Y}}_{t}, \hat{\mathbf{Y}}_{t-1}, \cdots, \hat{\mathbf{Y}}_{t-(w-1)}) \nonumber \\ 
      & \hspace{0.3in} - H(\hat{\mathbf{Y}}_{t}, \hat{\mathbf{Y}}_{t-1}, \cdots, \hat{\mathbf{Y}}_{t-(w-1)}| \bm{\sigma}_{t}^{*}), \label{eqn:break_1} 
      \end{align}
      where in step (a) holds true after a proper rounding. We next condition on $\bm{\sigma}_{t-1}^{*}, \cdots, \bm{\sigma}_{t-(w-1)}^{*}$ to further upper bound Eqn.  \eqref{eqn:break_1} as follows:
    \begin{align}
  \eqref{eqn:break_1} &\leq H(\hat{\mathbf{Y}}_{t}, \hat{\mathbf{Y}}_{t-1}, \cdots, \hat{\mathbf{Y}}_{t-(w-1)}) \nonumber \\ 
    & \hspace{0.2in}  - H(\hat{\mathbf{Y}}_{t}, \hat{\mathbf{Y}}_{t-1}, \cdots, \hat{\mathbf{Y}}_{t-(w-1)}| \bm{\sigma}_{t}^{*}, \bm{\sigma}_{t-1}^{*}, \cdots, \bm{\sigma}_{t-(w-1)}^{*}) \nonumber \\ 
    & \leq \sum_{i = 0}^{w-1}  H(\hat{\mathbf{Y}}_{t-i}) \nonumber\\
    & \hspace{0.2in} - H(\hat{\mathbf{Y}}_{t}, \hat{\mathbf{Y}}_{t-1}, \cdots, \hat{\mathbf{Y}}_{t-(w-1)}| \bm{\sigma}_{t}^{*}, \bm{\sigma}_{t-1}^{*}, \cdots, \bm{\sigma}_{t-(w-1)}^{*}) \nonumber \\ 
    & \overset{(b)} =  \sum_{i = 0}^{w-1}  H(\hat{\mathbf{Y}}_{t-i}) - \sum_{i = 0}^{w-1}  H(\hat{\mathbf{Y}}_{t-i} | \bm{\sigma}_{t-i}^{*}) \nonumber \\
    & = \sum_{i = 0}^{w-1}  I(\bm{\sigma}_{t-i}^{*}; \hat{\mathbf{Y}}_{t-i})  \overset{(c)} = \sum_{i = 0}^{w-1}  I(\bm{\sigma}^{*}_{t-i}; \hat{\bm{\sigma}}_{t-i}) \nonumber \\ 
     & \leq \sum_{i = 0}^{w-1}  \left[ n - H(\hat{\bm{\sigma}}_{t-i}| \bm{\sigma}_{t-i}^{*}) \right], \label{eqn:break_2}
    \end{align}
     where step (b) follows that the labeling estimate $\hat{\mathbf{Y}}_{t-i}$ only depends on the ground truth labelling vector $\bm{\sigma}^{*}$ through the construction of our proposed algorithm. Step (c) holds true after a proper rounding of the SDP-based estimate $\hat{\mathbf{Y}}_{t-i}$.
     Next, we condition on the realization of the adjacency matrix to step $t-1$ and invoke edge DP constraints in the following set of steps:
    \begin{align}
     \eqref{eqn:break_2}   & \leq \sum_{i = 0}^{w-1}  \left[ n - H(\hat{\bm{\sigma}}_{t-i}| \bm{\sigma}_{t-i}^{*}, {\mathbf{A}}_{t-i} = \mathbf{A}) \right] \nonumber \\ 
     &  \overset{(d)} \leq \sum_{i = 0}^{w-1}  \left[ n - e^{- \epsilon} H(\hat{\bm{\sigma}}_{t-i}|\bm{\sigma}_{t-i}^{*}, {\mathbf{A}}_{t-i} = \mathbf{A}') \right] \nonumber \\ 
    & \overset{(e)} = \sum_{i = 0}^{w-1}  \left[ n - n e^{- \epsilon} \cdot h_{2}\left(p_{e} \right) \right] \nonumber \\ 
    &  = \sum_{i = 0}^{w-1}  n \left[ 1 - e^{- \epsilon} \cdot h_{2}\left(p_{e}\right) \right],
\end{align}
}

\noindent where step (d), we can easily show under $\epsilon$-edge DP that $H(\hat{\bm{\sigma}}_{t-i}| \mathbf{A}_{t-i} = \mathbf{A}) \geq e^{- \epsilon} H(\hat{\bm{\sigma}}_{t-i} | \mathbf{A}_{t-i} = \mathbf{A}')$.  Step (e) follows that $\hat{\bm{\sigma}}_{t-i}$ is a sequence of Bernoulli random variables with parameter $p_{e} = \mathds{E}\left[ \operatorname{err} \operatorname{rate} (\hat{\bm{\sigma}}(\mathbf{A}), {\bm{\sigma}}^{*})  \right]$ (i.e., element-wise error probability), in addition, $p_{e} \rightarrow 0$ and $n \rightarrow \infty$, and $h_{2}(\cdot)$ is the binary entropy function. 

To this end, we have 
\begin{align}
    I(\bm{\sigma}_{t}^{*}; \hat{\bm{\sigma}_{t}}) &\leq w n \left[ 1 - e^{- \epsilon} \cdot h_{2}\left(p_{e} \right) \right].
\end{align}
Putting everything together gives us the following:
\begin{align}
     & w n \left[ 1 - e^{- \epsilon} \cdot  h_{2}\left(p_{e} \right) \right] \geq \operatorname{Pr}(\hat{\bm{\sigma}}_{t} = {\bm{\sigma}}_{t}^{*}) \times n -1 \nonumber \\ 
     & \Rightarrow w \geq \bigg(\operatorname{Pr}(\hat{\bm{\sigma}}_{t} = {\bm{\sigma}}_{t}^{*}) - \frac{1}{n}\bigg) \times \left[ 1 -  e^{- \epsilon} \cdot h_{2}\left(p_{e} \right) \right]^{-1}. 
\end{align}
A more stringent condition is
\begin{align}
    w \geq  \operatorname{Pr}(\hat{\bm{\sigma}}_{t} = {\bm{\sigma}}_{t}^{*})  \times \left[ 1 -  e^{- \epsilon}\right]^{-1} =  \operatorname{Pr}(\hat{\bm{\sigma}}_{t} = {\bm{\sigma}}_{t}^{*})  \times \frac{e^{\epsilon}}{e^{\epsilon}-1},
\end{align}
where we used the fact that $h_{2}\left( \cdot \right) \leq 1$.

\section{Additional Details of the Detection Algorithm}\label{app:unknown-para}

We provide additional details on how to perform detection when the pre-change community structure $\bm\sigma^{\pre}$ and the parameters $p$ and $\zeta$ are unknown, as briefly mentioned in Remark III.1. When the pre-change structure and parameters are unknown, we can estimate them using available data. If we have access to a separate set of historical data, we can add a pre-processing step to estimate the pre-change distribution before running the detection algorithm—this is shown in Algorithm~\ref{alg:alg-pre1} (for local DP) and Algorithm~\ref{alg:alg-pre2} (for central DP). For clarity, we present these algorithms assuming a single historical data. In the case of multiple historical samples, we can perform community label estimation on each data point individually and aggregate the results using a majority vote. Even when no separate historical data are available, the pre-change distribution can still be estimated by reserving the first few observations as ``historical'' samples, assuming the change-point does not occur too early. In both scenarios, if some effective historical data is used for estimation, the resulting estimate will be close to the true pre-change distribution. For this reason, we mainly assume the pre-change distribution is known, in order to simplify the analysis and focus on the core detection problem.

\begin{algorithm}[ht!]
   \caption{Pre-processing: Pre-change distribution estimation under LDP setting}
   \label{alg:alg-pre1}
\begin{algorithmic}[1]
   \STATE {\bfseries Input:} Historical data $\{\bA_{0}\}$, number of nodes $n$, private budget $\epsilon$.
   \STATE {\bfseries Output:} The differentially private estimate of pre-change community structure and parameters $\bm{\hat\sigma^\pre}$, $\hat p$, and $\hat \zeta$.
   \STATE Perturb historical data to $\{\tilde\bA_{0}\}$ according to the Privacy Mechanism (1) in Section III-A. 
   \STATE Calculate $\bm{\hat\sigma^\pre}$ via solving Eq. (3) on perturbed graphs:
   \[
   \hat{\mathbf{Y}} = \arg\max_{\mathbf{Y}}  \hspace{0.1in}\text{tr}( \tilde\bA_{0}\mathbf{Y})    \quad \text{s.t.}  \hspace{0.1in}  \mathbf{Y} \succcurlyeq \mathbf{0} , \, Y_{i, i} = 1, \forall i \in [n], 
   \]
   and $\hat{\mathbf{Y}}=\bm{\hat\sigma^\pre} (\bm{\hat\sigma^\pre})^T$.
   \STATE Calculate the maximum likelihood estimate $\hat p=E_0/{n \choose 2}$ and $\hat \zeta = \frac12 - \frac{(\bm{\hat\sigma^\pre})^{T}\tilde\bA_{0} \bm{\hat\sigma^\pre}}{4E_0}$, where $E_0$ denote the total number of edges in $\tilde\bA_{0}$. 
\end{algorithmic}
\end{algorithm}

\begin{algorithm}[ht!]
   \caption{Pre-processing: Pre-change distribution estimation under CDP setting}
   \label{alg:alg-pre2}
\begin{algorithmic}[1]
   \STATE {\bfseries Input:} Historical data $\{\bA_{0}\}$, private budget $\epsilon$.
   \STATE {\bfseries Output:} The differentially private estimate of pre-change community structure and parameters $\bm{\hat\sigma^\pre}$, $\hat p$, and $\hat \zeta$.
   \STATE Compute $d(\bA_{0})$, stability of $\bm{\hat\sigma}$ with respect to graphs $\{\bA_{0}\}$.
   \STATE $\tilde d = d(\bA_{0})+\operatorname{Lap}(0, 1/\epsilon)$.
   \IF{$\tilde{d}  > \frac{\log{1/\delta}}{\epsilon}$}
    \STATE $\bm{\hat\sigma^\pre} 
= \hat{\bm{\sigma}}(\bA_{0})$
    \ELSE
    \STATE $\bm{\hat\sigma^\pre} = \perp$ (a random label vector) 
    \ENDIF
    
   \STATE Compute $\hat p$ and $\hat \zeta$ as the maximum likelihood estimate from $\{\bA_{0}\}$. 
   \STATE $\hat p = \hat p + \operatorname{Lap}(0, 1/\epsilon)$, $\hat \zeta = \hat \zeta + \operatorname{Lap}(0, 1/\epsilon) $. 
\end{algorithmic}
\end{algorithm}

Moreover, in practice, the distributional parameters $p$ and $\zeta$ may also change after the change-point. The proposed detection algorithms can be easily extended to such cases by estimating the post-change parameters $p$ and $\zeta$ together with $\bm\sigma^{\post}$ and then substituting into the detection algorithm. Taking the detection under the LDP setting as an example (the CDP case can be handled similarly), in Algorithm~\ref{alg:alg1} of the main paper, we compute the maximum likelihood estimates $\bm{\hat{\sigma}}^{t-1}$, $\hat{p}$, and $\hat{\zeta}$ on the perturbed graph $\tilde{\bA}^{t-1}$, as described in Appendix~\ref{sec:llr}. The detection statistic is then updated as $S_{t}=(S_{t-1})^{+} + \log\frac{\Pr(\tilde\bA_t;\bm{\hat\sigma}_{t-1}, \hat p, \hat\zeta)}{\Pr(\tilde\bA_t;\bm{\sigma}^{\text{pre}}, p^\pre, \zeta^\pre)}$, where $p^\pre$ and $\zeta^\pre$ are the pre-change distributional parameters. Since the MLEs $\hat{p}$ and $\hat{\zeta}$ have closed-form expressions, the additional estimation incurs minimal computational cost. Moreover, as demonstrated in the real-data experiments in Fig. \ref{fig:real-air}, such adaptive estimation on $p$ and $\zeta$ may improve detection performance, especially when there is indeed a significant change in $p$ and $\zeta$.


\end{document}